\documentclass[twoside,table,11pt]{article}
\usepackage{jair, theapa, rawfonts}
\usepackage[english, vlined, ruled, linesnumbered]{algorithm2e}
\usepackage{enumerate,setspace}
\usepackage{amsmath,amssymb,amsfonts,ascmac,amsthm}
\usepackage{booktabs,microtype}
\usepackage{graphicx}
\usepackage{tikz}
\usetikzlibrary{arrows,shapes, calc, fit, positioning}
\usepackage{tkz-graph}
\usepackage[small]{caption}
\usepackage{subcaption}
\usepackage{enumitem}

\newcommand*{\citet}[1]{\shortciteA{#1}}
\newcommand*{\citep}[1]{\shortcite{#1}}
\newcommand{\bbR}{\mathbb{R}}

\newcommand{\calS}{\mathcal{S}}

\newcommand{\calP}{\mathcal{P}}

\newcommand{\calC}{\mathcal{C}}

\newcommand{\GASP}{{\sf GASP}}
\newcommand{\sfGASP}{{\sf GASP}}
\newcommand{\gGASP}{{\sf gGASP}}
\newcommand{\NS}{{\sf NS}}
\newcommand{\IS}{{\sf IS}}
\newcommand{\CR}{{\sf CR}}
\newcommand{\PS}{{\sf PS}}
\newcommand{\WS}{{\sf WS}}
\newcommand{\DS}{{\sf DS}}

\newcommand{\poly}{\mathrm{poly}}
\newcommand{\desc}{\mbox{\rm desc}}
\newcommand{\ch}{\mbox{\rm ch}}
\newcommand{\MCC}{\textsc{Multicolored Clique}}
\newcommand{\MIS}{\textsc{Multicolored Independent Set}}
\newcommand{\height}{\mbox{\rm height}}
\newcommand{\colour}{{\sf Color}}
\newcommand{\pair}{{\sf ColorPair}}
\newcommand{\Dummy}{{\sf Dummy}}
\makeatletter
\newtheorem*{rep@theorem}{\rep@title}
\newcommand{\newreptheorem}[2]{%
\newenvironment{rep#1}[1]{%
 \def\rep@title{#2 \ref{##1}}%
 \begin{rep@theorem}}%
 {\end{rep@theorem}}}
\makeatother

\theoremstyle{plain}
	  \newtheorem{theorem}{Theorem}
	  \newtheorem{cor}[theorem]{Corollary}
	  \newtheorem{lemma}[theorem]{Lemma}
	  \newtheorem{proposition}[theorem]{Proposition}
	  \newtheorem{observation}[theorem]{Observation}
	  \newtheorem{claim}[theorem]{Claim}
\theoremstyle{definition}
	  \newtheorem{define}[theorem]{Definition}
	  
	  \newtheorem{example}[theorem]{Example}
	  
\theoremstyle{remark}
	  
\ShortHeadings{Group Activity Selection on Social Networks}
{Igarashi, Bredereck, Peters, \& Elkind}
\firstpageno{1}

\begin{document}

\title{Group Activity Selection on Social Networks}

\author{\name Ayumi Igarashi \email ayumi.igarashi@cs.ox.ac.uk\\
       \name Robert Bredereck \email robert.bredereck@tu-berlin.de\\
       \name Dominik Peters \email dominik.peters@cs.ox.ac.uk\\
       \name Edith Elkind \email elkind@cs.ox.ac.uk\\
       \addr University of Oxford\\
       Wolfson Building, Parks Road, Oxford OX1 3QD
       %\AND
       }

% For research notes, remove the comment character in the line below.
% \researchnote

\maketitle

\begin{abstract}
We propose a new variant of the group activity selection problem (\sfGASP), where the agents are placed on a social 
network and activities can only be assigned to connected subgroups (\gGASP). We show that if multiple groups
can simultaneously engage in the same activity, finding a stable outcome is easy as long as the network
is acyclic. In contrast, if each activity can be assigned to a single group only, finding stable outcomes 
becomes computationally 
intractable, even if the underlying network is very simple:
the problem of determining whether 
a given instance of a \gGASP\ admits a Nash stable outcome turns out to be NP-hard 
when the social network is a path or a star, or if the size of each connected component is bounded by a constant.
%parameterized complexity
We then study the parameterized complexity of finding outcomes of \gGASP{} that are
Nash stable, individually stable or core stable.
For the parameter `number of activities', we propose an FPT algorithm for Nash stability 
for the case where the social network is acyclic and
obtain a W[1]-hardness result for cliques (i.e., for standard \GASP); similar results hold for individual stability.
In contrast, finding a core stable outcome is hard even if the number of activities 
is bounded by a small constant, both for standard \GASP{} and when the social network is a star.
For the parameter `number of players', all problems we consider are in XP for arbitrary social networks; 
on the other hand, we prove W[1]-hardness results with respect to the parameter `number of players'
for the case where the social network is a clique. \footnote{Preliminary versions of this article appeared in Proceedings of the 31st AAAI Conference on Artificial Intelligence \citep{Igarashi2016b}, and Proceedings of the 16th International Joint Conference on Autonomous Agents and Multi Agent Systems \citep{Igarashi2017}. In this longer version, we provide all proofs that were omitted from the conference version and improve some of the results.}
%Ayumi: The pointers for the conference versions are both Igarashi et al. (2017). Is there any way we can avoid this?
\end{abstract}

\section{Introduction}
Companies assign their employees to different departments, large decision-making bodies split their members into expert 
committees, and university faculty form research groups: division of labor, and thus group formation, is 
everywhere. For a given assignment of agents to activities (such as management, product development, or marketing) to be 
successful, two considerations are particularly important: the agents need to be capable to work on their activity, and 
they should be willing to cooperate with other members of their group.

Many relevant aspects of this setting are captured by the {\em group activity selection problem} (\sfGASP), 
introduced by \citet{Darmann2012}. In \sfGASP, players have preferences over pairs of the form 
(activity, group size).
The intuition behind this formulation is that certain tasks are best performed in small or large groups, 
and agents may differ in their preferences over group sizes; however, they are indifferent about
other group members' identities. In the analysis of \sfGASP, desirable outcomes 
correspond to \emph{stable} and/or \emph{optimal} assignments of players to activities, i.e., assignments 
that are resistant to player deviations and/or
that maximize the total welfare. 

The basic model of \sfGASP\ ignores the relationships among the agents: 
Do they know each other? Are their working styles and personalities compatible?
Typically, we cannot afford to ask each agent about her preferences over all pairs of the form 
(coalition, activity), as the number of possible coalitions grows quickly with the number of agents.
A more practical alternative is to adopt the ideas of \citet{Myerson1977} and assume
that the relationships among the agents are encoded by a \emph{social network},
i.e., an undirected graph where nodes correspond to players and edges represent 
communication links between them; one can then require that each group is connected
with respect to this graph.

In this paper we extend the basic model of \sfGASP\ to take into account the agents' social network (\gGASP). 
We formulate several notions of stability for this setting, including Nash stability and core stability,
and study the complexity of computing stable outcomes in our model. 
These notions of stability are inspired by the hedonic 
games literature \citep{Banerjee2001,Bogomolnaia2002}, and were applied in the 
\sfGASP\ setting by \citet{Darmann2012} and \citet{Darmann2015}.
 
Hedonic games on social networks were recently considered by \citet{Igarashi2016a},
who showed that if the underlying network is acyclic, stable outcomes are guaranteed to exist
and some of the problems known to be computationally hard for the unrestricted setting
become polynomial-time solvable. We obtain a similar result for \sfGASP, but only
if several groups of agents can simultaneously engage in the same activity, i.e., 
if the activities are {\em copyable}. In contrast, we show that if each activity 
can be assigned to at most one coalition, a stable outcome may fail to exist, and moreover finding them is computationally hard
even if the underlying network is very simple.
Specifically, checking the 
existence of Nash stable, individually stable, or core stable outcomes turns out to be NP-hard 
even for very restricted classes of graphs, including paths, stars, 
and graphs with constant-size connected components. This result stands 
in sharp contrast to the known computational results in the literature; indeed, in the context of of cooperative games, such restricted networks 
usually enable one to design efficient algorithms for computing stable solutions (see, e.g., 
\citet{Chalkiadakis2016}; \citet{Elkind2014}; 
\citet{Igarashi2016a}).

Given these hardness results, we switch to the fixed parameter tractability paradigm. A problem is said to be {\em fixed parameter tractable} (FPT) with respect to a parameter $k$ if each instance $I$ of this problem can be solved in time $f(k) \poly(|I|)$, and to be {\em slice-wise polynomial} (XP) with respect to a parameter $k$ if each instance $I$ of this problem can be solved in time $|I|^{f(k)}$, respectively; here $f$ is a computable function that depends only on $k$.
In the context of \sfGASP, a particularly relevant parameter is the number of activities: generally speaking, we expect the number of available activities to be small in many practical applications. For instance, companies can only assign a limited number of projects to their employers; a workshop can usually organise a couple of social events; and schools can offer few facilities to their students. 
We show that the problem of deciding the existence of Nash stable outcomes for \gGASP s on acyclic graphs is in FPT with respect to the number of activities. For general graphs, we obtain a W[1]-hardness result, 
implying that this problem is unlikely to admit an FPT algorithm.
This hardness result holds even for \gGASP s on cliques; thus, it is also W[1]-hard to decide the existence of a Nash stable outcome in a standard \GASP.

While we find that from an algorithmic point of view, 
individual stability is very similar to Nash stability,
%core stability
unfortunately, our FPT results do not extend to core stability: we prove that checking the existence of core stable 
assignments is NP-complete even for \gGASP s on stars with two activities; for standard \GASP, we can prove
that this problem is hard if there are at least four activities.
On the other hand, if there is only one activity, 
a core stable assignment always exists and can be constructed efficiently, for any network structure.

Another parameter we consider is the number of players. This restriction applies to many practical scenarios. For example, in research teams with limited human resources, there are a limited number of researchers who are able to conduct the projects. 
Somewhat surprisingly, we 
show that the parameterization by the number of players does not give rise to an FPT algorithm for \gGASP s on general networks. 
Specifically, for all stability notions we consider, 
it is W[1]-hard to decide the existence of a stable outcome even when the underlying graph is a clique. Again, our hardness result particularly implies the W[1]-hardness of computing stable outcomes in a standard \GASP.
%It is worth noting that in our proofs, we essentially show the hardness of determining whether a fixed coalition structure can be stabilized for some assignment. The computational intractability is thus due to the difficulty of assigning activities to coalitions when players have non-trivial preferences. 
We summarize our complexity results in Table~\ref{table}. 

\begin{table*}[t]
	\footnotesize
	\centering
	\begin{tabular}{lcccc}
		\toprule
		& general case & few activities ($p$)& few players ($n$) & copyable activities \\
		\midrule
		\multicolumn{3}{l}{Nash stability and individual stability} \\
		\midrule
		cliques & \cellcolor{gray!30} NP-c. & W[1]-h. (Th. \ref{thm:W1:activities:clique:NS}, \ref{thm:W1:activities:clique:IS})& W[1]-h. (Th. \ref{thm:W1players:clique:NSIS}) &\cellcolor{gray!30} NP-c. \\
		acyclic & NP-c. (Th. \ref{thm:NPhardness:path:NS}, \ref{thm:NPhardness:path:star:sc:CRIS}) & FPT (Th. \ref{thm:FPT:tree:NS},\ref{thm:FPT:tree:IS}) & XP (Obs. \ref{obs:XPplayers}) & P (Th. \ref{thm:NS:copyable}, \ref{thm:IS:copyable}) \\
		paths & NP-c. (Th. \ref{thm:NPhardness:path:NS}, \ref{thm:NPhardness:path:star:sc:CRIS}) & FPT (Th. \ref{thm:FPT:tree:NS},\ref{thm:FPT:tree:IS})& XP (Obs. \ref{obs:XPplayers}) & P (Th. \ref{thm:NS:copyable}, \ref{thm:IS:copyable}) \\
		stars & NP-c. (Th. \ref{thm:NPhardness:star:NS}, \ref{thm:NPhardness:path:star:sc:CRIS})& FPT (Th. \ref{thm:FPT:tree:NS},\ref{thm:FPT:tree:IS})& XP (Obs. \ref{obs:XPplayers}) & P (Th. \ref{thm:NS:copyable}, \ref{thm:IS:copyable}) \\
		small comp. & NP-c. (Th. \ref{thm:NPhardness:sc:NS}, \ref{thm:NPhardness:path:star:sc:CRIS})&FPT (Th. \ref{thm:FPT:smallcomponents:NSIS}) & XP & P (Th. \ref{thm:NS:copyable}, \ref{thm:IS:copyable}) \\
		\midrule
		
		\multicolumn{3}{l}{core stability} \\
		\midrule
		cliques &\cellcolor{gray!30}  NP-c. & NP-c. for $p=4$ (Th. \ref{thm:NP:clique:core}) & W[1]-h. (Th. \ref{thm:W1players:clique:core})&\cellcolor{gray!30}  NP-c. \\	
		acyclic & NP-c. (Th. \ref{thm:NPhardness:path:star:sc:CRIS})& NP-c. for $p=2$ (Th. \ref{thm:NP:star:core}) & XP (Obs. \ref{obs:XPplayers})&\cellcolor{gray!30} P\\
		paths & NP-c. (Th. \ref{thm:NPhardness:path:star:sc:CRIS})& XP (Prop. \ref{prop:core-xp}) & XP (Obs. \ref{obs:XPplayers})&\cellcolor{gray!30} P \\
		stars & NP-c. (Th. \ref{thm:NPhardness:path:star:sc:CRIS})& NP-c. for $p=2$ (Th. \ref{thm:NP:star:core})& XP (Obs. \ref{obs:XPplayers})&\cellcolor{gray!30} P\\
		small comp. & NP-c. (Th. \ref{thm:NPhardness:path:star:sc:CRIS})&FPT (Th. \ref{thm:FPT:smallcomponents:CR})& XP (Obs. \ref{obs:XPplayers})&\cellcolor{gray!30} P\\		
		\bottomrule
	\end{tabular}
\caption{Overview of our complexity results. %`NS' stands for Nash stability, `IS' stands for individual stability, `CR' stands for core stability.
All W[1]-hardness results are accompanied by XP-membership proofs.
For all `XP'-entries, the question whether the problem is fixed-parameter tractable remains open. Grey entries are results that have appeared in the literature before.}
\label{table}
\end{table*}

\subsection{Related work}
%GASP
\citet{Darmann2012} initiate the study of \GASP. In the work of \citet{Darmann2012}, players are assumed 
to have approval preferences, and a particular focus is placed on individually rational assignments 
with the maximum number of participants and Nash stable assignments. They obtained a number of complexity results of computing these outcomes while concerning special cases where players have increasing/degreasing preferences on the size of a group. Subsequently, \citet{Darmann2015}
investigated a model where players submit ranked ballots. In this work, stability concepts such as core stability and individual stability have been adapted from the hedonic games literature. Further, \citet{Lee2015} studied {\em stable invitation problems}, which is a subclass of \GASP\ where only one activity is assigned to players. This problem was inspired by settings in which an organizer of an event chooses a {\em stable} set of guests who have preferences over the number of participants of the event. We refer the reader to the recent survey by \citet{Darmann2017} for the relation between these models. 

Recently, the parameterized aspects of \GASP\ have been considered by several authors. \citet{Lee2017b} studied the complexity of standard \GASP, with parameter being the number of groups. They showed that computing a maximum individually rational assignment is in FPT with respect to that parameter; however, they proved that this does not extend to other solution concepts, such as Nash stability and envy-freeness, by obtaining a number of W[1]-hardness results. More recently, \citet{Gupta2017} investigated the parameterized complexity of finding Nash stable outcomes in the context of \gGASP. In their work, computation of a Nash stable outcome was shown to be in FPT with respect to the combined parameters: the maximum size of a group and the maximum degree in the underlying social network. They also presented an FPT algorithm with respect to the number of activities when the underlying network has a bounded tree-width. This generalizes our FPT result for trees and improves the bound on the running time. 

%Hedonic games
\GASP\  are closely related to hedonic games \cite{Banerjee2001,Bogomolnaia2002}. Much work has been devoted to the complexity study of hedonic games when there is no restriction on coalition formation (see e.g. \citet{Woeginger2013} and \citet{Aziz2016}).
%Ballester
In particular, the copyable setting of \GASP\ includes a class of anonymous hedonic games where players' preferences are only determined by the size of the coalition to which they belong. \citet{Ballester2004} showed that computing Nash, core, and individually stable outcomes of the game is NP-hard for anonymous games; this translates into the NP-hardness of these solutions for copyable instances of \gGASP\ when the social network is a clique. Nevertheless, our positive results for copyable activities imply that in anonymous hedonic games, one can compute a stable outcome in polynomial time if the underlying social network is acyclic. 
%DP: maybe cite my paper on hedonic games with dichotomous preferences 

It is worth mentioning that models with graph connectivity constraints have been studied in different settings from ours \citep{Talmon2017,Bouveret2017,Warut2017}. \citet{Talmon2017} considered the multiwinner problem when each winner has to represent a connected voting district. In the work of \citet{Talmon2017}, a similar hardness result concerning optimal committees for paths was obtained; further, computing an optimal committee on graphs with bounded tree-width was shown to be polynomial-time solvable for {\em non-unique} variants of the problem where several connected districts can be represented by the same winner. This restriction corresponds to our copyalbe cases of \gGASP.  
In the fair division literature, \citet{Bouveret2017} investigated a fair allocation of indivisible goods under graph connectivity constraints: the graph represents the dependency among the items, and each player's bundle must be connected in this graph. Similarly, computing envy-free and proportional allocations was proven to be NP-hard even when the graph among the items is a path; however, they showed that computing maximin fair allocations can be done in polynomial time when the graph is acyclic.

\section{Preliminaries}
\subsection{Group Activity Selection Problems}
For $s\in{\mathbb N}$, let $[s]=\{1,2,\ldots,s\}$. For $s,t\in{\mathbb Z}$ where $s \le t$, let $[s,t]=\{s,s+1,s+2,\ldots,t\}$.
An instance of the {\em Group Activity Selection Problem} (\sfGASP) is given by a finite set of {\em players} 
$N=[n]$, a finite set of {\em activities} $A=A^{*}\cup\{a_{\emptyset}\}$, where 
$A^{*}=\{a_{1},a_{2},\ldots,a_{p}\}$ and $a_{\emptyset}$ is the {\em void activity}, and a {\em profile} 
$(\succeq_{i})_{i \in N}$ of complete and transitive preference relations over the set of {\em alternatives} 
$(A^{*} \times [n])\cup \{(a_{\emptyset},1)\}$. Intuitively, $a_\emptyset$ corresponds to staying 
alone and doing nothing; multiple agents can make that choice independently from each other.

Throughout the paper, we assume that we can determine in unit time whether each player $i$ prefers $(a,k)$ to $(b,\ell)$, prefers $(b,\ell)$ to $(a,k)$, or is indifferent between them. We will write $x \succ_i y$ or $i: x\succ y$ to indicate that player
$i$ strictly prefers alternative $x$ to alternative $y$;
similarly, we will write $x\sim_i y$ or $i: x\sim y$ if $i$ is indifferent between $x$ and~$y$.
Also, given two sets of alternatives $X, Y$ and a player $i$, we write
$X\succ_i Y$ to indicate that $i$ is indifferent among all alternatives in $X$
as well as among all alternatives in $Y$, and prefers each alternative in $X$
to each alternative in $Y$.

We refer to subsets $S\subseteq N$ of players as {\em coalitions}. 
We say that two non-void activities $a$ and $b$ are {\em equivalent} if for every player $i \in N$ and every $\ell \in [n]$ 
it holds that $(a,\ell) \sim_{i} (b,\ell)$. A non-void activity $a \in A^{*}$ is called {\em copyable} if $A^{*}$ contains 
at least $n$ activities that are equivalent to $a$ (including $a$ itself). 
We say that player $i \in N$ {\em approves} an alternative $(a,k)$ if $(a,k) \succ_{i} (a_{\emptyset},1)$. 

An outcome of a \sfGASP\ is an {\em assignment} of activities $A$ to players $N$, i.e., a mapping $\pi:N \rightarrow A$.
Given an assignment $\pi:N \rightarrow A$ and a non-void activity $a \in A^{*}$,
we denote by $\pi^{a}=\{\, i \in N \mid \pi(i)=a \,\}$ the set of players 
assigned to $a$. Also, if $\pi(i)\neq a_\emptyset$,
we denote  by $\pi_{i}=\{i\}\cup\{\, j \in N \mid \pi(j)=\pi(i)\}$ 
the set of players assigned to the same activity as player $i \in N$;
we set $\pi_i=\{i\}$ if $\pi(i)= a_\emptyset$.
%IR
An assignment $\pi:N \rightarrow A$ for a \sfGASP\ is {\em individually rational} (IR) if 
%each player weakly prefers her own activity to doing nothing, i.e., 
for every player $i \in N$ with $\pi(i)\neq a_\emptyset$
we have $(\pi(i),|\pi_{i}|) \succeq_{i} (a_{\emptyset},1)$.
%blocking
A coalition $S \subseteq N$ and an activity $a \in A^{*}$ {\it strongly block} an assignment $\pi:N \rightarrow A$ if 
$\pi^a\subseteq S$ and $(a,|S|) \succ_{i} (\pi(i),|\pi_{i}|)$ for all $i \in S$.
%CR
An assignment $\pi:N \rightarrow A$ for a \sfGASP\ is called {\em core stable} (CR) if it is individually rational, and 
there is no coalition $S \subseteq N$ and activity $a \in A^{*}$ 
such that $S$ and $a$  strongly block $\pi$.
%deviation
Given an assignment $\pi:N \rightarrow A$ of a \gGASP, a player $i \in N$ is said to have 
\begin{itemize}
\item an {\em NS-deviation} 
to activity $a \in A^{*}$ if $i$ strictly prefers to join the group $\pi^a$, i.e., $(a,|\pi^a|+1) \succ_{i} (\pi(i),|\pi_{i}|)$.
\item an {\em IS-deviation} if it is an NS-deviation, and all players in $\pi^a$ accept it, 
i.e., $(a,|\pi^a|+1) \succeq_{j} (a,|\pi^a|)$ for all $j \in \pi^a$.
\end{itemize}

\subsection{Graphs and digraphs}
An {\em undirected graph}, or simply a {\em graph}, is a pair $(N,L)$, 
where $N$ is a finite set of {\em nodes} and 
$L\subseteq \{\, \{i, j\}\mid i, j\in N,i \neq j \,\}$ is a collection of {\em edges} between nodes. 
Given a set of nodes $S$, the {\em subgraph of $(N, L)$ induced by $S$}
is the graph $(S, L|S)$, where $L|S=\{\{i, j\}\in L\mid i, j\in S\}$.
%path
For a graph $(N,L)$, a sequence of distinct nodes $(i_1, i_2, \ldots, i_k)$, $k\geq 2$, 
is called a {\it path} in $L$ if $\{i_h,i_{h+1}\} \in L$ for $h\in [k-1]$. 
%cycle
A path $(i_1, i_2, \ldots, i_{k})$, $k\geq 3$, is said to be a {\it cycle} in $L$ if $\{i_{k},i_{1}\} \in L$.
%forest
A graph $(N,L)$ is said to be a {\it forest} if it contains no cycles.
%incident
An edge $e$ is {\em incident} to a node $i$ if $i \in e$. 
%adjacency
A pair of distinct nodes $i,j$ are {\em adjacent} if $\{i,j\} \in L$.
%connectivity
A subset $S \subseteq N$ is said to be {\it connected} in $(N, L)$ if for every pair of distinct nodes 
$i, j \in S$ there is a path between $i$ and $j$ in $L|S$. 
%tree
A forest $(N,L)$ is said to be a {\it tree} if $N$ is connected in $(N, L)$.
%star
A tree $(N,L)$ is called a {\it star} if there exists a central node $c \in N$ that is adjacent to every other node.
%clique
A subset $S\subseteq N$ of a graph $(N,L)$ is said to be a {\it clique} if 
for every pair of distinct nodes $i, j\in S$, $i$ and $j$ are adjacent.
%directed graph
\par
A {\it directed graph}, or a {\it digraph}, is a pair $(N, T)$ where $N$ is a finite set of nodes 
and $T \subseteq N \times N$ is a collection of {\em arcs} between nodes. 
%directed path
A sequence of distinct nodes $(i_1, i_2, \ldots, i_k)$, $k\geq 2$, is called a {\it directed 
path} in $T$ if $(i_h,i_{h+1}) \in T$ for $h=1,2,\ldots,k-1$.
%underlying graph
Given a digraph $(N, T)$, let $L(T)=\{\, \{i,j\} \mid (i,j) \in T \,\}$: 
the graph $(N, L(T))$ is the {\em undirected version} of $(N, T)$. 
%rooted tree
A digraph $(N, T)$ is said to be a {\it rooted tree} if $(N,L(T))$ 
is a tree and each node has at most one arc entering it. A rooted tree has exactly one node 
that no arc enters, called the {\it root}, and there exists a unique directed path from the root to every node of $N$. 
%parent,child, descendants
Let $(N, T)$ be a rooted tree. A node $i \in N$ is said to be a {\it child} of $j$ in $T$ if $(j,i) \in T$, and to be the {\it parent} of $j$ in $T$ if $(i,j) \in T$. A node $i \in N$ is called a {\it descendant} of $j$ in $T$ if there exists a directed path from $j$ to $i$ in $T$; here, $j$ is called a {\em predecessor} of $i$.

\section{Our Model}
We now define a group activity selection problem where communication links between the players 
are represented by an undirected graph.
\begin{define}
An instance of the {\em Group Activity Selection Problem with graph structure} (\gGASP) is given by an instance 
$(N,(\succeq_{i})_{i \in N},A)$ of a \sfGASP\ and a set of communication links between players $L \subseteq \{\, \{i,j\} 
\mid i,j\in N \land i\neq j \,\}$.
\end{define}

A coalition $S \subseteq N$ is said to be {\em feasible} if $S$ is connected in the graph $(N,L)$. An outcome of a 
\gGASP\ is a {\em feasible assignment} $\pi:N \rightarrow A$ such that $\pi_{i}$ is a feasible 
coalition for every $i \in N$. We adapt the definitions of stability concepts to our setting as follows. We say that a 
deviation by a group of players is {\em feasible} if the deviating coalition itself is feasible; 
a deviation by an individual player where player $i$ 
joins activity $a$ is {\em feasible} if $\pi^{a}\cup \{i\}$ is feasible. We modify the definitions in the previous section by 
only requiring stability against feasible deviations.
Note that an ordinary \sfGASP\ (without graph structure) is equivalent to a \gGASP\ 
where the underlying graph $(N,L)$ is complete.

A key feature of \gGASP\ as well as \GASP\ is that players' preferences are {\em anonymous}, i.e., players do not care about the identities of the group members.  We can thus show that checking whether a given feasible assignment is core stable is easy, irrespective of the structure of the social network. The proposition below generalizes Theorem~11 of \citet{Darmann2015}. Note that in many other contexts, deciding whether a given assignment is core stable is coNP-hard, for example in additively separable hedonic games \citep{Sung2007}. 
\begin{proposition}\label{prop:in-core}
Given an instance \,\,$(N, (\succeq_i))_{i\in N}, A, L)$\,\,
of \gGASP\ and a feasible assignment $\pi$ for that instance, 
we can decide in $O(pn^3)$ time whether $\pi$ is core stable. 
\end{proposition}
\begin{proof}
Let $A=A^*\cup\{a_\emptyset\}$ and let $n=|N|$.
By scanning the assignment $\pi$ and the players' preferences, 
we can check whether $\pi$ is individually rational.
Now, suppose that this is the case. Then, for each $a\in A^*$
and each $s\in[n]$ we can check if there is a deviation
by a connected coalition of size $s$ that engages in $a$.
To this end, we consider the set $S_{a, s}$ of all players who strictly
prefer $(a, s)$ to their assignment under $\pi$ and verify
whether $S_{a, s}$ has a connected component of size at least $s$
that contains $\pi^a$; if this is the case, $\pi^a$ (which is itself connected or empty)
could be extended to a connected coalition of size exactly 
$s$ that strongly blocks $\pi$. If no such deviation exists, 
$\pi$ is core stable. The existence of a connected component of a given size can be checked in $O(n^2)$ time by using depth first search algorithm.
\end{proof}

In this paper, we will be especially interested in \gGASP s where $(N,L)$ is \emph{acyclic}. This restriction 
guarantees the existence of stable outcomes in many other cooperative game settings \citep{Demange2004}. However, this is not the case 
for \gGASP: here, all stable outcomes under consideration may fail to exist, even if $(N,L)$ is a path or a star.

\begin{example}[Stalker game]\label{ex:NS:empty}
	Consider a \gGASP\ with $N=\{1,2\}$, $A^{*}=\{a\}$, $L=\{\{1,2\}\}$, where preferences $(\succeq_{i})_{i \in N}$ are given as follows:
	\begin{align*}
	1:&~ (a,1) \succ (a_{\emptyset},1) \\
	2:&~ (a,2) \succ (a_{\emptyset},1)
	\end{align*}
	Thus, player 1 wishes to participate in activity $a$ alone, while player 2 (the ``stalker'') wants to participate in activity $a$ together with player 1.

This instance admits no Nash stable outcomes: If all players engage in the void activity, player 1 wants to start doing activity $a$. If player 1 does activity $a$,
then player 2 wants to join her coalition, causing player 1 to deviate
to the void activity. \qed
\end{example}

Similarly, a core stable outcome is not guaranteed to exist even for \gGASP s on paths and stars, as the following example shows.
\begin{example}\label{ex:core:empty}
Consider a \gGASP\ with $N=\{1,2,3\}$, $A^{*}=\{a,b\}$, $L=\{\{1,2\},\{2,3\}\}$, where preferences $(\succeq_{i})_{i \in N}$ are given as follows:
\begin{align*}
1:&~ (b,2) \succ (a,3) \succ (a_{\emptyset},1)\\
2:&~ (a,2) \succ (b,2) \succ (a,3) \succ (a_{\emptyset},1)\\
3:&~ (a,3) \succ (b,1) \succ (a,2) \succ (a_{\emptyset},1)
\end{align*}
We will argue that each individually rational feasible assignment $\pi$ admits a strongly blocking feasible coalition and activity. 
If all players do nothing, then player $3$ and activity $b$ strongly block $\pi$.
Now, there are only four individually rational feasible assignments where some player is engaged in a non-void activity; each of them is strongly blocked by some coalition and activity as follows (we write $S\to x$ to indicate that coalition $S$ strongly blocks $\pi$ together with activity $x$): 
\begin{enumerate}[label=(\arabic*)]
\item
$\pi(1)=b$, $\pi(2)=b$, $\pi(3)=a_{\emptyset}$: $\{2,3\} \to a$; 
\item
$\pi(1)=a_{\emptyset}$, $\pi(2)=a$, $\pi(3)=a$: $\{3\} \to b$;
\item
$\pi(1)=a_{\emptyset}$, $\pi(2)=a_{\emptyset}$, $\pi(3)=b$: $\{1,2,3\} \to a$;
\item
$\pi(1)=a$, $\pi(2)=a$, $\pi(3)=a$: $\{1,2\} \to b$; \qed
\end{enumerate}
\end{example}

\citet{Igarashi2016a} showed that in the context of hedonic games, 
acyclicity is sufficient for individually stable outcomes to exist: 
an individually stable partition of players always exists and can be computed in polynomial time. 
In contrast, it turns out that for \gGASP s this is not the case:
an individually stable outcome may fail to exist even if the underlying social network is a path;
moreover, this may happen even if there are only three players and their preferences are strict.

\begin{example}\label{ex:IS:empty}
Consider a \gGASP\ with $N=\{1,2,3\}$, $A^{*}=\{a,b,c\}$, $L=\{\{1,2\},\{2,3\}\}$, 
where players' preferences are as follows:
\begin{align*}
1:&~ (b,2) \succ (a,1) \succ (c,3) \succ (c,2) \succ (c,1) \succ (a_{\emptyset},1)\\
2:&~ (c,3) \succ (c,2) \succ (a,2) \succ (b,2) \succ (b,1) \succ (a_{\emptyset},1)\\
3:&~ (c,3) \succ (a,2) \succ (a,1) \succ (a_{\emptyset},1)
\end{align*}

We will argue that each individually rational feasible assignment $\pi$ admits an IS-deviation.
Indeed, if $\pi(1)=a_\emptyset$ then no player is engaged in $c$ and hence player $1$
can deviate to $c$. Similarly, if $\pi(2)=a_\emptyset$ then no player is engaged in $b$
and hence player $2$ can deviate to $b$. There are $9$ individually rational feasible assignments
where $\pi(1)\neq a_\emptyset$, $\pi(2)\neq a_\emptyset$; for each of them we can find an IS deviation
as follows (we write $i\to x$ to indicate that player $i$ has an IS-deviation to activity $x$): 
\begin{enumerate}[label=(\arabic*)]
\item
$\pi(1)=a$, $\pi(2)=b$, $\pi(3)=a_{\emptyset}$: $1\to b$; %player $1$ has an IS-deviation to $b$; 
\item
$\pi(1)=b$, $\pi(2)=b$, $\pi(3)=a_{\emptyset}$: $3\to a$; %player $3$ has an IS-deviation to $a$;
\item
$\pi(1)=b$, $\pi(2)=b$, $\pi(3)=a$: $2\to a$; %player $2$ has an IS-deviation to $a$;
\item
$\pi(1)=c$, $\pi(2)=a$, $\pi(3)=a$: $2\to c$; %player $2$ has an IS-deviation to $c$;
\item
$\pi(1)=c$, $\pi(2)=b$, $\pi(3)=a_\emptyset$: $3\to a$;
\item
$\pi(1)=c$, $\pi(2)=b$, $\pi(3)=a$: $2\to a$;
\item
$\pi(1)=c$, $\pi(2)=c$, $\pi(3)=a_{\emptyset}$: $3\to a$; %player $3$ has an IS-deviation to $a$;
\item
$\pi(1)=c$, $\pi(2)=c$, $\pi(3)=a$: $3\to c$; %player $3$ has an IS-deviation to $c$;
\item
$\pi(1)=c$, $\pi(2)=c$, $\pi(3)=c$: $1\to a$. %player $1$ has an IS-deviation to $a$. 
\end{enumerate}
Notice that the instance does not admit a core stable outcome either: if such an outcome existed, the assignment would satisfy individual stability due to the fact that all the preferences are strict, a contradiction to what we have seen above. 
\qed
\end{example}

\subsection{Copyable cases}
If all activities are copyable, we can effectively 
treat \gGASP\ as a non-transferable utility game on a graph. In particular, 
we can invoke a famous result of \citet{Demange2004} concerning the stability of non-transferable 
utility games on trees. Thus, requiring all activities to be copyable 
allows us to circumvent the non-existence result for the core 
(Example~\ref{ex:core:empty}). The argument is constructive.

\begin{theorem}[implicit in the work of~\citet{Demange2004}]\label{thm:core:copyable}
For every \gGASP\ where each activity $a \in A^{*}$ is copyable 
and $(N,L)$ is acyclic, a core stable feasible assignment exists and can be found in time polynomial in $p$ and~$n$.
\end{theorem}
\begin{proof}
If the input graph $(N, L)$ is a forest, we can process each of its connected components separately, so we assume that $(N, L)$ is a tree. Prior to giving a formal description of the algorithm (Algorithm \ref{alg:core}), we outline the basic idea. 
We choose an arbitrary node $r$ as the root and construct a rooted tree $(N,T)$ by orienting the edges in $L$ towards the leaves. We denote by $\ch(i)$ the set of children of $i$ and by $\desc(i)$ the set of descendants of $i$ (including $i$) in the rooted tree. For each $i  \in N$, we define $\height(i)=0$ if $\ch(i)=\emptyset$, and $\height(i)=1+\max \{\, \height(j)\mid j \in \ch(i)\,\}$ otherwise. We denote by $(N,T|S)$ the subdigraph induced by $S \subseteq N$, i.e., $T|S=\{\, (i,j) \in T \mid i,j \in S \,\}$. We define $\calC(i)$ to be the set of connected subsets of $\desc(i)$ for each $i \in N$.

The algorithm has two different phases: the bottom-up and the top-down phase. 
\begin{itemize}
	\item \emph{Bottom-up phase}: In the bottom-up phase, we will determine a {\em guaranteed} activity $a(i)$ and coalition $S(i)$ for every subroot $i$. To this end, we choose a connected subset $S(i)$ of $\desc(i)$ that maximizes $i$'s utility under the constraint that every descendant $j$ of $i$ in the coalition can agree, i.e., for any descendant $j \in S(i)$ of $i$, $(a(i),|S(i)|)$ is at least as preferred as $(a(j),|S(j)|)$. The utility level of $(a(i),|S(i)|)$ determined for each player $i \in N$ can be interpreted as $i$'s guaranteed utility in the final outcome.  
	\item \emph{Top-down phase}: In the top-down phase, the algorithm builds a feasible assignment $\pi$, by iteratively choosing a root $r^{\prime}$ of the remaining rooted trees and reassigning the activity $a(r^{\prime})$ to its coalition $S(r^{\prime})$. Since each activity is copyable, we can always find an activity that is equivalent to $a(r^{\prime})$ and has not been used by their predecessors.
\end{itemize}

This procedure is formalized in Algorithm~\ref{alg:core}.
\begin{algorithm}
\SetKwInOut{Input}{input} 
\SetKwInOut{Output}{output}
\SetKw{And}{and}
\SetKw{None}{None}
\caption{Finding core stable partitions}\label{alg:core}
\Input{tree $(N,L)$, activity set $A=A^*\cup \{a_{\emptyset}\}$, $r\in N$, preference $\succeq_i$, $i\in N$}
\Output{$\pi:N \rightarrow A$}
	{\tt // Bottom-up phase: assign activities to players in a bottom-up manner}\;
	make a rooted tree $(N,T)$ with root $r$ by orienting all the edges in $L$\;
	initialize $S(i)\leftarrow \{i\}$ and $a(i) \leftarrow a_{\emptyset}$ for each $i \in N$\;
	\ForEach{$t=0,\ldots,\height(r)$}{
	\ForEach{$i \in N$ with $\height(i)=t$}{
	{\tt // Find $i$'s favourite alternative where all members of $i$'s coalition can agree to join.}
	\ForEach{$(a,k) \in (A^{*} \times [n])\cup \{(a_{\emptyset},1)\}$}{
	\If{there exists $S \in \calC(i)$ such that $|S|=k$, $(a,k) \succeq_j (a(j),|S(j)|)$ for all $j \in S$, and $(a,k) \succ (a(i),|S(i)|)$ \label{step:favourite}}{
	set $S(i) \leftarrow S$ and $a(i) \leftarrow a$\;
	}
	}
	}
	}
	{\tt // Top-down phase: relabel players with their predecessor's activities}\;
	set $N^{\prime} \leftarrow N$ and $A^{\prime} \leftarrow A^*$\;
	\While{$N^{\prime} \neq \emptyset$}{
	choose a root $r^{\prime}$ of some connected component of the digraph $(N^{\prime},T|{N^{\prime}})$ and find an activity $b \in A^{\prime}\cup \{a_{\emptyset}\}$ that is equivalent to $a(r^{\prime})$\;
	set $\pi(i) \leftarrow b$ for all $i \in S(r^{\prime})$\;
	set $N^{\prime} \leftarrow N^{\prime} \setminus S(r^{\prime})$ and $A^{\prime} \leftarrow A^{\prime} \setminus \{b\}$\;
}
\end{algorithm}

We will now argue that Algorithm \ref{alg:core} correctly finds a core stable feasible assignment.
Observe that the following lemma holds due to the {\bf if} statement in line \ref{step:favourite} of the algorithm.
\begin{lemma}\label{lem0}
For all $i \in N$, the following statements hold:
\begin{itemize}
\item[$(${\rm i}$)$] $i$ has no incentive to deviate to an alternative of size $1$, i.e., $(a(i),|S(i)|)\succeq_i (b,1)$ for all $b \in A$,
\item[$(${\rm ii}$)$] all players in $S(i)$ weakly prefer $(a(i),|S(i)|)$ to their guaranteed alternative $(a(j),|S(j)|)$.
\end{itemize}
\end{lemma}

Now, by $(${\rm ii}$)$ in Lemma \ref{lem0}, it can be easily verified that 
\begin{equation}\label{eq:core}
(\pi(i),|\pi_i|) \succeq_i (a(i),|S(i)|)~\mbox{for all}~i \in N. 
\end{equation}
Combining this with $(${\rm i}$)$, we know that at the assignment $\pi$, all players weakly prefer their alternatives to engaging alone in unused activities or the void activity. It thus remains to show that no connected coalition together with a non-void activity strongly blocks $\pi$.

Take any connected subset $S \subseteq N$ and activity $a \in A^*$. Let $i$ be the subroot of the coalition $S$ so that $S \in \calC(i)$. First, consider the case when $(a(i),|S(i)|) \succeq_i (a,|S|)$. By \eqref{eq:core}, it is clear that the coalition $S$ and the activity $a$ do not strongly block $\pi$. Second, consider the case when $(a,|S|) \succ_i (a(i),|S(i)|)$. By line \ref{step:favourite} of the algorithm, this means that there is a player $j \in S$ such that $(a(j),|S(j)|)\succ_{j} (a,|S|)$. Thus, $S$ and $a$ do not strongly block $\pi$. We conclude that $\pi$ is core stable. 

It remains to analyze the running time of Algorithm~\ref{alg:core}.
Consider the execution of the algorithm for a fixed player $i$.
Let $d=|\desc(i)|$. Line~\ref{step:favourite} checks whether there is a connected coalition of size $k$ that can engage in $a$. Similarly to the proof of Proposition \ref{prop:in-core}, we do this by computing the set $S$ of all descendants in $\desc(i)\setminus \{i\}$ who weakly prefer $(a,k)$ to their guaranteed coalition and verify whether the set $S \cup \{i\}$ has a connected component of size at least $k$. This procedure requires at most $d$ queries: no descendant of $i$ is queried more than once. 
Summing over all players, we conclude that the number of queries for the bottom-up phase is bounded by $O(n^3p)$. It is immediate that the top-down phase can be done in polynomial time.
This completes the proof of the theorem. 
\end{proof}

Similarly, an individually stable outcome is guaranteed to exist in copyable cases if the underlying graph is acyclic. 
Moreover, we can adapt the result of \citet{Igarashi2016a} for hedonic games
and obtain an efficient algorithm for computing an individually stable outcome. The proof can be found in the appendix.

\begin{theorem}\label{thm:IS:copyable}
Each instance of \gGASP\ where each activity $a \in A^{*}$ is copyable 
and $(N,L)$ is acyclic admits an individually stable feasible assignment; 
moreover, such an assignment can be found in time polynomial in $p$ and $n$.
\end{theorem}

The stalker game in Example~\ref{ex:NS:empty} does not admit a Nash stable outcome even if we make all
activities copyable. 
Thus, in contrast to core and individual stability, there is no existence guarantee for Nash stability even if activities are copyable.
However, with copyable activities, we can still \emph{check} for the existence of a Nash stable outcome
in polynomial time if the social network is acyclic.

\begin{theorem}\label{thm:NS:copyable}
Given an instance $(N,A,(\succeq_{i})_{i \in N},L)$ of \gGASP\ where each activity $a \in A^{*}$ is copyable 
and the graph $(N,L)$ is acyclic, one can decide whether it admits a Nash stable outcome in time polynomial in $p$ and 
$n$.
\end{theorem}
\begin{proof}
Again, we assume that $(N, L)$ is a tree. We choose an arbitrary node as the root and construct a rooted tree by orienting the edges in $L$ towards the leaves. We denote by $\ch(i)=\{j_1,j_2,\ldots,j_{|\ch(i)|}\}$ the set of children of $i$; and we denote by $\desc(i,c)$ the set of the descendants of the first $c$-th children of $i$ (including $j_1,j_2,\ldots,j_c$) in the rooted tree. 
Then, for each player $i \in N$, each $c \in [|\ch(i)|]$, each alternative $(a,k) \in X$, and $t\in[k]$
we set $\NS[i,c,(a,k),t]$ to \emph{true} if there exists a feasible assignment $\pi:N \rightarrow A$ such that $|\pi_i \cap (\desc(i,c)\cup \{i\})|= t$, $\pi(i)=a$, each player in $\pi_i \cap (\desc(i,c)\cup \{i\})$ likes $(a, k)$ at least as much as any alternative she can deviate to (including the void activity), and no player in $\desc(i,c)\setminus\pi_i$ has an NS feasible deviation.
Otherwise, we set $\NS[i,c,(a,k),t]$ to \emph{false}.
By construction, there exists a Nash stable feasible assignment if and only if $\NS[r,|\ch(r)|,(a,k),k]$ is {\em true} for some alternative $(a,k) \in X$, where $r$ is the root of the rooted tree.

For each player $i \in N$, each $c \in [|\ch(i)|]$, each alternative $(a,k) \in X$, 
and each $t \in [k]$, we initialize $\NS[i,c,(a,k),t]$ to {\em true} if $t=1$ 
and $i$ weakly prefers $(a,k)$ to any alternative of size $1$, 
and we set $\NS[i,c,(a,k),t]$ to {\em false} otherwise. Then, for $i \in N$ from the bottom to the root, we iterate through all the 
children $j_c$ of $i$ and update $\NS[i,c,(a,k),t]$ step-by-step; more precisely, for each child $j_c$ of $i$ and for 
$t \in [k]$, we set $\NS[i,c,(a,k),t]$ to {\em true} if 
\begin{itemize}
\item $t \geq 2$ and there exists an $x \in [t-1]$ such that both $\NS[i,c-1,(a,k),t-x]$ and 
$\NS[j,c,(a,k),x]$ are {\em true}, or 
\item $\NS[i,c-1,(a,k),t]$ is {\em true}, and players $i$ and $j_c$ can be {\em separated} from each other, i.e., there exists $(b,\ell) \in X$ such that 
(i) $\NS[j_c,|\ch(j)|,(b,\ell),\ell]$ is {\em true}, (ii) $b=a_{\emptyset}$ or $(a,k) \succeq_{i} (b,\ell+1)$, and (iii) $a=a_{\emptyset}$ or $(b,\ell) \succeq_{j} (a,k+1)$.
\end{itemize}
In cases where $\NS[r,|\ch(r)|,(a,k),k]$ is {\em true} for some alternative $(a,k) \in X$, a Nash stable feasible assignment can be found using dynamic programming. 

This can be done in polynomial time since the size of the dynamic programming table is at most $n^3p$ and each entry can be filled in time $O(n^2p)$. This completes the proof of the theorem. 
\end{proof}

We note that these tractability results for copyable cases do not extend to arbitrary graphs: \citet{Ballester2004} showed that it is NP-complete to determine the existence of core or individually or Nash stable outcomes for anonymous hedonic games, which can be considered as a subclass of \gGASP s whose graph is a clique. 
\section{NP-completeness results}
We now move on to the case where each activity can be used at most once. For other types of cooperative games, many desirable outcomes can be computed in polynomial time if the underlying network structure is simple \citep{Chalkiadakis2016,Elkind2014,Igarashi2016a}. In particular, \citet{Igarashi2016a} showed it is easy to compute stable partitions for hedonic games on trees with polynomially bounded number of connected coalitions. However, computing stable 
solutions of \gGASP\ turns out to be NP-complete even if the underlying network is a path, a star, or a graph with constant-size 
connected components. For each family of graphs, we will reduce from a different combinatorial problem that is structurally similar to our problem. %Essentially, we create unstable gadgets; in order to stabilize them, we need to assign activities to players so that the assignments correspond to desired solutions of the respective problem.

%path%%%%%%%%%%%%%%%%%%%%%%%%%%%%%%%%%%%%%%%%%%%%
\paragraph{Paths.}
Our proof for paths is by reduction from a restricted version of the NP-complete problem {\sc Rainbow Matching}. Given a 
graph $G$ and a set of colors $\calC$, a {\em proper edge coloring} is a mapping $\phi:E \rightarrow \calC$ where $\phi(e) \neq \phi(e^{\prime})$ 
for all edges $e, e^{\prime}$ such that $e\neq e^\prime$ and $e\cap e^{\prime}\neq\emptyset$.
Without loss of generality, we assume that $\phi$ is surjective.
 A {\em properly edge-colored graph} $(G,\calC,\phi)$ is a graph together with a set of colors 
and a proper edge coloring. A matching $M$ 
in an edge-colored graph $(G,\calC,\phi)$ is called a {\em rainbow matching} if all edges of $M$ have different colors. Given a properly edge-colored graph $(G, \calC,\phi)$ together with an integer $k$, {\sc Rainbow Matching} asks whether $G$ admits a rainbow matching with at least $k$ edges.  
\citet{Le2014} show that {\sc Rainbow Matching} remains NP-complete even for properly edge-colored paths. 

\begin{theorem}\label{thm:NPhardness:path:NS}
Given an instance of \gGASP\ whose underlying graph is a path, it is {\em NP}-complete to determine whether it has a Nash stable feasible assignment.
\end{theorem}
\begin{proof}
Clearly, our problems are contained in NP since we can easily check whether a given assignment is Nash stable. The hardness proof proceeds by a reduction from {\sc Path Rainbow Matching}.

\emph{Construction.}
Given an instance $(G, \calC, \phi, k)$ of {\sc Path Rainbow Matching} where 
$|\calC|=q$, we create a vertex 
player $v$ for each $v \in V$ and an edge player $e$ for each $e \in E$. To create the social network, 
we start with $G$ and place each edge player in the middle of the respective edge, i.e., 
we let $N_{G}=V \cup E$ and $L_{G}=\{\, \{v,e\} \mid v \in e \in E \,\}$. To the 
right of the graph $(N_G,L_G)$, we attach a path consisting of ``garbage collectors'' $\{g_{1},g_{2},\ldots,g_{q-k}\}$ 
and $q$ copies $(N_{c},L_{c})$ of the stalker game where $N_{c}=\{c_1,c_2\}$ and
$L_{c}=\{ \{c_1,c_2\}\}$ for each $c \in \calC$. See Figure \ref{fig:path:Nash}.

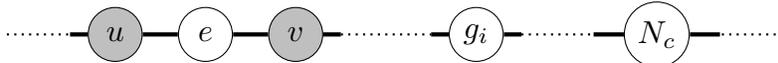
\begin{figure*}[htb]
\centering
%path
\begin{tikzpicture}[scale=1.2, transform shape, inner sep=0pt, minimum size=17pt, font=\small]
	\node[draw, circle,fill=gray!50](u) at (0,0) {$u$};
	\node[draw, circle](e) at (1,0) {$e$};
	\node[draw, circle,fill=gray!50](v) at (2,0) {$v$};
	\node[draw, circle](g) at (4,0) {$g_i$};
	\node[draw, circle](Nc) at (6,0) {$~N_c~$};

	\draw[dotted, >=latex,thick] (-1.2,0)--(-0.5,0);				
	\draw[-, >=latex,ultra thick] (-0.5,0)--(u);	
	
	\draw[-, >=latex,ultra thick] (u)--(e);
	\draw[-, >=latex,ultra thick] (e)--(v);
	\draw[-, >=latex,ultra thick] (2.5,0)--(v);		
	
	\draw[dotted, >=latex,thick] (2.5,0)--(3.5,0);
	\draw[-, >=latex,ultra thick] (g)--(3.5,0);
	\draw[-, >=latex,ultra thick] (g)--(4.5,0);

	\draw[dotted, >=latex,thick] (4.5,0)--(5.3,0);
	\draw[-, >=latex,ultra thick] (Nc)--(5.3,0);
	\draw[-, >=latex,ultra thick] (Nc)--(6.7,0);
	\draw[dotted, >=latex,thick] (7.4,0)--(6.7,0);
\end{tikzpicture}

\caption{
The graph constructed in the hardness proof for \gGASP s on paths. 
\label{fig:path:Nash}
}
\end{figure*}

We introduce a color activity $c$ for each color $c \in \calC$.
Each vertex player $v$ approves color activities $\phi(e)$ of its adjacent edges $e$ with size $3$; each edge player $e$ 
approves the color activity $\phi(e)$ of its color with size $3$; each garbage collector $g_{i}$ approves any color 
activity $c$ with size $1$; finally, for players in $N_{c}$, $c \in \calC$,  player $c_1$ approves its color activity 
$c$ with size $1$, whereas player $c_2$ approves $c$ with size $2$.

\emph{Correctness.}
We show
that $G$ has a rainbow matching of size at least $k$ if and only if there exists a Nash stable 
feasible assignment.

Suppose that there exists a rainbow matching $M$ of size $k$. We construct a feasible assignment 
$\pi$ where for each $e=\{u, v\}\in M$ we set $\pi(e)=\pi(u)=\pi(v)=\phi(e)$,
each garbage collector $g_{i}$, $i\in[q-k]$, is arbitrarily assigned to one of the remaining $q-k$ color activities, 
and the remaining players are assigned to the void activity. The assignment $\pi$ is Nash stable, since every garbage 
collector as well as every edge or vertex player assigned to a color activity
are allocated their top alternative, and no remaining player has an NS feasible deviation.

Conversely, suppose that there is a Nash stable feasible assignment $\pi$. Let $M=\{\, e \in E(G) \mid \pi(e) \in \calC 
\,\}$. We will show that $M$ is a rainbow matching of size at least $k$. To see this, notice that $\pi$ cannot 
allocate a color activity to a member of $N_{c}$, since otherwise no feasible assignment would be Nash stable. Further, at 
most $q-k$ color activities are allocated to the garbage collectors, which means that at least $k$ color activities 
should be assigned to vertex and edge players. The only individually rational way to do this is to select triples of the 
form $(u,e,v)$ where $e=\{u,v\} \in E(G)$ and assign to them their color activity $\phi(e)$. Thus, $M$ is a rainbow 
matching of size at least $k$.
\end{proof}

%star%%%%%%%%%%%%%%%%%%%%%%%%%%%%%%%%%%%%%%%%%%%
\paragraph{Stars.}
For \gGASP s on stars we provide a reduction from the NP-complete problem {\sc Minimum Maximal 
Matching} (MMM). Given a graph $G$ and a positive integer $k \leq |E|$, MMM asks whether 
$G$ admits a maximal matching with at most $k$ edges. The problem remains NP-complete for 
bipartite graphs \citep{Demange2008}.

\begin{theorem}\label{thm:NPhardness:star:NS}
Given an instance of \gGASP\ whose underlying graph is a star, 
it is {\em NP}-complete to determine whether it has a Nash stable feasible assignment.
\end{theorem}
\begin{proof}
Clearly, our problem is in NP. To prove NP-hardness, we reduce from MMM on bipartite graphs.

\emph{Construction.}
Given a bipartite graph $G=(U\cup V,E)$ with vertex bipartition $(U,V)$ and an integer $k$, we create a star with center $c$ and 
$|V|+1$ leaves: one leaf for each vertex player $v \in V$ plus one stalker $s$. See Figure \ref{fig:star:Nash}.

\begin{figure*}[htb]
\centering
%star
\begin{tikzpicture}[scale=1, transform shape]
	\def \radius {1.5cm}
	\node[draw, circle, minimum size=22pt](c) at (0:0) {$c$};
	\node[draw, circle](s1) at ({90}:\radius) {$s$};
	\node[draw, circle,fill=gray!50](node1) at ({195}:\radius) {$v_{1}$};
	\node[draw, circle,fill=gray!50](node2) at ({232}:\radius) {$v_{2}$};
	\node[draw, circle,fill=gray!50](node3) at ({270}:\radius) {$v_{3}$};
	\node[draw, circle,fill=gray!50](node4) at ({345}:\radius) {$v_{n}$};
    
 	\draw[-, >=latex,ultra thick] (c)--(node1);
	\draw[-, >=latex,ultra thick] (c)--(node2);
	\draw[-, >=latex,ultra thick] (c)--(node3);	
	\draw[-, >=latex,ultra thick] (c)--(node4);	
	\draw[-, >=latex,ultra thick] (c)--(s1);
	\draw [dotted,thick] (0.7,-\radius+0.5) arc [radius=1.5, start angle=290, end angle= 320];
\end{tikzpicture}
\caption{
The graph constructed in the hardness proof for \gGASP s on stars. 
\label{fig:star:Nash}
}
\end{figure*}
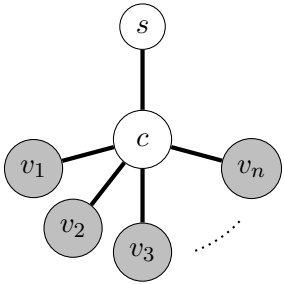

We introduce an activity $u$ for each $u \in U$, and two additional activities $a$ and $b$. 
A player $v \in V$ approves $(u, 1)$ for each activity $u$ such that $\{u, v\}\in E$ as well as  
$(a, |V|-k+1)$ and prefers the former to the latter. That is, 
$(u,1) \succ_{v}(a,|V|-k+1)$ for every $u \in U$ with $\{u, v\}\in E$; $v$ is indifferent among 
the activities associated with its neighbors in the graph, that is, 
$(u,1)\sim_{v} (u^{\prime},1)$ for all $u,u^{\prime} \in U$ such that $\{u, v\} \in E$ and $\{u^{\prime}, v\}\in E$. 
The center player $c$ approves both $(a, |V|-k+1)$ and $(b,1)$, and prefers the former to the latter, 
i.e., $(a,|V|-k+1) \succ_{c} (b,1) \succ_{c} (a_{\emptyset},1)$. 
Finally, the stalker $s$ only approves $(b,2)$.

\emph{Correctness.}
We now show that $G$ admits a maximal matching $M$ with at most $k$ edges
if and only if our instance of \gGASP\ admits a Nash stable assignment.

Suppose that $G$ admits a maximal matching $M$ with at most $k$ edges. We construct a feasible assignment $\pi$ 
by setting $\pi(v)=u$ for each $\{u, v\}\in M$, assigning $|V|-k$ vertex players and the center to $a$,
and assigning the remaining players to the void activity. Clearly, the center $c$ has no 
incentive to deviate and no vertex player in a singleton coalition wants to deviate to the coalition of the center. 
Further, no vertex $v$ has an NS-deviation to an unused activity $u$, since if $\pi$ admits such a deviation, this 
would mean that $M\cup \{u,v\}$ forms a matching, contradicting maximality of $M$. 
Finally, the stalker player does not deviate since 
the center does not engage in $b$. Hence, $\pi$ is Nash stable.

Conversely, suppose that there exists a Nash stable feasible assignment $\pi$ and let $M=\{\, \{\pi(v),v\} \mid v \in V 
\land \pi(v) \in U \,\}$. We will show that $M$ is a maximal matching of size at most $k$. By Nash stability, the stalker 
player should not have an incentive to deviate, and hence the center player and $|V|-k$ vertex players are assigned to 
activity $a$. It follows that $k$ vertex players are not assigned to $a$, and therefore 
$|M| \leq k$. Moreover, $M$ is a matching since each vertex player is assigned to at most one activity, 
and by individual rationality 
each activity can be assigned to at most one player. Now suppose towards a contradiction that $M$ is not 
maximal, i.e., there exists an edge $\{u,v\} \in E$ such that $M\cup \{u,v\}$ is a matching. 
This would mean that in $\pi$ no player is assigned to $u$, and $v$ is assigned to the void activity; 
hence, $v$ has an NS-deviation to $u$, contradicting the Nash stability of~$\pi$.
\end{proof}

%%%%%%%%%%%%%%%%%%%%%%%%%%%%%%%%%%%%%%%%%%%%
\paragraph{Small components.}
In the analysis of cooperative games on social networks one can usually assume that the social network
is connected: if this is not the case, each connected component can be processed separately.
This is also the case for \gGASP\ as long as all activities are copyable. However, if each activity
can only be used by a single group, different connected components are no longer independent,
as they have to choose from the same pool of activities. Indeed, we will now show
that the problem of finding stable outcomes remains NP-hard even if the size of each connected component 
is at most four. Our hardness proof for this problem proceeds by reduction from a restricted version of {\sc 3Sat}. 
Specifically, we consider (3,B2)-{\sc Sat}: in this version of {\sc 3Sat} each clause contains exactly $3$ literals, and each variable occurs exactly twice positively and twice negatively. This problem is known to be NP-complete \citep{Berman2003}. 

\begin{theorem}\label{thm:NPhardness:sc:NS}
Given an instance of \gGASP\ whose underlying graph has connected components whose size is bounded by $4$, 
it is {\em NP}-complete to determine whether it has a Nash stable feasible assignment.
\end{theorem}
\begin{proof}
Our problem is in NP. We reduce from (3,B2)-{\sc Sat}. 

\emph{Construction.}
Consider a formula $\phi$ with variable set $X$ and clause set $C$, 
where for each variable $x\in X$ we write $x_1$ and $x_2$ for the two positive occurrences of $x$, 
and $\bar x_1$ and $\bar x_2$ for the two negative occurrences of $x$. 
For each $x\in X$, we introduce four players $p_1(x), p_2(x), {{\bar p_1}(x)}$, and ${{\bar p_2}(x)}$. 
For each clause $c\in C$, we introduce one stalker $s_{c}$ 
and three other players $c_1,c_2$, and $c_3$.
The network $(N,L)$ consists of one component for each clause---a star with center $s_{c}$ and leaves  
$c_1$, $c_2$, and $c_3$---and of two components for each variable $x\in X$ consisting of a single edge each: $\{p_1(x), p_2(x)\}$ and $\{{\bar p_1}(x), {\bar p_2}(x)\}$. See Figure \ref{fig:sc:Nash}.
Thus,
the size of each component of this graph is at most~$4$.

\begin{figure*}[htb]
	\centering
	\begin{tikzpicture}[scale=0.9, transform shape]
	\def \radius {1.5cm}
	
	\draw[dotted, >=latex,thick] (-2.5,-1)--(-2,-1);
	\draw[dotted, >=latex,thick] (2.5,-1)--(2,-1);
	
	\node[draw, circle](c) at (0,0) {$s_c$};
	\node[draw, circle](node1) at ({232}:\radius) {$c_{1}$};
	\node[draw, circle](node2) at ({270}:\radius) {$c_{2}$};
	\node[draw, circle](node3) at ({305}:\radius) {$c_{3}$};
    
 	\draw[-, >=latex,ultra thick] (c)--(node1);
	\draw[-, >=latex,ultra thick] (c)--(node2);
	\draw[-, >=latex,ultra thick] (c)--(node3);		

	\draw[dotted, >=latex,thick] (-4,-3)--(-4.6,-3);
	\draw[dotted, >=latex,thick] (4.6,-3)--(4,-3);
			
	\node[draw,circle](x1) at (-3,-3) {$p_1(x)$};
	\node[draw,circle](x2) at (-1,-3) {$p_2(x)$};
	\draw[-, >=latex,ultra thick] (x1)--(x2);		

	\node[draw,circle](bx1) at (1,-3) {${{\bar p_1}(x)}$};
	\node[draw,circle](bx2) at (3,-3) {${{\bar p_2}(x)}$};
	\draw[-, >=latex,ultra thick] (bx1)--(bx2);				
	\end{tikzpicture}
	
	\caption{
		The graph used in the hardness proof for \gGASP s on graphs with small components. 
		\label{fig:sc:Nash}
	}
\end{figure*}
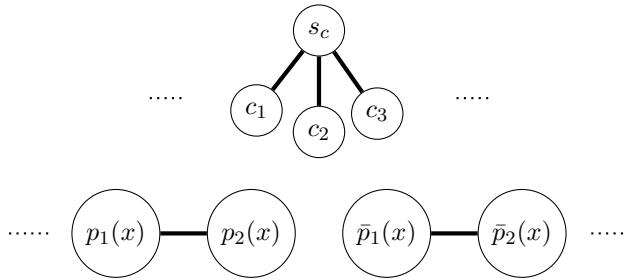

For each $x\in X$ we introduce one variable activity $x$, 
two positive literal activities $x_1$ and $x_2$, and
two negative literal activities ${\bar x_1}$ and ${\bar x_2}$, which correspond to the four occurrences of $x$; 
also, we introduce two further activities $a(x)$ and ${\bar a}(x)$. 
Finally, we introduce an activity $c$ for each clause $c \in C$. 
Thus,
\[
A^{*}=\bigcup_{x \in X}\{x,x_1,x_2,{\bar x_1},{\bar x_2},a(x),{\bar a}(x)\} \cup C.
\]
For each $x\in X$ the preferences of the positive literal players $p_1(x)$ and $p_2(x)$ are given as follows:
\begin{align*}
&p_1(x)\text{:}\,(x,2) \succ (x,1) \succ (x_1,1)\succ (x_2,2) \succ (a(x),1)\succ (a_{\emptyset},1),\\
&p_2(x)\text{:}\,(x,2)\succ (x_2,1)\succ (x_1,2) \succ (a(x),2)\succ (a_{\emptyset},1).
\end{align*}
If one of the positive literal players $p_1(x)$ and $p_2(x)$ is engaged in the void activity and the other is engaged alone in a non-void activity, this would cause the former player to deviate to another activity; thus, in a Nash stable assignment, none of the activities $a(x)$ and $a_\emptyset$ can be assigned to positive literal players. Similarly, for each $x\in X$ the preferences of the negative literal players ${{\bar p_1}(x)}$ and ${{\bar p_2}(x)}$ are given as follows:
\begin{align*}
&{{\bar p_1}(x)}\text{:}\, (x,2) \succ (x,1) \succ ({\bar x_1},1)\succ ({\bar x_2},2) \succ ({\bar a}(x),1)\succ (a_{\emptyset},1),\\
&{{\bar p_2}(x)}\text{:}\, (x,2)\succ ({\bar x_2},1)\succ ({\bar x_1},2) \succ ({\bar a}(x),2)\succ (a_{\emptyset},1).
\end{align*}
As argued above, Nash stable assignments cannot allocate activities  ${\bar a}(x)$ and $a_\emptyset$ to negative literal players. Hence, if there exists a Nash stable assignment, there are only two possible cases: 
\begin{itemize}
\item both players $p_1(x)$ and $p_2(x)$ 
are assigned to $x$, and players ${{\bar p_1}(x)}$ and ${{\bar p_2}(x)}$ are assigned to ${\bar x_1}$ and ${\bar x_2}$, 
respectively; 
\item both players ${{\bar p_1}(x)}$ and ${{\bar p_2}(x)}$ are assigned to $x$, and players $p_1(x)$ and $p_2(x)$  
are assigned to $x_1$ and $x_2$, respectively.
\end{itemize}

For players in $N_{c}$ where $\ell^c_1$, $\ell^c_2$, and $\ell^c_3$ are the literals in a clause $c$, the preferences are given by
\begin{align*}
&c_r:~(\ell^c_{r},1)\succ (c,2) \succ (a_{\emptyset},1), \qquad (r = 1,2,3)\\
&s_{c}:~(\ell^c_{1},2) \sim (\ell^c_{2},2) \sim (\ell^c_{3},2) \sim (c,2)\succ (a_{\emptyset},1).
\end{align*}
That is, players $c_1$, $c_2$, and $c_3$ prefer to engage alone in their approved literal activity, 
whereas $s_{c}$ wants to join one of the adjacent leaves whenever 
$\pi(s_c)=a_\emptyset$ and that leaf is assigned a literal activity; 
however, the leaf would then prefer to switch to the void activity. 
This means that if there exists a Nash stable outcome, 
at least one of the literal activities must be used outside of $N_{c}$, 
and some leaf and the stalker $s_{c}$ must be assigned to activity $c$.

\emph{Correctness.}
We will show that $\phi$ is satisfiable if and only if there exists a Nash stable outcome.

Suppose that there exists a truth assignment that satisfies~$\phi$. First, for each variable $x$ that is set to {\em true}, 
we assign positive literal activities $x_1$ and $x_2$ to the positive literal players $p_1(x)$ and $p_2(x)$, respectively, 
and assign $x$ to the negative literal players ${{\bar p_1}(x)}$ and ${{\bar p_2}(x)}$. For each variable $x$ that is set to {\em false}, 
we assign negative literal activities ${\bar x_1}$ and ${\bar x_2}$ to 
the  negative literal players ${{\bar p_1}(x)}$ and ${{\bar p_2}(x)}$, respectively, and assign $x$ to 
the positive literal players $p_1(x)$ and $p_2(x)$. Note that this procedure uses at least one of the literal activities $\ell^c_{1}$, 
$\ell^c_{2}$ and $\ell^c_{3}$ of each clause $c \in C$, since the given truth assignment satisfies $\phi$. Then, for each 
clause $c \in C$, we select a player $c_{j}$ whose approved activity $\ell^c_{j}$ has been assigned to some literal player, 
and assign $c_{j}$ and the stalker to $c$, and the rest of the clause players to their approved literal activity 
if it is not used yet, and to the void activity otherwise. It is easy to see that the resulting assignment $\pi$ is Nash 
stable.

Conversely, suppose that there exists a Nash stable feasible assignment $\pi$. By Nash stability, for each variable $x 
\in X$, either a pair of positive literal players $p_1(x)$ and $p_2(x)$ or a pair of negative literal players ${\bar 
x_1}$ and ${{\bar p_2}(x)}$ should be assigned to the corresponding pair of literal activities; in addition, for each clause 
$c \in C$, the stalker $s_{c}$ and one of the players $c_1$, $c_2$, and $c_3$ should engage in the activity $c$, thereby 
implying that the approved literal activity of the respective leaf
should be assigned to some literal players. Then, take the truth assignment 
that sets the variable $x$ to {\em true} if its positive literal players $p_1(x)$ and $p_2(x)$ 
are assigned to positive literal activities $x_1$ and $x_2$; 
otherwise, $x$ is set to {\em false}. This assignment can be easily seen to satisfy $\phi$.
\end{proof}

The hardness reductions for the core and individual stability are similar to the respective reductions for Nash stability; essentially, we have to replace copies of the stalker game from Example~\ref{ex:NS:empty} with copies of games with an empty core (Example~\ref{ex:core:empty}) or with no individually stable outcome (Example~\ref{ex:IS:empty}).

\begin{theorem}\label{thm:NPhardness:path:star:sc:CRIS}
Given an instance of \gGASP\ whose underlying graph is a path, a star, or 
has connected components whose size is bounded by $3$, 
it is {\em NP}-complete to determine whether it has a core stable feasible assignment, and it is {\em NP}-complete to determine whether it has an individually stable feasible assignment.
\end{theorem}
\begin{proof}
Clearly, our problems are contained in NP for any social network due to Proposition \ref{prop:in-core} and the fact that we can easily check whether a given assignment is individually stable. The hardness proofs can be found in the appendix. 
\end{proof}

\section{Few Activities}
In the instances of \gGASP\ that are created in our hardness proofs, the number of activities $p$
is unbounded. In reality, however, there are many settings where this parameter can be very small. For instance, when organizing social events for a workshop, the number of activities one could organize is often restricted, due to a limited number of facilities for sports competition or a limited number of buses for a bus trip.  
It is thus natural to 
wonder what can be said when there are few activities to be assigned. It turns out that 
for each of the restricted families of graphs considered in the previous section, 
finding some stable assignments in \gGASP\ is fixed parameter tractable 
with respect to the number of activities. The basic idea behind each of the algorithms we present
is that we fix a set of activities that will be assigned to the players, 
and for each possible subset $B\subseteq A^*$ of activities we 
check whether there exists a stable assignment using the activities
from that subset only. 

\subsection{Small components}
We first present an algorithm for small components based on dynamic programming, allowing us to build up the set $B$ step-by-step. We order the components, and, for each prefix of that ordering, 
we check if a given subset of activities can be assigned to that
prefix in a stable way.
Within each component, we have enough time to consider all possible assignments, 
and each potential deviation involves at most one component. 
The resulting algorithm is FPT with respect to the combined parameter $p+c$, 
where $c$ is a bound on the size of the components of the network.

\begin{theorem}
	\label{thm:FPT:smallcomponents:NSIS}
	There exists an algorithm that given an instance of \gGASP\ checks whether it has a Nash stable or an individually stable, finds one if it exists, 
	and runs in time $O(8^p p^{c+1}c^2 n)$ where $p$ is the number of activities and $c$ is the maximum size of connected components.
\end{theorem}
\begin{proof}
	We give a dynamic programming algorithm. Suppose our graph $(N,L)$ has $k$ connected components $(N_{1},L_{1}), 
	(N_{2},L_{2}), \ldots, (N_{k},L_{k})$ and each component has size at most $c$, i.e., $|N_i| \leq c$ for all $i \in [k]$. For each $i\in[k]$, each set $B \subseteq A^{*}$ of activities 
	assigned to $N$, and each set $B^{\prime} \subseteq B$ of activities assigned to $\bigcup^{i}_{j=1}N_{j}$, we let 
	$\NS[i,B,B^{\prime}]$ and $\IS[i,B,B^{\prime}]$ denote whether there is such an assignment that gives rise to a Nash stable outcome and an individually stable outcome, respectively. 
	Specifically, $\NS[i,B,B^{\prime}]$ (respectively, $\IS[i,B,B^{\prime}]$) is \emph{true} if there exists an individually rational feasible 
	assignment $\pi:\bigcup^{i}_{j=1}N_{j} \rightarrow A$ such that
	\begin{itemize}
		\item the set of activities assigned to $\bigcup^{i}_{j=1}N_{j}$ is exactly $B^{\prime}$; and
		\item no player in $\bigcup^{i}_{j=1}N_{j}$ has an NS (respectively, IS) feasible deviation to 
		an activity in $B^{\prime}$ or to an activity in $A^*\setminus B$.
	\end{itemize}
	Otherwise, $\NS[i,B,B^{\prime}]$ (respectively, $\IS[i,B,B^{\prime}]$) is \emph{false}. 	
	
	For $i=1$, each $B \subseteq A^{*}$, and each $B^{\prime}\subseteq B$, we compute the value of $\NS[1,B,B^{\prime}]$ (respectively, $\IS[1,B,B^{\prime}]$)
	by trying all possible mappings $\pi:N_1 \rightarrow B^{\prime} \cup \{a_{\emptyset}\}$, and checking whether it is an 
	individually rational feasible assignment using all activities in $B^{\prime}$ and such that no player in $N_{1}$ has 
	an NS (respectively, IS) feasible deviation to a used activity in $B^{\prime}$ or an unused activity in $A^*\setminus B$.

	For $i=2,3,\ldots,k$, each $B \subseteq A^{*}$, and $B^{\prime} \subseteq B$, we set 
	$\NS[i,B,B^{\prime}]$ (respectively, $\IS[i,B,B^{\prime}]$) to {\em true} if there exists a bipartition of $B^{\prime}$ into $P$ and $Q$ such that $\NS[i-1,B,P]$ 
	is {\em true} and there exists a mapping $\pi:N_{i} \rightarrow Q\cup \{a_{\emptyset}\}$ such that $\pi$ is an individually 
	rational feasible assignment using all the activities in $Q$, and no player in $N_{i}$ has an NS (respectively, IS) feasible deviation to a used 
	activity in $Q$ or an unused activity in $A^{*}\setminus B$. 
	This property holds because players cannot deviate to an activity played within a different connected component by feasibility.

	It is not difficult to see that a Nash stable (respectively, individually stable) solution 
	exists if and only if $\NS[k,B,B]$ (respectively, $\IS[k,B,B]$) is {\em true} for some $B \subseteq A^{*}$. If this is the case, such a stable feasible 
	assignment can be found using standard dynamic programming techniques. The size of the dynamic programming table $\NS[i,B,B^{\prime}]$ (respectively, $\IS[i,B,B^{\prime}]$) is at most $4^p n$. For each entry, we consider $O(2^p)$ bipartitions of $B^{\prime}$ and the time required to find a stable assignment for a single component is $O(p^{c+1}c^2)$. Thus, each entry can be filled in time $O(2^p p^{c+1}c^2)$. This completes the proof.
\end{proof}

The FPT result for graphs with small connected components can also be adapted to the core. The proof can be found in the appendix.

\begin{theorem}
	\label{thm:FPT:smallcomponents:CR}
	There exists an algorithm that given an instance of \gGASP\ checks whether it has a core stable feasible assignment, finds one if it exists, 
	and runs in time $O(8^p p^{c+1} c^3 n)$ where $p$ is the number of activities and $c$ is the maximum size of connected components.
\end{theorem}

\subsection{Acyclic graphs}
We will next show that computing Nash and individually stable outcomes is in FPT for arbitrary acyclic networks. Essentially, we will construct a rooted tree and check in a bottom-up manner whether there is a partial assignment that is extensible to a stable outcome. 

\begin{theorem}\label{thm:FPT:tree:NS}
The problem of deciding whether an instance of $\gGASP$ with $|A^*|=p$ 
whose underlying social network $(N, L)$ is acyclic
has a Nash stable feasible assignment is in {\sc FPT} with respect to $p$.
\end{theorem}
\begin{proof}
We will first present a proof for the case where $(N, L)$ is a tree; in the end,
we will show how to extend the result to arbitrary forests.
Fix an instance $(N, A, (\succeq_i)_{i\in N}, L)$ of \gGASP\ such that $(N, L)$ is a tree.
We choose an arbitrary node in $N$ as the root of this tree, thereby making $(N, L)$
a rooted tree; we denote by $\ch(i)$ the set of children of $i$ and by $\desc(i)$
the set of descendants of $i$ (including $i$ herself).
Intuitively, we fix the set $B \subseteq A^*$ of activities assigned to the set of all players; then, we process the nodes from the leaves to the root, and for each player $i$ and each subset $B^{\prime}$ of $B$, we check whether there is a partial assignment of $B^{\prime}$ to the descendants of $i$ that is extensible to a stable assignment. 

Formally, for each $i\in N$, each $B\subseteq A^*$, each $B^\prime\subseteq B$,
each $(a,k) \in (B^{\prime} \times [n]) \cup \{(a_{\emptyset},1)\}$, and each $t\in[k]$,
we let $\NS[i,B,B^{\prime},(a,k),t]$ be {\em true} if there is an assignment $\pi:N\to A$
where 
\begin{enumerate}
\item[$(${\rm i}$)$] the set of activities assigned to players in $\desc(i)$ is exactly $B^{\prime}$;
\item[$(${\rm ii}$)$]  player $i$ is assigned to $a$ and is in a coalition with $k$ other players;
\item[$(${\rm iii}$)$]  exactly $t$ players in $\desc(i)$ belong to the same group as $i$, so $|\desc(i)\cap \pi_i| = t$;
\item[$(${\rm iv}$)$]  the $t$ players in $\desc(i)\cap \pi_i$ 
      weakly prefer $(a,k)$ to $(b, 1)$ for each $b\in A\setminus B$, 
      and have no incentive to deviate to the other groups, 
      i.e., every player in $\desc(i)\cap \pi_i$ whose children do not belong to $\pi_i$ likes $(a, k)$ 
      at least as much as each of the alternatives she can deviate to; 
\item[$(${\rm v}$)$]  the players in $\desc(i)\setminus \pi_i$ weakly prefer their alternative under $\pi$ 
      to engaging alone in any of the activities in $A\setminus B$, 
      have no NS-deviation to activities in $B^{\prime}\setminus \{a\}$, 
      and have no incentive to deviate to $i$'s coalition, i.e., if $a \neq a_{\emptyset}$, 
      then every player $j\in \desc(i)\setminus \pi_i$ whose parent belongs to $\pi_i$ 
      likes $(\pi(j), |\pi_j|)$ at least as much as $(a,k+1)$.
\end{enumerate}
Otherwise, we let $\NS[i,B,B^{\prime},(a,k),t]$ be \emph{false}. By construction, our instance admits a Nash stable assignment
if and only if $\NS[r,B,B^{\prime},(a,k),k]$ is {\em true}
for some combination of the arguments $B, B^\prime$, and $(a, k)$,
where $r$ is the root, because by (iv) and (v) there are then no NS-deviations. Notice that the size of the dynamic programming table is at most $p 4^p n^3$.
See Figure \ref{fig:tree1} for the high-level idea of the algorithm.

To complete the proof, we show that values $\NS[i,B,B^{\prime},(a,k),t]$ can be efficiently computed in a bottom-up manner. 
%leaves
If $i$ is a leaf, we set $\NS[i,B,B^{\prime},(a,k),t]$ to \emph{true} if 
$B^{\prime}=\{a\}$, $t=1$, and $i$ weakly prefers $(a,k)$ to every alternative $(b, 1)$ 
such that $b\in A\setminus B$; otherwise, we set $\NS[i,B,B^{\prime},(a,k),t]$ to \emph{false}. 

%internal vertices
Now, consider the case when $i$ is an internal vertex. We order the children of $i$ and let $\ch(i)=\{j_1,j_2,\ldots,j_{|\ch(i)|}\}$. Observe that in each feasible assignment $\pi:\desc(i) \rightarrow B^{\prime}$ the coalition of players that engage in the same activity $b \in B^{\prime}\setminus \{a\}$ is fully contained in a subtree of some child of $i$. This induces a partition $\calP=\{P_1,P_2,\ldots, P_{|\calP|}\}$ of the activity set $B^{\prime}\setminus \{a\}$ and a permutation/indexing of these activity sets 
%aligned in increasing order with respect to the indexes of their matched subtrees. 
so that the activities in $P_i$ are assigned to the subtree of $j_i$.
We will thus go through all possible partitions of $B^{\prime}\setminus \{a\}$ and their permutations, and try to find a feasible assignment compatible with it. Figure \ref{fig:tree2} illustrates how each subtree is matched to a different activity set. 

%IR condition
We first check whether $i$ strictly prefers some alternative $(b,1)$ such that $b\in A\setminus B$ to $(a,k)$; 
if so, we set $\NS[i,B,B^{\prime},(a,k),t]$ to \emph{false}. 
Otherwise, we proceed and guess a partition and indexing $\calP=\{P_1,P_2,\ldots,P_{|\calP|}\}$ of $B^{\prime}\setminus \{a\}$. 
%Without loss of generality, we consider an ordering $P_1,P_2,\ldots,P_{|\calP|}$ and seek to assign each activity set to the subtrees in that order. 
Now we will show that we can determine in polynomial time whether there exists of a stable assignment compatible with this indexed partition.

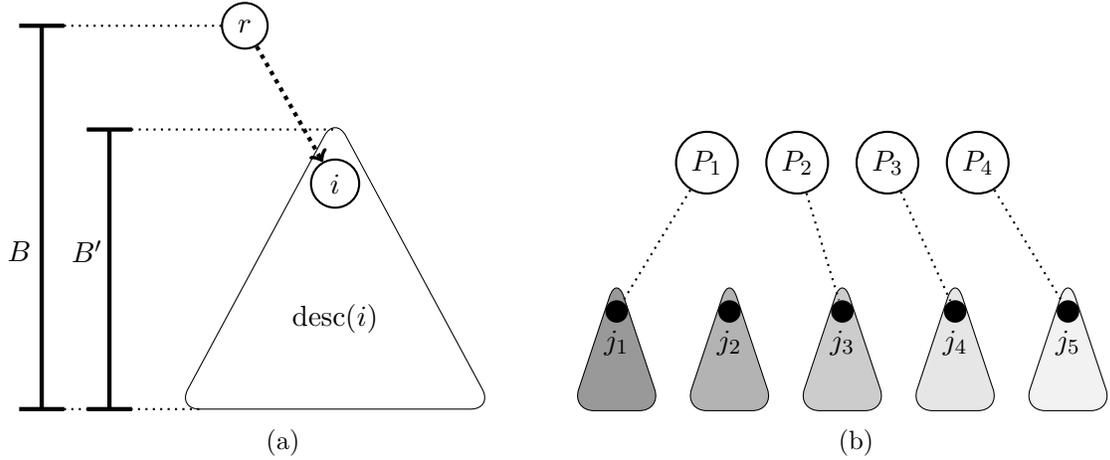
\begin{figure*}[tb]
\centering
\begin{subfigure}[t]{0.5\textwidth}
\begin{tikzpicture}[scale=0.6]

		\node[circle,thick,draw] (r) at (-1.5,4.5) {$r$};
		\node[circle,thick,draw] (0) at (0.5,1) {$i$};
		\node at (0.5,-2) {$\desc(i)$};
		\draw[->,ultra thick, dotted] (r)--(0);
		
		\draw [rounded corners=3mm]  (-3,-4)--(4,-4)--(0.5,2.5)--cycle;
				
		\draw[-, thick,dotted] (-6,4.5)--(-2,4.5);
		\draw[-, thick,dotted] (-5,2.2)--(0.5,2.2);
		\draw[-, thick,dotted] (-6,-4)--(-2.5,-4);
		
		\draw[-,ultra thick] (-6,4.5)--(-6,-4);
		\draw[-,ultra thick] (-6.5,4.5)--(-5.5,4.5);
		\draw[-,ultra thick] (-6.5,-4)--(-5.5,-4);	
		\node at (-6.5,-0.5) {$B$};
		
		\draw[-,ultra thick] (-4.5,2.2)--(-4.5,-4);
		\draw[-,ultra thick] (-5,2.2)--(-4,2.2);
		\draw[-,ultra thick] (-5,-4)--(-4,-4);
		
		\node at (-5,-0.5) {$B^\prime$};

		\end{tikzpicture}
	\caption{}
	\label{fig:tree1}
\end{subfigure}%
\begin{subfigure}[t]{0.5\textwidth}
\begin{tikzpicture}[scale=0.6]

		\node[circle,thick,draw,inner sep=3pt] (p1) at (3,1.5) {$P_1$};
		\node[circle,thick,draw,inner sep=3pt] (p2) at (5,1.5) {$P_2$};
		\node[circle,thick,draw,inner sep=3pt] (p3) at (7,1.5) {$P_3$};
		\node[circle,thick,draw,inner sep=3pt] (p4) at (9,1.5) {$P_4$};

		\draw [rounded corners=3mm,fill=gray!80]  (0,-4)--(2,-4)--(1,-1)--cycle;
		\draw [rounded corners=3mm,fill=gray!60]  (2.5,-4)--(4.5,-4)--(3.5,-1)--cycle;
		\draw [rounded corners=3mm,fill=gray!40]  (5,-4)--(7,-4)--(6,-1)--cycle;
		\draw [rounded corners=3mm,fill=gray!20]  (7.5,-4)--(9.5,-4)--(8.5,-1)--cycle;
		\draw [rounded corners=3mm,fill=gray!10]  (10,-4)--(12,-4)--(11,-1)--cycle;
		
		\node[fill=black, circle,inner sep=3pt] (1) at (1,-1.8) {};
		\node[fill=black, circle,inner sep=3pt] (2) at (3.5,-1.8) {};
		\node[fill=black, circle,inner sep=3pt] (3) at (6,-1.8) {};
		\node[fill=black, circle,inner sep=3pt] (4) at (8.5,-1.8) {};
		\node[fill=black, circle,inner sep=3pt] (5) at (11,-1.8) {};
		
		\node at (1,-2.5) {$j_1$};
		\node at (3.5,-2.5) {$j_2$};
		\node at (6,-2.5) {$j_3$};
		\node at (8.5,-2.5) {$j_4$};
		\node at (11,-2.5) {$j_5$};
		
		\draw[thick,dotted] (p1)--(1) (p2)--(3) (p3)--(4) (p4)--(5);
		\end{tikzpicture}
	\caption{}
	\label{fig:tree2}
\end{subfigure}
\caption{Figure \ref{fig:tree1} describes the high-level idea of the algorithm: at each subtree, the algorithm checks whether the descendants $\desc(i)$ can be assigned to the activity set $B^{\prime}$, when the whole tree only uses the activity set $B$. Figure \ref{fig:tree2} illustrates how each subtree is matched to a different activity set.}
\label{fig:tree}
\end{figure*}

\begin{claim}\label{claim:tree}
There exists an algorithm that given a partition $\calP=\{P_1,P_2,\ldots,P_{|\calP|}\}$ of $B^{\prime} \setminus \{a\}$, checks whether there exists a feasible assignment $\pi:N \rightarrow A$ such that the conditions $(${\rm i}$)$ to $(${\rm v}$)$ hold, and each activity set in $\calP$ is assigned to the players in $\desc(j)$ for some $j \in \ch(i)$ in such a way that for every $q \in [|\calP|]$ and $c \in [|\ch(i)|]$, $P_q$ is assigned to $\desc(j_c)$ only if the prefix $P_1,P_2,\ldots,P_{q-1}$ are assigned to the subtrees $\desc(j_1),\desc(j_2),\ldots, \desc(j_{c-1})$, finds one if it exists, and runs in time $O(p^2n^3)$ with respect to the number of activities $p$.
\end{claim}
\begin{proof}
We give a dynamic programming algorithm. Observe that a stable assignment can allocate only activity $a$ or nothing to players before allocating the activity set $P_1$ to some subtree. For convention, we thus denote $P_0$ by the empty activity set, which may be assigned to the first $c$ subtrees where the subtree $\desc(j_{c+1})$ is assigned to the activity set $P_1$. We will now determine for each $c \in [|\ch(i)|]$ and $q \in [0,|\calP|]$ whether the activity sets $P_0,P_1,\ldots,P_{q}$ can be assigned to the subtrees 
rooted at $j_1,j_2,\ldots,j_c$, and exactly $\ell$ players can be assigned 
to the activity $a$; we refer to this subproblem by $T[j_c,P_q,\ell]$.

%initialization
We initialize $T[j_1,P_q,\ell]$ to {\em true} if
\begin{itemize}
\item the empty activity set $P_0$ can be allocated to the first subtree $\desc(j_1)$, i.e., $q=0$, $\ell=0$, and $\NS[j_1,B,\emptyset,(a_{\emptyset},1),1]$ is {\em true}, and ($a=a_{\emptyset}$ or $j$ weakly prefers $(a_{\emptyset},1)$ to $(a,k+1)$); or
\item only the activity $a$ can be assigned to $\ell$ players in $\desc(j_1)$, i.e., $q=0$, $\ell\geq 1$, and $\NS[j_1,B,\{a\},(a,k),\ell]$ is {\em true};
\item only the activity set $P_1$ can be assigned to players in $\desc(j_1)$, i.e., $q=1$, $\ell=0$, and there exists an alternative $(b,x)\in P_1 
\times [n]\cup \{(a_{\emptyset},1)\}$ such that $\NS[j_1,B,P_1,(b,x),x]$ is {\em true}, $b=a_{\emptyset}$ or $i$ weakly prefers $(a,k)$ to $(b,x+1)$, and ($a=a_{\emptyset}$ or $j_1$ weakly prefers $(b,x)$ to $(a,k+1)$); or
\item $P_1$ can be assigned to players in $\desc(j_1)$, and activity $a$ can be assigned to $\ell$ players from $\desc(j_1)$, i.e., $q=1$, $\ell\geq 1$, and $\NS[j_1,B,P_1 \cup \{a\},(a,k),\ell]$ is {\em true}.
\end{itemize}
We set $T[j_1,P_q,\ell]$ to {\em false} otherwise.  

Then, we iterate through $j_1,j_2,\ldots,j_{|\ch(i)|}$ and $P_0,P_1,\ldots,P_{|\calP|}$, and update $T[j_c,P_q,\ell]$: for each $c \in [|\ch(i)|]$, for each $q\in [0,|\calP|]$, and for each $\ell \in [0,t]$, we set $T[j_c,P_q,\ell]$ to {\em true} if one of the following statements holds:
\begin{itemize}
\item $P_1, P_2, \ldots, P_{q}$ can be assigned to the subtrees $\desc(j_1), \desc(j_2),\ldots, \desc(j_{c-1})$ with $\ell$ players from the subtrees being assigned to the activity $a$, and the void activity can be assigned to the subtree $\desc(j_c)$, i.e., both  $T[j_{c-1},P_q,\ell]$ and $\NS[j_c,B,\emptyset,(a_{\emptyset},1),1]$ are {\em true}, 
and $a=a_{\emptyset}$ or player $j_c$ weakly prefers $(a_{\emptyset},1)$ to $(a,k+1)$;
\item $P_1, P_2, \ldots, P_{q}$ can be assigned to the subtrees $\desc(j_1), \desc(j_2),\ldots, \desc(j_{c-1})$ while the activity $a$ is assigned to $x$ players from the subtrees $\desc(j_1), \desc(j_2),\ldots, \desc(j_{c-1})$ and to $\ell - x$ players from the subtree $\desc(j_c)$, i.e., 
$\ell \geq 2$, and there exists an $x \in [\ell-1]$ such that $T[j_{c-1},P_q,\ell-x]$ and $\NS[j_c,B,\{a\},(a,k),x]$ are {\em true}; 
\item $P_1, P_2, \ldots, P_{q-1}$ can be assigned to the subtrees $\desc(j_1), \desc(j_2),\ldots, \desc(j_{c-1})$ with $\ell$ players from the subtrees being assigned to the activity $a$, and only the activity set $P_q$ can be assigned to the subtree $\desc(j_c)$, i.e., $T[j_{c-1},P_{q-1},\ell]$ is {\em true}, and there exists an alternative $(b,x)\in P_q 
\times [n]\cup \{(a_{\emptyset},1)\}$ such that $\NS[j_c,B,P_q,(b,x),x]$ is {\em true}, ($b=a_{\emptyset}$ or $i$ weakly prefers $(a,k)$ to $(b,x+1)$), and ($a=a_{\emptyset}$ or $j_c$ weakly prefers $(b,x)$ to $(a,k+1)$);
\item $P_1, P_2, \ldots, P_{q-1}$ can be assigned to the subtrees $\desc(j_1), \desc(j_2),\ldots, \desc(j_{c-1})$ and $P_q$ can be assigned to the subtree $\desc(j_c)$ while the activity $a$ is assigned to $x$ players from the subtrees $\desc(j_1), \desc(j_2),\ldots, \desc(j_{c-1})$ and to $\ell - x$ players from the subtree $\desc(j_c)$, i.e., $\ell \geq 2$, and there exists an $x \in [\ell-1]$ such that both $T[j_{c-1},P_{q-1},\ell-x]$ and $\NS[j_c,B,P_q \cup \{a\},(a,k),x]$ are {\em true}.
\end{itemize}
We set $T[j_c,P_q,\ell]$ to {\em false} otherwise. 

Clearly, a desired assignment exists if and only if $T[j_{|\ch(i)|}, P_{|\calP|}, t-1]$ is {\em true}. The size of the dynamic programming table is at most $pn^2$ and each entry can be filled in $O(pn)$ time. This prove the claim. 
\end{proof}

The size of the dynamic programming table $\NS[i,B,B^{\prime},(a,k),t]$ is at most $4^p p n^3$. For each entry, we consider at most $(p+1)^p$ partitions of $B^{\prime}\setminus \{a\}$, and for each partition at most $p!$ permutations; further, for each permutation we have shown that there exists an $O(p^2n^3)$ time algorithm that checks the existence of a stable assignment compatible with it. Hence we conclude that our problem for trees can be solved in $O(4^p (p+1)^p p! p^3 n^6)$ time.

Now, if $(N, L)$ is a forest, we can combine the FPT algorithm described above with the algorithm for graphs with small connected components in the proof of Theorem~\ref{thm:FPT:smallcomponents:NSIS}.
The running time of the latter algorithm is a product of the time required to find a Nash stable assignment
for a single connected component and the time required to combine solutions for different components;
As we have seen, the former is $O(p^{c+1}c^2)$, where $c$ is the maximum component size and the latter is $O(2^pn)$.
In our case, each connected component is a tree, so instead of the $O(p^{c+1}c^2)$ algorithm for general graphs
we can use our FPT algorithm for trees. This shows that our problem is in FPT for arbitrary forests.
\end{proof}

The proof for individual stability is similar and can be found in the appendix.

\begin{theorem}\label{thm:FPT:tree:IS}
The problem of deciding whether an instance of $\gGASP$ with $|A^*|=p$ 
whose underlying social network $(N, L)$ is acyclic has an individually stable feasible assignment is in {\sc FPT} with respect to $p$.
\end{theorem}

Core stability turns out to be more 
computationally challenging than Nash stability and individual stability
when the number of activities is small and the graph is a star: we will now show that 
core stable assignments are hard to find even 
if there are only two activities 
and the underlying graph is a star (and thus one cannot expect an FPT result
with respect to the number of activities for this setting).
Later, we will see that this hardness result can be extended to the case where $(N, L)$ is a clique, 
i.e., to standard \GASP, thereby solving a problem left open in the work
of \citet{Darmann2015}.

\begin{theorem}\label{thm:NP:star:core}
It is {\em NP}-complete to determine whether an instance of \gGASP\ 
admits a core stable assignment even when the underlying 
graph is a star and the number of non-void activities is~$2$.
\end{theorem}
\begin{proof}
Our problem is in NP by Proposition~\ref{prop:in-core}.
To establish NP-hardness, we reduce from the NP-complete problem {\sc Exact Cover by $3$-Sets} (X3C)~\cite{gj}. 
Given a finite set $V=\{v_1,v_2,\ldots,v_{3k}\}$ and a collection $\calS=\{S_1,S_2,\ldots,S_{m}\}$ of $3$-element subsets of $V$, X3C asks whether $\calS$ admits an {\em exact cover} $\calS^{\prime}\subseteq \calS$, i.e., $|\calS^{\prime}|=k$ and $V=\bigcup_{S_i \in \calS^{\prime}}S_i$.

Given an instance of X3C, we construct an instance of \gGASP\ as follows.
We define the set of activities to be $A^{*}=\{a,b\}$. 
We introduce a center player $c$, two players $x_1$ and $x_2$, and a player $S_j$ for each $S_j \in \calS$.

\emph{Idea.}
We will construct the dummies of each $v_i \in V$ and the preferences of players as follows:
\begin{enumerate}
\item The players $c$, $x_1$, and $x_2$ form an empty core instance; thus, in order to stabilize them, the center $c$ needs to be assigned to activity $a$, together with $k$ players from $\calS$, and $x_2$ must be assigned to activity $b$. 
\item If the dummies of $v_i$ and the players corresponding to the sets including $v_i$ (denoted by $\calS(v_i)$) are assigned to the void activity, then these players together with the center $c$ and $x_2$ will form a blocking coalition of size $\beta(v_i)$ and deviate to activity $b$.
\end{enumerate}

\emph{Construction details.}
For each $v_i \in V$, we let $\calS(v_i)=\{\, S_j \mid v_i \in S_j \in \calS \,\}$ and $\beta(v_i)=i+3k+1$ 
and create a set $\Dummy(v_i)=\{d^{(1)}_i,d^{(2)}_i,\ldots,d^{(\beta(v_i)-|\calS(v_i)|-2)}_i\}$ of dummy players. 
We attach each of the dummies to the center $c$. Intuitively, for each $i\in[3k]$
the number $\beta(v_i)$ is the target coalition size when all players 
in $\Dummy(v_i)$ are engaged in activity $b$, together with $c$, $x_2$, and the players in $\calS(v_i)$. Note that the values $\beta(v_i)$ are at least $3$ and distinct over $i\in[3k]$. We then attach to the center, the players $x_1$, $x_2$, all the players in $\calS$, and all the dummy players. The illustration of the graph is presented in Figure \ref{fig:star:core}. 

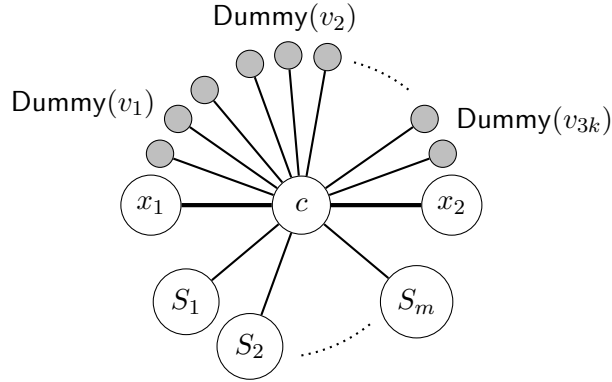
\begin{figure*}[htb]
\centering
%star
\begin{tikzpicture}[scale=1, transform shape]
	\def \radius {2cm}
	\node[draw, circle, minimum size=22pt](c) at (0:0) {$c$};
	\node[draw, circle](x1) at ({180}:\radius) {$x_1$};
	\node[draw, circle](x2) at ({0}:\radius) {$x_2$};
	
	\node[draw, circle](node1) at ({220}:\radius) {$S_{1}$};
	\node[draw, circle](node2) at ({250}:\radius) {$S_{2}$};
	\node[draw, circle](node3) at ({320}:\radius) {$S_{m}$};
 	
	\node[draw, circle,fill=gray!50](node4) at ({160}:\radius) {};
	\node[draw, circle,fill=gray!50](node5) at ({145}:\radius) {};
	\node[draw, circle,fill=gray!50](node6) at ({130}:\radius) {};
	\node at ({155}:3.2) {$\Dummy(v_1)$};
	
	\node[draw, circle,fill=gray!50](node7) at ({110}:\radius) {};
	\node[draw, circle,fill=gray!50](node8) at ({95}:\radius) {};
	\node[draw, circle,fill=gray!50](node9) at ({80}:\radius) {};
	\node at ({95}:2.5) {$\Dummy(v_2)$};
	
	\draw [dotted,thick] (0.7,1.9) arc [radius=2, start angle=75, end angle= 50];
	\node[draw, circle,fill=gray!50](node10) at ({35}:\radius) {};
	\node[draw, circle,fill=gray!50](node11) at ({20}:\radius) {};
	\node at ({20}:3.3) {$\Dummy(v_{3k})$};
    
 	\draw[-, >=latex,  thick] (c)--(node1);
	\draw[-, >=latex,  thick] (c)--(node2);
	\draw[-, >=latex,  thick] (c)--(node3);	
	\draw[-, >=latex,  thick] (c)--(node4);
	\draw[-, >=latex,  thick] (c)--(node5);
	\draw[-, >=latex,  thick] (c)--(node6);
	\draw[-, >=latex,  thick] (c)--(node7);
	\draw[-, >=latex,  thick] (c)--(node8);
	\draw[-, >=latex,  thick] (c)--(node9);
	\draw[-, >=latex,  thick] (c)--(node10);
	\draw[-, >=latex,  thick] (c)--(node11);

	\draw[-, >=latex,ultra thick] (c)--(x1);
	\draw[-, >=latex,ultra thick] (c)--(x2);
	\draw [dotted,thick] (0,-\radius) arc [radius=2, start angle=280, end angle= 310];
\end{tikzpicture}
\caption{
The graph constructed in the hardness proof for Theorem \ref{thm:NP:star:core}. 
\label{fig:star:core}
}
\end{figure*}

%preferences
Now the players' preferences are defined as follows. For each $S_j \in \calS$, we let 
$B_j=\{b\}\times \{\, \beta(v_i)\mid v_i \in S_j \,\}$; also, set $B = \bigcup_{S_j\in \calS}B_j$.
The preferences of each player $S_j \in \calS$ are given by
\begin{align*}
&S_j:~(a,k+1)\succ B_j \succ (a_{\emptyset},1).
\end{align*}
For each $i\in[3k]$ the dummy players in  $\Dummy(v_i)$ only approve the alternative $(b,\beta(v_i))$.

Finally, the preferences of  $x_1$, $c$, and $x_2$ are given by
\begin{align*}
&x_{1}:~ (b,2)  \succ (a,3) \succ (a_{\emptyset},1)\\
&c:~(a,2) \succ (b,2) \succ (a,3) \succ B \succ (a,k+1)\succ (a_{\emptyset},1)\\
&x_{2}:~(a,3) \succ B \succ (b,1)\succ (a,2) \succ (a_{\emptyset},1).
\end{align*} 
Note that the preferences of $c$, $x_1$ and $x_2$, when restricted to $A\times[1, 2, 3]$,
form an instance of \gGASP\ with an empty core (Example~\ref{ex:core:empty}).

\emph{Correctness.}
We will show that $(V, \calS)$ admits an exact cover if and only if there exists a core stable feasible assignment.

Let $\calS^{\prime}$ be an exact cover in $(V, \calS)$.
Then, we construct a feasible assignment 
$\pi$ by assigning activity $a$ to the player $c$ and the players 
$S_j \in \calS^{\prime}$, assigning $b$ to the player $x_2$, and assigning the void activity to the remaining players.
Clearly, $\pi$ is individually rational since $|\calS^{\prime}|=k$ and hence $|\pi^a|=k+1$. Further, notice that no connected 
subset $T$ together with activity $a$ strongly blocks $\pi$: by the definition of a blocking coalition, every such subset has to contain the players in $\calS^{\prime}$, who  
are currently enjoying one of their top alternatives. 
It remains to show that no connected subset $T$ together with activity $b$ strongly blocks $\pi$. Suppose 
towards a contradiction that such a subset $T$ exists; 
as $x_2 \in T$, it must be the case 
that $|T|=\beta(v_i)$ for some $v_i \in V$ and hence $T$ consists of agents who approve $(b, \beta(v_i))$, i.e., 
$T=\{c,x_2\}\cup \calS(v_i) \cup \Dummy(v_i)$ for some $v_i\in V$. However, since $\calS^{\prime}$ is an exact cover, 
there is a set $S_j\in \calS^{\prime}\cap \calS(v_i)$ with $\pi(S_j)=a$, and this agent prefers $(a, k+1)$
to $(b, \beta(v_i))$, a contradiction.
Hence, $\pi$ is core stable.

Conversely, suppose that there exists a core stable feasible assignment $\pi$ and let 
$\calS^{\prime}=\{\, S_j \in \calS \mid \pi(S_j)=a \,\}$. 

We will first argue that $\pi(c)=a$. Indeed, if $\pi(c)=b$ and the only other agent to engage in $b$ is $x_1$, then $\pi^a=\emptyset$ since no player approves $(a,1)$; thus, the players $c$ and $x_2$ can deviate to $a$. If $\pi(c)=b$ and $|\pi^b|=\beta(v_i)$ for some $v_i\in V$ then $\pi^a=\emptyset$ and $c$, $x_1$ and $x_2$ can deviate to $a$. If $\pi(c)=a_{\emptyset}$, then $\pi^a=\emptyset$ and the players $c$, $x_1$, and $x_2$ would deviate to $a$. 
It follows that $\pi(c)=a$. 
Now, if $\pi^a=\{c, x_2\}$ then $\pi^b=\emptyset$ and agent $x_2$ can deviate to $b$. Similarly, if $\pi^a=\{c, x_1, x_2\}$ then $\pi^b=\emptyset$, and $c$ and $x_1$ can deviate to $b$. If follows that $|\pi^a|=k+1$ and hence $|\calS^{\prime}|= k$. Also, since there is no coalition assigned to an alternative in $B$, $x_2$ must be engaged alone in activity $b$; otherwise, the player can deviate to $b$.

Now, if $\calS^{\prime}$ is not an exact cover for $(V, \calS)$,
there exists an element $v_i\in V$ such that no player in $\calS(v_i)$ is assigned to $a$; that is, all players in $\calS(v_i)$ are assigned to the void activity.
Then the coalition $\{c, x_2\}\cup \calS(v_i) \cup \Dummy(v_i)$ 
together with the activity $b$ strongly blocks $\pi$, contradicting the stability of $\pi$.
Thus, $\calS^\prime$ is an exact cover for $(V, \calS)$.
\end{proof}

The hardness result in Theorem~\ref{thm:NP:star:core}
immediately generalizes to instances of \gGASP\
with more than two activities: we can modify the construction in our hardness reduction
by introducing additional activities that no player wants to engage in.
In contrast, checking the existence of core stable assignments in
\gGASP\ is easy if $|A^*|=1$, irrespective of the structure of the social network. 

\begin{proposition}\label{prop:core-single}
Every instance of \gGASP\ with $A=\{a, a_\emptyset\}$ 
admits a core stable assignment; moreover, 
such an assignment can be computed in polynomial time.
\end{proposition}
\begin{proof}
Consider an instance $(N, (\succeq_i))_{i\in N}, A, L)$ of \gGASP\ with $A=\{a, a_\emptyset\}$, 
and let $n=|N|$. For each $s\in[n]$, let $S_s$
be the set of all players who weakly prefer $(a, s)$ to $(a_\emptyset, 1)$.
If for each $s\in[n]$ the largest connected component of $S_s$ 
with respect to $(N, L)$ contains 
fewer than $s$ agents, then no outcome in which a non-empty set 
of agents engages in $a$ is individually rational, whereas the assignment
$\pi$ given by $\pi(i)=a_\emptyset$ for all $i\in N$ is core stable.
Otherwise, consider the largest value of $s$ such that $S_s$
has a connected component of size at least $s$. Find a connected subset of $S_s$
of size exactly $s$, and assign the agents in this subset to $a$;
assign the remaining agents to $a_\emptyset$. To see that this assignment is core stable, 
note that for every deviating coalition $S$ we would have $|S|>s$, and hence 
such a coalition is either disconnected or some players in $S$ prefer $(a_\emptyset, 1)$
to $(a, |S|)$.  
\end{proof}

It is not clear if the problem of finding core stable assignments in \gGASP\ 
is in FPT with respect to the number of activities
when the underlying graph is a path. However, we can place this problem in XP
with respect to this parameter. 
In fact, we have the following more general result 
for graphs with few connected coalitions (see the work of \citet{Elkind2014}
for insights on the structure of such graphs).

\begin{proposition}\label{prop:core-xp}
Given an instance $(N, (\succeq_i))_{i\in N}, A, L)$
of \gGASP\ with $|N|=n$, $|A|=p+1$, 
such that the number of non-empty connected subsets of $(N, L)$ 
is $\kappa$, we can decide in time $O(\kappa^p)\cdot\poly(n, p)$
whether this instance admits a core stable assignment,
and find one such assignment if it exists.
\end{proposition}
\begin{proof}
Let $C_1, \dots, C_\kappa$ be the list of all non-empty connected subsets
of $(N, L)$. Since each assignment of players to activities
has to assign a connected (possibly empty) subset of players to each activity, 
we can bound the number of possible assignments by $(\kappa+1)^p$:
each of the $p$ non-void activities is assigned to a coalition in our list or to no player at all, 
and the remaining players are assigned the void activity.
We can then generate all these assignments one by one and check
if any of them is core stable; by Proposition~\ref{prop:in-core}, 
the stability check can be performed in time polynomial in $n$ and $p$.
\end{proof}

If the social network $(N, L)$ is a path, it has $O(|N|^2)$ connected subsets.
Thus, we obtain the following corollary.

\begin{cor}\label{cor:core-xp-path}
The problem of deciding whether a given instance
of \gGASP\ whose underlying graph is a path admits
a core stable assignment is in {\em XP} with respect to the number of activities. 
\end{cor}

\subsection{Cliques}
It is unlikely that our FPT results can be extended to even cliques: 
our next result shows that the problem of finding a Nash stable 
outcome is W[1]-hard with respect to the number of activities
even for classic \GASP, i.e., when the social network imposes no constraints
on possible coalitions and players have approval preferences.

\begin{theorem}\label{thm:W1:activities:clique:NS}
The problem of determining whether an instance of \gGASP\ 
admits a Nash stable is {\em W[1]}-hard with respect to the number of activities,
even if the underlying graph $G=(N, L)$ is a clique.
\end{theorem}
\begin{proof}
%We present a proof for Nash stability. The proof for individual stability can be found in the appendix.
We will provide a parameterized reduction from {\sc Clique} on regular graphs. Given an undirected graph $G=(V,E)$ and a positive integer $k$, {\sc Clique} asks if $G$ admits a clique of size $k$. The problem is known to be W[1]-hard even for {\em regular graphs}, i.e., graphs where each vertex has the same degree
(see e.g. Theorem 13.4 in the book of \citet{Cygan2015}).

Given a regular graph $G=(V,E)$ and a positive integer $k$, where $|V|=n$, $|E|=m$, and each vertex of $G$ has degree $\delta \geq k-1$, 
we create an instance of \gGASP\ whose underlying graph is a clique, as follows. 
%Vertex and Edge Activities
We create one {\em vertex activity} $a_i$ for each $i \in [k]$, one {\em edge activity} $b_j$ for each $j \in [k(k-1)/2]$, and two other activities $c$ and $d$.
The set of non-void activities is given by 
\[
A^*=\{\, a_i \mid i \in [k] \,\} \cup \{\, b_j \mid j \in  [k(k-1)/2]\,\} \cup \{c,d\}.
\]
%Vertex and Edge Players
For each $v \in V$, we create one {\em vertex player} $v$, and for each edge $e=\{u,v\} \in E$, we create two {\em edge players} 
$e_{uv}$ and $e_{vu}$.

\emph{Idea.} We will create dummies of vertices and edges in such a way that a stable assignment has the following properties:
\begin{enumerate}
\item if a vertex player $v$ is assigned to a vertex activity, it forms a coalition of size $\alpha(v)$ that consists of the player $v$, its dummies, and $\delta-k+1$ edge players incident to the vertex; and
\item if an edge player $e_{vu}$ is assigned to an edge activity, it forms a coalition of size $\beta(e)$ that consists of the edge player $e_{uv}$, $e_{vu}$, and the dummies of the edge $\{u,v\}$.
\end{enumerate}

%Vertex and Edge dummy players
\emph{Construction details.} Let $P=\{2j+n\mid j\in[n]\}$, and let $\alpha:V\to P$ be a bijection that assigns a 
distinct number in $P$ to each vertex $v\in V$. Note that $u\neq v$
implies that $\alpha(u)$ and $\alpha(v)$ differ by at least $2$, i.e., $\alpha(u)+1 \neq \alpha(v)$ and $\alpha(v)+1 \neq \alpha(u)$.
Similarly, let $Q=\{2j\mid j\in[m]\}$ and 
let $\beta:E \rightarrow Q$ be a bijection that assigns a distinct number
in $Q$ to each edge $e\in E$.
For each $v\in V$ we construct a set $\Dummy(v)$ of $\alpha(v)-\delta+k-2$ {\em dummy vertex players}. 
Similarly, for each $e\in E$ we construct a set $\Dummy(e)$ of $\beta(e)-2$ {\em dummy edge players}.
Lastly, we create one stalker $s$. 
The set of players is given by 
\[
N=V \cup \bigcup_{v \in V}\Dummy(v) \cup E \cup \bigcup_{e \in E}\Dummy(e) \cup \{s\}.
\]

%Preferences for vertex and edge players
Now the preferences are given as follows.
\begin{itemize}
\item Each vertex player $v \in V$ approves the alternatives $(c,n-k)$ and $(d,1)$ as well as each alternative $(a_i,\alpha(v))$ where $i\in [k]$.
\item Each dummy in $\Dummy(v)$ of a vertex player $v$ approves the alternatives $(a_i,\alpha(v))$ and $(a_i,\alpha(v)+1)$ where $i\in [k]$.
\item Each edge player $e_{vu}$ approves the alternatives $(a_i,\alpha(v))$ and $(a_i,\alpha(v)+1)$ where $i\in [k]$ as well as the alternatives
$(b_j,\beta(e))$ where $j \in [k(k-1)/2]$.
\item  The dummies in $\Dummy(e)$ only approve the alternatives $(b_j,\beta(e))$ where $j \in [m]$.
\item The stalker $s$ approves the alternative $(d, 2)$.
\end{itemize}
All of these players are indifferent among all alternatives they approve. 
Finally, we take the underlying social network to be a complete graph.
Note that the number of activities depends on $k$, but not on $n$, and the size of our instance of \gGASP\ is bounded by $O(n^2+m^2)$.

\emph{Correctness.}
We will now argue that the graph $G$ contains a clique of size $k$ if and only if there exists a Nash stable assignment
for our instance of \gGASP. 

Suppose that $G$ contains a clique $S$ of size $k$. We construct an assignment $\pi$ as follows. 
We assign the activity $c$ to all vertex players who do not belong to $S$. The size of this coalition is $n-k$. 
We then assign the vertex activities to the rest of vertex players and its dummies:
We take some bijection $\eta$ between $S$ and $[k]$, and for each $v\in V$
we form a coalition of size $\alpha(v)$ that engages in $a_{\eta(v)}$: this coalition consists
of $v$, all players in $\Dummy(v)$, and all edge players $e_{vu}$ such that $u\not\in S$. 
Lastly, we assign the edge activities to the edge players both of whose endpoints are in $S$: we establish a bijection $\xi$ between the edge set $\{\,e=\{u,v\}\in E\mid u,v \in S\,\}$ 
and $[k(k-1)/2]$, and assign the activity $b_{\xi(e)}$ 
to the edge players $e_{uv}$ and $e_{vu}$ as well as to all players in $\Dummy(e)$. 
All other players, namely the stalker $s$ as well as the remaining edge players and dummy players, are assigned the void activity.
We will now argue that the resulting assignment $\pi$ is Nash stable. 

%assigned vertex player
Clearly, no player assigned to an activity $a_i$, $b_j$, or $c$
wishes to deviate. 
Now, consider a dummy of some vertex $v$ that is assigned to the void activity.
By construction, we have $\pi(v)=c$. The dummy only wants to join a coalition if there is a coalition which engages in an activity $a_i$
and whose size is $\alpha(v)-1$ or $\alpha(v)$; however, no such coalition
exists. 
Similarly, consider an edge player $e_{vu}$ with $\pi(e_{vu})=a_\emptyset$.
We have $v\not\in S$, and therefore $e_{vu}$ does not want to join
any of the existing coalitions because they have the wrong sizes; the same argument applies to all the dummies of~{$e$}.
Further, the stalker $s$ does not want to deviate since there is no coalition of size 
$1$ that engages in an activity $d$. 
Hence, $\pi$ is Nash stable. 

Conversely, suppose that there exists a Nash stable feasible assignment $\pi$. 
By individual rationality of $\pi$, every player engages either in the void activity or in an approved alternative.
Notice that if a vertex player $v$ were assigned to the activity $d$, the stalker $s$ would have an NS-deviation to $d$; thus, no player is assigned to $d$. If a vertex player $v$ were assigned to the void activity, then $v$ would have an NS-deviation to activity $d$.
Thus, $\pi$ cannot allocate vertex players to either $d$ or the void activity.
This means that $\pi$ must allocate the activity $c$ to $n-k$ vertex players, and allocate the $k$ vertex activities $a_i$ to the remaining $k$ vertex players $v$, each together with $\alpha(v)-1$ other players.
%%%%%
Now, let $S=\{v\in V\mid \pi(v) = a_i\text{ for some }i \in [k]\}$.
By construction, $|S|=k$. We will show that $S$ is a clique. That is, we prove the following claim. 

\begin{claim}\label{claim:clique}
For any $v \in S$, $v$ is adjacent to every other vertex in $S$, i.e., there are $k-1$ edge players $e_{vu}$ with $u \in S$.
\end{claim}
\begin{proof}
Consider a player $v\in S$, and let $\pi(v)=a_i$. Observe that by individual rationality $|\pi^{a_i}|=\alpha(v)$; thus, all players in $\Dummy(v)$ are assigned to $a_i$ since otherwise they could deviate to $a_i$.
The only other players who approve $(a_i, \alpha(v))$ are edge players $e_{vu}$.
Thus, $\delta-k+1$ such players must be assigned to~{$a_i$}, and the remaining 
$k-1$ of these players must be assigned to some activity $b_j$,
since otherwise they would deviate to~{$a_i$}. 
It remains to show that the $k-1$ edge players assigned to some $b_j$ are incident to vertices in $S$. For each $v\in S$, consider the set of players 
$E_{\pi, v} = \{e_{vu} \mid v\in S,  \pi(e_{vu})=b_j~\mbox{for some}~j \in [k(k-1)/2]\}$. We have argued that $|E_{\pi, v}| =k-1$ for each $v$.
Thus, $|\bigcup_{v\in S}E_{\pi, v}|=k(k-1)$. 

Now, let $E_{\pi}$ be the set of all edge players who are assigned to edge activities $b_j$. Clearly $\bigcup_{v\in S}E_{\pi, v} \subseteq E_{\pi}$. Note that by individual rationality, each edge activity $b_j$ can be assigned to at most two edge players simultaneously (namely, to the two edge players $e_{uv}, e_{vu}$ corresponding to edge $e$). 
Since there are $k(k-1)/2$ edge activities, it follows that $|E_{\pi}| \le k(k-1)$. Therefore, we must have
\begin{equation}\label{eq}
\bigcup_{v\in S}E_{\pi, v}=E_{\pi}. 
\end{equation}

Now, consider an edge player 
$e_{vu}$ with $v \in S$ and $\pi(e_{vu})=b_j$ for some $j \in [k(k-1)/2]$.
By individual rationality
we have $|\pi^{b_j}|=\beta(\{u, v\})$ and hence $\pi^{b_j}$ consists of
$e_{vu}$, $e_{uv}$ and all dummies of the edge $\{u, v\}$. 
By \eqref{eq}, we have $e_{uv} \in \bigcup_{v\in S}E_{\pi, v} \subseteq E_{\pi}$, and so we see that $u \in S$.
\end{proof}
Hence, we conclude that $S$ is a clique of size $k$.
\end{proof}

We note that if players have approval preferences, individually and core stable outcomes can be computed in polynomial time for any network structure: one can repeatedly find activity $a$ and a maximal unanimous coalition $S$ all of whose members approve $a$, assign $S$ to $a$, and remove the activity a and the players in $S$ from consideration. However, if we allow full preference orders, deciding the existence of an individually stable outcome is W[1]-hard with respect to the number of activities for standard \GASP\, i.e., the graph is a clique. The proof can be found in the appendix.

\begin{theorem}\label{thm:W1:activities:clique:IS}
The problem of determining whether an instance of \gGASP\ 
admits an individually stable assignment 
is {\em W[1]}-hard with respect to the number of activities,
even if the underlying graph $G=(N, L)$ is a clique.
\end{theorem}

On the positive side, for \gGASP s on cliques, a Nash stable assignment
turns out be computable in polynomial time if the number of activities is a constant. However, it is open if this result
can be extended to general \gGASP s. 

\begin{theorem}\label{thm:XP:NS}
There exist an algorithm that, given an instance $(N, A, (\succeq_i)_{i\in N}, L)$
of \gGASP\ with $|N|=n$, $|A|=p+1$ such that $(N, L)$ is a clique, 
determines whether it admits a Nash stable 
assignment in time $(n+1)^{p+1}O(np(n+p))$.
\end{theorem}
\begin{proof}
For every mapping $f:A \rightarrow [0, n]$ where $\sum_{a \in A}f(a)=n$, we will check if there is a Nash stable assignment $\pi$ such that 
$|\pi^a|=f(a)$ for each $a \in A^*$ and $|\{\, i \in N \mid \pi(i)=a_{\emptyset} \,\} |=f(a_{\emptyset})$. There are at most $(n+1)^{p+1}$ such mappings; hence, it remains to show that each 
check will take at most $\poly(n)$ steps.

Fix a mapping $f:A \rightarrow [0, n]$ with $\sum_{a \in A}f(a)=n$. We construct an instance of the network flow problem as follows. We introduce a source $s$, a sink $t$, a node $i$ for each player $i \in N$, and a node $a$ for each activity $a \in A$. 
We create an arc with unit capacity from the source $s$ to each player, and an arc with capacity $f(a)$ from node $a \in A$ 
to the sink $t$. Then, for each $i\in N$
we create an arc of unit capacity from player $i$ to an activity $a \in A$ if and only if 
\begin{itemize}
\item $a \neq a_{\emptyset}$, and $i$ weakly prefers $(a,f(a))$ to $(a_\emptyset, 1)$ and to all pairs of the form $(b, f(b)+1)$, where $b\in A^* \setminus \{a\}$; or 
\item $a = a_{\emptyset}$, and $i$ weakly prefers $(a_{\emptyset},1)$ to all pairs of the form $(b, f(b)+1)$, where $b\in A^*$.
\end{itemize}
It can be easily verified that an integral flow of size $n$ in this network corresponds to a 
Nash stable assignment where exactly $f(a)$ players are engaged in each activity $a \in A$. It remains to note 
that one can check in $O(|V||E|)$ time whether a given network $(V,E)$ admits a flow of a given size \citep{Orlin2013}.
\end{proof}

%The results for individual stability presented so far indicate that from the complexity perspective it is very similar to Nash stability. However, it is not clear if the XP algorithm presented in Theorem~\ref{thm:XP:NS} extends to individual stability. The difficulty is that, to determine whether an agent $i$ has an IS deviation to an activity $a$, it is not sufficient to know how many players engage in $a$: knowing their preferences is important to decide whether $i$'s deviation will be vetoed by one of the players currently assigned to $a$.

A similar argument applies to individual stability: finding an individually stable assignment can be done in polynomial time if the number of activities is a constant.
\begin{theorem}\label{thm:XP:IS}
There exist an algorithm that, given an instance $(N, A, (\succeq_i)_{i\in N}, L)$
of \gGASP\ with $|N|=n$, $|A|=p+1$ such that $(N, L)$ is a clique, 
determines whether it admits an individually stable 
assignment in time $2^p(n+1)^{p+1}O(np(n+p))$.
\end{theorem}
\begin{proof}
For every mapping $f:A \rightarrow [0, n]$ and $g:A^* \rightarrow \{0, 1\}$, we will check if there is an individually stable assignment such that 
\begin{itemize}
\item[$(${\rm i}$)$] for each $a \in A^*$, $|\pi^a|=f(a)$, and $|\{\, i \in N \mid \pi(i)=a_{\emptyset} \,\} |=f(a_{\emptyset})$;
\item[$(${\rm ii}$)$] for each $a \in A^*$ with $g(a)=1$, $\pi^a$ contains a player $i$ who does not want to increase the size of a group, that is, who strictly prefers $(a, f(a))$ to $(a, f(a) + 1)$; and
\item[$(${\rm iii}$)$] for each $a \in A^*$ with $g(a)=0$, no player wants to deviate to $\pi^a$.
\end{itemize}
There are at most $2^p(n+1)^{p+1} $ such combinations of mappings. We will show that each check will take at most $\poly(n)$ steps.

Fix a mapping $f:A^* \rightarrow [0, n]$ and $g:A^* \rightarrow \{0, 1\}$.  
Let $A_0=\{\, a \in A^* \mid g(a)=0\,\}$ and $A_1=\{\, a \in A^* \mid g(a)=1\,\}$.

We construct an instance of the network flow problem as follows. We introduce a source~$s$, a sink $t$, a node $i$ for each player $i \in N$, a node $a$ for each $a \in A_0 \cup \{a_{\emptyset}\}$, and two nodes $a$ and $x(a)$ for each activity $a \in A_1$. 
We create an arc with unit capacity from the source $s$ to each player, and an arc with capacity $f(a)$ from node $a \in A_0\cup \{a_{\emptyset}\}$ to the sink $t$. For each $a \in A_1$ with $f(a) \geq 1$, we create an arc with capacity $f(a)-1$ from node $a$ to the sink $t$, and an arc with unit capacity from node $x(a)$ to the sink $t$. 
Then, for each $i\in N$ 
\begin{itemize}
\item we create an arc of unit capacity from player $i$ to a node $a \in A^*$ if and only if $i$ weakly prefers $(a,f(a))$ to $(a_\emptyset, 1)$ and to all pairs of the form $(b, f(b)+1)$, where $b\in A_0 \setminus \{a\}$; 
\item we create an arc of unit capacity from player $i$ to a node $x(a)$ where $a \in A_1$ if and only if $i$ strictly prefers $(a,f(a))$ to $(a,f(a)+1)$, and $i$ weakly prefers $(a,f(a))$ to $(a_\emptyset, 1)$ and to all pairs of the form $(b, f(b)+1)$, where $b\in A_0 \setminus \{a\}$;
\item finally, we create an arc of unit capacity from player $i$ to the void activity $a_{\emptyset}$ if and only if $i$ weakly prefers $(a_{\emptyset},1)$ to all pairs of the form $(b, f(b)+1)$, where $b\in A_0$.
\end{itemize}

It can be easily verified that an integral flow of size $n$ in this network corresponds to an individually stable assignment satisfying the conditions $(${\rm i}$)$ - $(${\rm iii}$)$. Again, one can check in polynomial time whether a given network admits a flow of a given size \citep{Orlin2013}, and hence our theorem follows.
\end{proof}

For core stability, Theorem~\ref{thm:NP:star:core} can be extended from stars to cliques;
however, our proof for cliques relies on having at least four non-void activities.
It remains an interesting open problem whether core stable outcomes of \gGASP s on cliques
can be found efficiently if the number of activities does not exceed $3$.
We conjecture that the answer is `no', i.e., the problem of computing core stable outcomes
remains NP-hard for $|A^*|=2, 3$. The proof for the theorem below can be found in the appendix.

\begin{theorem}\label{thm:NP:clique:core}
It is {\em NP}-complete to determine whether an instance of \gGASP\ admits a core stable assignment even if the underlying 
graph is a clique and the number of activities is $4$.
\end{theorem}

\section{Few Players}
\noindent
In the previous sections, the parameter that we focused on
was the number of activities~$p$.
Although we expect this parameter to be small in many realistic settings, there
are also situations where players can choose from a large variety of possible activities. 
Assume, for instance, that a small number of scientists are gathering at the university 
of a large city which offers plenty of options for social, cultural, or sports activities.
Usually, the number of options is much larger than the number of people attending.
It is, thus, natural to ask if stability-related problems for \gGASP{}
are tractable in the number of players~$n$ is small.
In this section, we give a negative answer by showing that even classical \GASP{}
remains W[1]-hard with respect to the number of players~$n$.

We first observe that for all stability concepts considered in this paper,
the problem of finding a stable feasible assignment is in XP with respect to $n$:
we can simply guess the activity of each player (there are at most~$(p+1)^n$ possible guesses) 
and check whether the resulting assignment is feasible and stable.
\begin{observation}\label{obs:XPplayers}
The problem of deciding whether a given instance
of \gGASP\ admits a core stable, Nash stable, or individually stable assignment
is in {\em XP} with respect to the number of players $n$. 
\end{observation}

The following theorem shows that \gGASP\ is not fixed-parameter tractable with respect to~$n$ unless $\textup{FPT} = \textup{W[1]}$.
This is somewhat surprising, because an FPT algorithm with respect to~$n$ could afford 
to iterate through all possible partitions of the players into coalitions.
Thus, the computational difficulty must arise from deciding which group engages in which activity, rather than from deciding who belongs to which group.
Indeed, in the following hardness reduction, there is a \emph{unique} 
partition of players into coalitions that can potentially lead to a stable assignment;
thus, computational hardness comes solely from assigning known
coalitions of players to activities.

\begin{theorem}\label{thm:W1players:clique:core}
The problem of determining whether an instance of \gGASP\ whose underlying graph is a clique 
admits a core stable assignment is {\em W[1]}-hard
with respect to the number of players $n$. 
\end{theorem}
\begin{proof}
We describe a parameterized reduction from the W[1]-hard \MIS{} problem (see, e.g., Corollary 13.8 in the book of \citet{Cygan2015}).
Given an undirected graph~$G=(V,E)$, a positive integer~$h\in \mathbb{N}$,
and a vertex coloring~$\phi \colon V\to [h]$, \MIS{} asks whether~$G$ admits an $h$-colored independent set, that is,
a size-$h$ vertex subset~$Q\subseteq V$ such that no pair of distinct vertices in~$Q$ is adjacent
and the vertices in~$Q$ have $h$~pairwise distinct colors.
Without loss of generality, we assume that there are exactly $q$ vertices of each color for some $q\in\mathbb N$, 
and that there are no edges between vertices of the same color.

Let $(G,h,\phi)$ be an instance of \MIS{} with $G=(V, E)$.
For convenience,
we write $V^{(i)} = \{v_1^{(i)}, \ldots, v_q^{(i)}\}$ to denote the set of vertices of color~$i$
for every $i \in [h]$. 
We construct our \gGASP\ instance as follows.
We have one \emph{vertex activity}~$v$ for each vertex~$v \in V$,
one \emph{edge activity}~$e$ for each edge~$e \in E$, and
four other activities $a$, $b$, $c$, and $d$.

\emph{Idea.}
We will have one \emph{color gadget} $\colour(i)$ for each color~$i \in [h]$ and one empty core gadget~$N_g$.
For most of the possible assignments, the empty core gadget will be unstable, unless the following holds:
\begin{enumerate}
 \item For each color~$i \in [h]$, the players from the color gadget select
       a vertex of color~$i$ (by being assigned together to the corresponding vertex activity).
 \item For each pair of colors $\{i,j\}$, the corresponding selected vertices are not adjacent.
 \item The stabilizer player $g$ forms a coalition with the players in $N_g$. 
\end{enumerate}
If all three conditions hold, then the assignment encodes an $h$-colored independent set.

\emph{Construction details.}
The color gadget~$\colour(i)$, $i\in [h]$ consists of two players~$p_1^{(i)}$ and~$p_2^{(i)}$ with
the following preferences:
 \begin{align*}
p_1^{(i)}:&~  E(v^{(i)}_1) \times \{5\}  \succ (v_1^{(i)},2)  \succ E(v^{(i)}_2) \times \{5\}  \succ  (v_2^{(i)},2) \succ \cdots \\
& \cdots \succ E(v^{(i)}_q) \times \{5\}  \succ (v_q^{(i)},2)  \succ (a_{\emptyset},1) \text{,}\\
p_2^{(i)}:&~  E(v^{(i)}_q) \times \{5\} \succ (v_q^{(i)},2) \succ  E(v^{(i)}_{q-1})\times \{5\}  \succ (v_{q-1}^{(i)},2) \succ \cdots \\
& \cdots \succ E(v^{(i)}_{1}) \times \{5\}  \succ (v_1^{(i)},2) \succ (a_{\emptyset},1) \text{,}
\end{align*}
where $E(v)$ denotes the set of activities corresponding to edges incident to vertex~$v$.

%Preferences for empty IS instances
The stabilizer $g$ approves each alternative of the form $(e, 5)$ with $e \in E$ and is indifferent among them; 
she also approves $(d,4)$, but likes it less than all other approved alternatives:
\[
g: E \times \{5\} \succ (d,4) \succ (a_{\emptyset},1).
\]

Finally,  $N_g$ consists of three players $g_1$, $g_2$, and $g_3$ with the following preferences:
\begin{align*}
g_1:&~ (d,4) \succ (b,2) \succ (a,3) \succ (a_{\emptyset},1) \text{,}\\
g_2:&~ (d,4) \succ (a,2) \succ (b,2) \succ (a,3) \succ (a_{\emptyset},1) \text{,}\\
g_3:&~ (d,4) \succ (a,3) \succ (b,1) \succ (a,2) \succ (a_{\emptyset},1) \text{.}
\end{align*}
In a core stable assignment, the activity $d$ should be assigned to the players in $N_g$ together with the stabilizer $g$; otherwise, the assignment would not be stable as we have seen in Example \ref{ex:core:empty}. 

Together, there are $2h+4$ players, namely
\[  N = \bigcup_{i\in [h]} \colour(i) \cup \{g\} \cup N_g.  \]
We take the underlying social network to be a complete graph.

\emph{Correctness.}
We will now argue that the graph $G$ admits an $h$-colored independent set
if and only if our instance of \gGASP\ admits a core stable feasible assignment.

Suppose that there exists an $h$-colored independent set~$H$.
We will construct a core stable assignment $\pi$ as follows:
We assign the players of the color gadget~$\colour(i)$ to the activity corresponding to the vertex of color~$i$ from~$H$, and
we assign the activity $d$ to players in $N_g$ and the stabilizer $g$.
This assignment is core stable:
Clearly, no player from $N_g$ wishes to deviate since they are assigned to their top alternatives. Consider the players of color gadget~$\colour(i)$.
They cannot deviate to another vertex activity, since they have completely opposed preferences over vertex activities of size $2$.
Now suppose towards a contradiction that there is a coalition $S$ that blocks $\pi$ together with some edge activity $e=\{v^{(i)}_{\ell_i}, v^{(j)}_{\ell_j}\}$. Activity $e$ is approved by exactly 5 players, and only with size 5. Thus, $S$ must consist of exactly these 5 players, namely the stabilizer $g$ and the players in $\colour(i)$ and $\colour(j)$. Now we see that the edge $e$ must be incident to the vertex activity assigned to the two players in $\colour(i)$ (otherwise one of these players would prefer their current alternative to $(e,5)$) and similarly $e$ must be adjacent to the vertex activity assigned to the players in $\colour(j)$.
However, this means that there is a pair of adjacent vertices in $H$, contradicting the fact that $H$ is an independent set.
We conclude that $\pi$ is a core stable feasible assignment.

Conversely, suppose that there exists a core stable feasible assignment~$\pi$.
Then, by core stability, the stabilizer $g$ as well as the players in $N_g$ must be assigned to the activity $d$.
Now, consider some edge activity corresponding to some edge~$e=\{\smash{v^{(i)}_{\ell_i},v^{(j)}_{\ell_j}}\}$.
Recall that~$e$ is approved by exactly five players:
two player from color gadget~$\colour(i)$, two players from color gadget~$\colour(j)$, and the stabilizer~$g$.
However, all five players approve the edge activity only of size five and~$g$ does engage in activity~$d$ and not in~$e$.
Thus, for every~$i \in [h]$, the two players from $\colour(i)$ must be assigned to some vertex activity in $V^{(i)}$.
Now, let $H=\{\, v \mid \smash{\pi(p^{(i)}_1)}=v~\mbox{for some}~i \in [h]\,\}$.
We have argued that $|H|=h$. It remains to show that $H$ is an independent set.
Suppose towards a contradiction that there are vertices $\smash{v^{(i)}_{\ell_i}}$ and $\smash{v^{(j)}_{\ell_j}}$ in $H$ that are adjacent.
Notice that the players in $\colour(i)$ strictly prefers the edge activity $e=\{\smash{v^{(i)}_{\ell_i},v^{(j)}_{\ell_j}}\}$ with
size $5$ to their vertex activity with size $2$.
Similarly, the players in $\colour(j)$ as well as the stabilizer $g$ strictly prefer the edge activity $e$ with size $5$
to their alternatives.
Together with the activity $e$, these players can block $\pi$, a contradiction to the core stability of $\pi$.
\end{proof}

The proofs for Nash and individually stability are related to the above reduction,
but technically more involved and can be found in the appendix. 

\begin{theorem}\label{thm:W1players:clique:NSIS}
The problem of determining whether an instance of \gGASP\ whose underlying graph is a clique 
admits a Nash stable or individually stable assignment is {\em W[1]}-hard
with respect to the number of players $n$. 
\end{theorem}

Note that although we showed W[1]-hardness for each of the parameters~$p$ and $n$,
parameterizing by the combined parameter~$p+n$ immediately gives fixed-parameter tractability,
since the input size is trivially upper-bounded by $n^2\cdot p$.
\section{Conclusion and Discussion}
In this paper, we have initiated the study of group activity selection problems with network structure, and found that 
even for very simple families of graphs, computing stable outcomes is NP-hard. We identified several ways 
to circumvent this computational intractability. For \gGASP s with copyable activities, we showed that there 
exists a polynomial time algorithm to compute stable outcomes.
We then investigated the parameterized complexity of computing stable
outcomes of group activity selection problems on networks, with respect to two natural parameters.
Many of our hardness results hold for the standard \GASP, where there are no constraints
on possible coalitions; however, some of our positive results only hold for acyclic graphs.

Interestingly, one of our tractability results holds for \GASP, but it is not clear
if it can be extended to \gGASP; thus, while simple networks may decrease complexity, 
allowing for arbitrary networks may have the opposite effect.
However, counterintuitively, for core stability we obtain hardness with just $2$
activities when the undelying network is simple (a star), whereas for cliques
our construction uses $4$ activities. It is not clear if this is an artefact of our proof approach,
or whether finding core stable outcomes is genuinely easier for cliques than for stars.

%\acks{The authors wish to thank ...}

%\vskip 0.2in
\bibliography{gaspref}
\bibliographystyle{theapa}

\appendix
\section*{Appendix. Proof omitted from Section $3$}
\subsection*{Proof of Theorem \ref{thm:core:copyable}}
\begin{proof}
The algorithm is similar to the one for hedonic games~\citep{Igarashi2016a}. We first give an informal description of our algorithm (Algorithm \ref{alg:is}), followed by pseudocode. 
Again, if the input graph $(N, L)$ is a forest, we can process each of its connected components separately, so we assume that 
$(N, L)$ is a tree. 
We choose an arbitrary node $r$ as the root and construct a rooted tree $(N,T)$ by orienting the edges in $L$ towards the leaves. 
We denote by $\ch(i)$ the set of children of $i$ and by $\desc(i)$ the set of descendants of $i$ (including $i$) in the rooted tree. For each $i  \in N$, we define $\height(i)=0$ if $\ch(i)=\emptyset$, and $\height(i)=1+\max \{\, \height(j)\mid j \in \ch(i)\,\}$ otherwise. We denote by $(N,T|S)$ the subdigraph induced by $S \subseteq N$, i.e., $T|S=\{\, (i,j) \in T \mid i,j \in S \,\}$.

The algorithm has two different phases: the bottom-up and top-down phases. 
\begin{itemize}
	\item \emph{Bottom-up phase}: In the bottom-up phase, we will determine {\em guaranteed} activity $a(i)$ and coalition $S(i)$ for every subroot $i$. To this end, we start with the assignment obtained by combining
	the previously constructed assignments $a(j)$ for the children $j$ of $i$
	and assigning $i$ to the void activity. We then let player $i$ join the most preferred activity among those to which she has an IS deviation. After that we keep adding  
	players to $i$'s coalition $S(i)$ as long as the resulting coalition remains feasible,
	the player being added is willing to move, 
	and such a deviation is acceptable for all members of $i$'s coalition. 
	\item \emph{Top-down phase}: In the top-down phase, the algorithm builds a feasible assignment $\pi$, by iteratively choosing a root $r^{\prime}$ of the remaining rooted trees and reassigning the activity $a(r^{\prime})$ to its coalition $S(r^{\prime})$. Since each activity is copyalbe, we can always find an activity that is equivalent to $a(r^{\prime})$ and has not been used by their predecessors.
\end{itemize}

\begin{algorithm}
\SetKwInOut{Input}{input} 
\SetKwInOut{Output}{output}
\SetKw{And}{and}
\SetKw{None}{None}
\caption{Finding individually stable assignments}\label{alg:is}
\Input{tree $(N,L)$, activity set $A=A^*\cup \{a_{\emptyset}\}$, $r\in N$, preference $\succeq_i$, $i\in N$}
\Output{$\pi:N \rightarrow A$}
	{\tt // Bottom-up phase: assign activities to players in a bottom-up manner.}\;
	make a rooted tree $(N,T)$ with root $r$ by orienting all the edges in $L$\;
	initialize $S(i)\leftarrow \{i\}$ and $a(i) \leftarrow a_{\emptyset}$ for each $i \in N$\;
	\ForEach{$t=0,\ldots,\height(r)$}{
	\ForEach{$i \in N$ with $\height(i)=t$}{
	{\tt // Let $i$ join the favourite activity to which $i$ has an IS-deviation.}
	set $C^*(i)=\{\, j \in \ch(i)\mid (a(j),|S(j)|+1) \succeq_{k} (a(j),|S(j)|)~\mbox{for all}~k \in S(j) \,\}$\;
	\If{there exists $b \in A$ such that $(b,1) \succ_i (a(j),|S(j)|+1)$ for all $j \in C^*(i)$}{
	find $b^* \in A$ such that $(b^*,1) \succeq_i(b,1)$ for all $b \in A$\;
	set $S(i)\leftarrow \{i\}$ and $a(i) \leftarrow b^*$\;
	\Else{
	find $j^* \in C^*(i)$ such that $(a(j^*),|S(j^*)|+1) \succeq_i (a(j),|S(j)|+1)$ for all $j \in C^*(i)$\;
	set $S(i) \leftarrow S(j^*)\cup \{i\}$ and $a(i) \leftarrow a(j^*)$ \label{step:children}\;
	}
	}
	{\tt // Add a player to $S(i)$ as long as the deviation is IS-feasible.}\;
	\While{$(a(i),|S(i)|+1) \succeq_k (a(i),|S(i)|)$ for all $k \in S(i)$ and there exists $j \in N \setminus S(i)$ such that $j$'s parent belongs to $S(i)$ and $(a(i),|S(i)|+1) \succ_j (a(j), |S(j)|)$\label{step:criterion}}{
	$S(i) \leftarrow S(i) \cup \{j\}$\;
	}
	}
	}
	{\tt // Top-down phase: relabel players with their predecessor's activities}\;
	set $N^{\prime} \leftarrow N$ and $A^{\prime} \leftarrow A^*$\;
	\While{$N^{\prime} \neq \emptyset$}{
	choose a root $r^{\prime}$ of some connected component of the digraph $(N^{\prime},T|{N^{\prime}})$ and find an activity $b \in A^{\prime}\cup \{a_{\emptyset}\}$ that is equivalent to $a(r^{\prime})$\;
	set $\pi(i) \leftarrow b$ for all $i \in S(r^{\prime})$\;
	set $N^{\prime} \leftarrow N^{\prime} \setminus S(r^{\prime})$ and $A^{\prime} \leftarrow A^{\prime} \setminus \{b\}$\;
}
\end{algorithm}

We will now argue that Algorithm \ref{alg:is} correctly finds an individually stable feasible assignment.
We start by observing the following lemma.
\begin{lemma}\label{lem1}
For all $i \in N$, the following statements hold:
\begin{itemize}
\item[$(${\rm i}$)$] $i$ has no incentive to deviate to an alternative of size $1$, i.e., $(a(i),|S(i)|)\succeq_i (b,1)$ for all $b \in A$,
\item[$(${\rm ii}$)$] $i$ has no IS-deviation to her children's coalitions, i.e., $(a(i),|S(i)|)\succeq_i (a(j),|S(j)|+1)$ for any $j \in C^*(i)$, and 
\item[$(${\rm iii}$)$] all players in $S(i)$ weakly prefer $(a(i),|S(i)|)$ to their guaranteed alternative $(a(j),|S(j)|)$.
\end{itemize}
\end{lemma}
\begin{proof}
The statements $(${\rm i}$)$ and $(${\rm ii}$)$ in Lemma \ref{lem1} immediately follow from the choice of $a(i)$ in lines $4$ -- $11$ and the stopping criterion
of the {\bf while} loop in line \ref{step:criterion}. 

We will now prove $(${\rm iii}$)$. Assume that there is a pair of players not satisfying the condition $(${\rm iii}$)$. Among such pairs, take $i \in N$ and $j \in S(i)$ with $|\height(i)-\height(j)|$ being the minimum. 
First, suppose that $j$ joins the coalition $S(i)$ when $S(i)$ is initialized in line \ref{step:children}. Then, $(a(j^*),|S(j^*)|) \succeq_j (a(j),|S(j)|)$ by the minimality but we have $(a(i),|S(i)|) \succeq_j (a(j^*),|S(j^*)|)$ by the fact that adding players to $S(i)$ does not decrease $j$'s utility, which leads to $(a(i),|S(i)|) \succeq_j (a(j),|S(j)|)$, a contradiction.
Second, suppose that the player $j$ has been added to $S(i)$ in the {\bf while} loop in line \ref{step:criterion}. This would mean that $j$ strictly prefers $(a(i),|S(i)|)$ to $(a(j), |S(j)|)$ at the time of termination, a contradiction. 
\end{proof}

Now, by $(${\rm iii}$)$ in Lemma \ref{lem1}, it can be easily verified that $(\pi(i),|\pi_i|) \succeq_i (a(i),|S(i)|)$ for all $i \in N$. Combining this with $(${\rm i}$)$, we know that at the assignment $\pi$, all players weakly prefer their alternatives to engaging alone in unused activities or the void activity. It thus remains to show that no player has an IS deviation to used activity.
Suppose towards a contradiction that there is a player $i$ who has an IS feasible deviation to the activity $\pi(j)$ of her neighbour $j$. 
If the player $j$ is a child of $i$, this would mean that all players in $S(j)$ accept $i$, i.e., $j \in C^*(i)$, and $i$ strictly prefers $(\pi(j),|\pi_j|+1)$ to her alternative, namely,
\[
(a(j),|S(j)|+1)=(\pi(j),|\pi_j|+1) \succ_i (\pi(i),|\pi_i|) \succeq_i (a(i),|S(i)|),
\]
contradicting $(${\rm ii}$)$ in Lemma \ref{lem1}. If the player $j$ is the parent of $i$ in the rooted tree $T$, $j$ must have been added to $\pi_j$ in the {\bf while} loop in line $14$, a contradiction. We conclude that no player has an IS-deviation to used activities, and hence $\pi$ is individually stable. 

It remains to analyze the running time of Algorithm~\ref{alg:is}.
Consider the execution of the algorithm for a fixed player $i$.
Let $c=|\ch(i)|$ and $d=|\desc(i)|$. Line~4
requires at most $d$ queries: no descendant of $i$ is queried more than once.
Lines~5 -- 10 require $p$ queries.
Moreover, at each iteration of the {\bf while}
loop in lines 12--14 at least one player joins $S(i)$, so there are at most 
$d$ iterations, in each iteration we consider at most $d$ candidates,
and for each candidate we perform at most $d$ queries.
Summing over all players, we conclude that the number of queries for the bottom-up phase is bounded by $O(n(n^2+p))$. It is immediate that the top-down phase can be done in polynomial time.
This completes the proof of the theorem. 
\end{proof}

\section*{Appendix. Proofs omitted from Section 4}
\subsection*{Proof of Theorem \ref{thm:NPhardness:path:star:sc:CRIS}}
\begin{proof}
We give separate proofs for the three types of networks considered: paths, stars, and graphs with small components.
	
%path%%%%%%%%%%%%%%%%%%%%%%%%%%%%%%%%%%%%%%%%%%%%
\paragraph{Paths.}
We prove hardness via a reduction from {\sc Path Rainbow Matching}. 

\emph{Construction.} 
Given an instance $(G, \calC, \phi, k)$ of {\sc Path Rainbow Matching} where 
$|\calC|=q$, we create a vertex 
player $v$ for each $v \in V$ and an edge player $e$ for each $e \in E$. 
To create the social network, we first construct the graph $(N_{G},L_{G})$ as defined in the proof for  Theorem \ref{thm:NPhardness:path:NS}. 
To the right of the graph $(N_G,L_G)$, we attach a path consisting of garbage collectors $\{g_{1},g_{2},\ldots,g_{q-k}\}$ 
and $q$ copies $(N_{c},L_{c})$ of the empty-IS instance of Example \ref{ex:IS:empty} where $N_{c}=\{c_1,c_2,c_3\}$ and
$L_{c}=\{ \{c_1,c_2\},\{c_2,c_3\}\}$ for each $c \in \calC$. 
For each color $c \in \calC$, we introduce a color activity $c$, and introduce additional activities $a(c)$ and $b(c)$.
Each vertex player $v$ approves color activities $\phi(e)$ of its adjacent edges $e$ with size $3$; each edge player $e$ 
approves the color activity $\phi(e)$ of its color with size $3$; each garbage collector $g_{i}$ approves any color 
activity $c$ with size $1$; finally, the preference for players in $N_{c}$ $(c \in \calC)$ is cyclic and given by 
\begin{align*}
c_1:&~ (b(c),2) \succ (a(c),1) \succ (c,3) \succ (c,2) \succ (c,1) \succ (a_{\emptyset},1)\\
c_2:&~ (c,3) \succ (c,2) \succ (a(c),2) \succ (b(c),2) \succ (b(c),1) \succ (a_{\emptyset},1)\\
c_3:&~ (c,3) \succ (a(c),2) \succ (a(c),1) \succ (a_{\emptyset},1)
\end{align*}

\emph{Correctness.} 
We will now prove that the following three statements are equivalent. 
\begin{enumerate}
\item[$($i$)$] There exists  an individually stable feasible assignment.
\item[$($ii$)$] There exists a core stable feasible assignment.
\item[$($iii$)$] $G$ contains a rainbow matching of size at least $k$.
\end{enumerate}

\noindent
$(${\rm i}$)$ or $(${\rm ii}$)\Longrightarrow(${\rm iii}$)$: 
We will show that if one of two stable solutions exists, then there exists a rainbow matching of size at least $k$. First, suppose that there is a core or individually stable feasible assignment $\pi:N \rightarrow A$. Let $M=\{\, e \in E \mid \pi(e) \in \calC \,\}$. We will show that $M$ is a rainbow matching of size at least $k$. To see this, notice that at $\pi$, all the color activities should be played outside $N_{c}$'s, since otherwise no core or individually stable assignment would exist as we have seen in Example \ref{ex:IS:empty}. Further, at most $q-k$ colour activities are played among the garbage collectors, which means that at least $k$ colour activities should be assigned to vertex and edge players. The only individual rational way to do this is to select triples of the form $(u,e,v)$ where $e=\{u,v\} \in E$ and assign to them their colour activity $\phi(e)$; thus, $M$ is a rainbow matching of size at least $k$.

\noindent
$(${\rm iii}$)\Longrightarrow(${\rm i}$)$ and $(${\rm ii}$)$: 
Suppose that there exists a rainbow matching $M$ of size $k$. We construct a feasible assignment 
$\pi$ where for each $e=\{u, v\}\in M$ we set $\pi(e)=\pi(u)=\pi(v)=\phi(e)$,
each garbage collector $g_{i}$, $i\in[q-k]$, is arbitrarily assigned to one of the remaining $q-k$ color activities, each pair of $c_2$ and $c_3$ $(c \in \calC)$ is assigned to $a(c)$, 
and the remaining players are assigned to the void activity. 
The assignment $\pi$ is individually stable, since every garbage 
collector as well as every edge or vertex player assigned to a color activity
is allocated their top alternative, and no remaining player has an IS feasible deviation.
Similarly, it can be easily verified that $\pi$ is core stable.

%star%%%%%%%%%%%%%%%%%%%%%%%%%%%%%%%%%%%%%%%%%%%
\paragraph{Stars.}
Again, we reduce from a restricted variant of MMM where the graph is a bipartite graph. We have seen that computing a core stable outcome is NP-hard for stars in Theorem \ref{thm:NP:star:core}; hence, we only provide a hardness proof for individual stability. 

\emph{Construction.} Given a bipartite graph $(U,V,E)$ and an integer $k$, we construct an instance of \gGASP\  on a path as follows.
We create a star with center $c$ and the $|V|+2$ leaves: one leaf for each vertex $v \in V$ plus two other players $s_{1}$ and $s_{2}$. 
We then introduce an activity $u$ for each $u \in U$, and four other simple activities $a$, $x$, $y$, and $z$. 

A player $v \in V$ approves $(u,1)$ for each $u \in U$ such that $\{u,v\}\in E$ as well as $(a,|V|-k+1)$ and prefers the former to the latter.  That is, $(u,1) \succ_{v}(a,|V|-k+1)$ for any $u \in U$ with $\{u,v\}\in E$; $v$ is indifferent among activities associated with its neighbors in the graph, that is, $(u,1)\sim_{v} (u^{\prime},1)$ for all $u,u^{\prime} \in U$ such that $\{u,v\},\{u^{\prime},v\}\in E$. The center player $c$ strictly prefers $(a,|V|-k+1)$ to any other alternative, and has the same cyclic preferences over the alternatives of $x,y$, and $z$ as in Example \ref{ex:IS:empty} together with players $s_{1}$ and $s_{2}$; explicitly, the preferences are given by
\begin{align*}
&s_{1}:~(y,2) \succ (x,1) \succ (z,3) \succ (z,2) \succ (z,1)  \succ (a_{\emptyset},1)\\
&c: (a,|V|-k+1) \succ (z,3) \succ (z,2) \succ (x,2) \succ (y,2) \succ (y,1) \succ (a_{\emptyset},1)\\
&s_{2}:~(z,3) \succ (x,2) \succ (x,1) \succ (a_{\emptyset},1).
\end{align*}

Here, $s_{1}$'s (respectively, the center $c$ and the player $s_{2}$) preference corresponds to the one for player $1$ (respectively, player $2$ and player $3$) in Example \ref{ex:IS:empty}.

\emph{Correctness.} We will show that there exists an individually stable feasible assignment if and only if $G$ contains a maximal matching of size at most $k$.

Suppose that there exists an individually stable feasible assignment $\pi$ and let $M=\{\, \{\pi(v),v\} \mid v \in V \land \pi(v) \in U \,\}$. We will show that $M$ is a maximal matching of size at most $k$. By stability, the center player and $|V|-k$ vertex players are assigned to the activity $a$; otherwise, no individually stable outcome would exist as we have seen in Example \ref{ex:IS:empty}; thus, $|M| \leq k$. Notice further that $M$ is a matching since each vertex player $v$ plays at most one activity, and by individual rationality each vertex activity should be assigned to at most one player. Now suppose towards a contradiction that $M$ is not maximal, i.e., there exists an edge $\{u,v\} \in E$ such that $u \in U$, $v \in V$, and $M\cup \{u,v\}$ is a matching. This would mean that $\pi$ assigns no player to activity $u$ and no vertex activity to player $v$. Hence, the player $v$ has an IS-deviation to the vertex activity $u$, contradicting the stability of $\pi$. 

Conversely, suppose that $G$ admits a maximal matching $M$ with at most $k$ edges. We construct a feasible assignment $\pi$ by setting $\pi(v)=u$ for each $\{u,v\} \in M$, and assigning $|V|-k$ non-matched vertex players and the center to $a$, assigning $s_{1}$ to $x$, and assigning the remaining players to the void activity. The center $c$ is allocated to her top alternative; hence no player has an IS-deviation to an activity $a$ and the center does not want to deviate to a used vertex activity activity. No vertex player $v$ has an IS-deviation to an unused vertex activity $u$, since if such a pair $\{u,v\}$ existed, this would mean that $\{u,v\}$ is not included in $M$, and hence $M\cup \{u,v\}$ forms a matching, which contradicts the maximality of $M$. Hence, $\pi$ is individually stable. This completes the proof.

%%%%%%%%%%%%%%%%%%%%%%%%%%%%%%%%%%%%%%%%%%%%
\paragraph{Small components.}
We reduce from (3,B2)-{\sc Sat}. 

\emph{Construction.} Consider a formula $\phi$ with variable set $X$ and clause set $C$, where for each variable $x \in X$ we write $x_1$ and $x_2$ for the two positive occurrences of $x$, and ${\bar x_1}$ and ${\bar x_2}$ for the two negative occurrences of $x$. We construct an instance of \gGASP\ as follows. 
%Players
For each variable $x \in X$, we introduce three variable players $v_1(x)$, $v_2(x)$, and $v_3(x)$; and we also introduce three positive variable players $p_1(x)$, $p_2(x)$, and $p_3(x)$, and three negative variable players ${\bar p_1}(x)$, ${\bar p_2}(x)$, and ${\bar p_3}(x)$. For each clause $c \in C$, we introduce three players $N_c=\{c_1,c_2,c_3\}$. 
The network consists of one component for each clause $c \in C$: a star with center $c_2$ and leaves $c_1$ and $c_3$, and of three components for each variable $x \in X$: 
\begin{itemize}
\item a star with center $v_2(x)$ and leaves $v_1(x)$ and $v_3(x)$, 
\item a star with center $p_2(x)$ and leaves $p_1(x)$ and $p_3(x)$, and 
\item a star with center ${\bar p_2}(x)$ and leaves ${\bar p_1}(x)$ and ${\bar p_3}(x)$. 
\end{itemize}
Hence, the size of each connected component of this graph is at most $3$.
%Activities
Corresponding to the variable occurrences, we introduce two positive literal activities $x_1$ and $x_2$, 
two negative literal activities ${\bar x_1}$ and ${\bar x_2}$ for each $x\in X$; we also introduce one variable activity $x$, and eight further activities $y(x),z(x)$, $a(x),b(x),c(x)$, ${\bar a}(x),{\bar b}(x)$, and ${\bar c}(x)$. 
Thus, the set of activities given by
\[
A^*=\bigcup_{x \in X} \{x_1,x_2,{\bar x_1},{\bar x_2},x, y(x),z(x),a(x),b(x),c(x),{\bar a}(x),{\bar b}(x),{\bar c}(x)\}.
\]
For each $x \in X$, the preferences of variable players are given as follows:
\begin{align*}
&v_1(x):~(y(x),2) \succ (x,1) \succ (z(x),3) \succ (z(x),2) \succ (z(x),1) \succ (a_{\emptyset},1),\\
&v_2(x):~(z(x),3) \succ (z(x),2) \succ (x,2) \succ (y(x),2) \succ (y(x),1) \succ (a_{\emptyset},1),\\
&v_3(x):~(z(x),3) \succ (x,2) \succ (x,1) \succ (a_{\emptyset},1).
\end{align*} 
Notice that the preferences are cyclic; thus, in a core or individually stable assignment, one of the activities $x, y(x)$, and $z(x)$ must be allocated outside of the variable players $v_1(x)$, $v_2(x)$, and $v_3(x)$. For each $x \in X$, the preferences of the three positive variable players $p_1(x)$, $p_2(x)$, and $p_3(x)$ are given as follows: 
\begin{align*}
&p_1(x):~(x,3) \sim (x_1,1)\succ (b(x),2) \succ (a(x),1) \succ (c(x),3) \succ (c(x),2) \succ (c(x),1) \succ (a_{\emptyset},1),\\
&p_2(x):~(x,3) \sim (x_2,1) 
\succ (c(x),3) \succ (c(x),2) \succ (a(x),2) \succ (b(x),2) \succ (b(x),1) \succ (a_{\emptyset},1),\\
&p_3(x):~(x,3) \succ (c(x),3) \succ (a(x),2) \succ (a(x),1) \succ (a_{\emptyset},1).
\end{align*} 
Similarly, for each $x \in X$, the preferences of the negative variable players ${\bar p_1}(x)$, ${\bar p_2}(x)$, and ${\bar p_3}(x)$ are given as follows:
\begin{align*}
&{\bar p_1}(x):~(x,3) \sim ({\bar x_1},1)\succ ({\bar b}(x),2) \succ ({\bar a}(x),1) \succ ({\bar c}(x),3) \succ ({\bar c}(x),2) \succ ({\bar c}(x),1) \succ (a_{\emptyset},1),\\
&{\bar p_2}(x):~(x,3) \sim ({\bar x_2},1) 
\succ ({\bar c}(x),3) \succ ({\bar c}(x),2) \succ ({\bar a}(x),2) \succ ({\bar b}(x),2) \succ ({\bar b}(x),1) \succ (a_{\emptyset},1),\\
&{\bar p_3}(x):~(x,3) \succ ({\bar c}(x),3) \succ ({\bar a}(x),2) \succ ({\bar a}(x),1) \succ (a_{\emptyset},1).
\end{align*}
Again, observe that the preferences of each triple contains a cyclic relation, and hence in a core stable assignment and an individually stable assignment, there are three possible cases: 
\begin{itemize}
\item all the three players $p_1(x)$, $p_2(x)$, and $p_3(x)$ are assigned to $x_1$, $x_2$, and $a(x)$, respectively, and players ${\bar p_1}(x)$, ${\bar p_2}(x)$, and ${\bar x}$ are assigned to activities ${\bar x_1}$ , ${\bar x_2}$, and ${\bar a}(x)$, respectively; 
\item all the three players $p_1(x)$, $p_2(x)$, and $p_3(x)$ are assigned to $x$, and players ${\bar p_1}(x)$, ${\bar p_2}(x)$, and ${\bar x}$ are assigned to activities ${\bar x_1}$ , ${\bar x_2}$, and ${\bar a}(x)$, respectively; or, 
\item all the players ${\bar p_1}(x)$, ${\bar p_2}(x)$, and ${\bar p_3}(x)$ are assigned to $x$, and players $p_1(x)$, $p_2(x)$, and $p_3(x)$ are assigned to activities $x_1$, $x_2$, and $a(x)$, respectively.
\end{itemize}
Later we will see that the last two cases can only occur in a stable assignment.

For each clause $c \in C$ where $\ell^c_1$, $\ell^c_2$, and $\ell^c_3$ are the literals in a clause $c$, the preferences for clause players $c_1$, $c_2$, and $c_3$ are again cyclic and given as follows: 
\begin{align*}
&c_1:~ (\ell^c_2,2) \succ (\ell^c_1,1) \succ  (\ell^c_3,3) \succ  (\ell^c_3,2) \succ (\ell^c_3,1) \succ (a_{\emptyset},1),\\
&c_2:~(\ell^c_3,3) \succ (\ell^c_3,2) \succ (\ell^c_1,2)\succ (\ell^c_2,2) \succ (\ell^c_2,1) \succ  (a_{\emptyset},1),\\
&c_3:~(\ell^c_3,2)\succ (\ell^c_1,3) \succ  (\ell^c_1,1) \succ (a_{\emptyset},1).
\end{align*}
If there exists a core or individually stable outcome, it must be the case that at least one of the literal activities $\ell^c_1$, $\ell^c_2$, and $\ell^c_3$ must be used outside of the three players $c_1$, $c_2$, and $c_3$; otherwise, no feasible assignment would be individually stable.

\emph{Correctness.} Now we will show that the following three statements are equivalent. 
\begin{enumerate}
\item[$($i$)$] There exists  an individually stable feasible assignment.
\item[$($ii$)$]  There exists a core stable feasible assignment.
\item[$($iii$)$] The formula $\phi$ is satisfied by some truth assignment. 
\end{enumerate}

\noindent
$(${\rm i}$)$ or $(${\rm ii}$)\Longrightarrow(${\rm iii}$)$: 
Suppose that there exists a core or individually stable feasible assignment $\pi$. By stability, for each variable $x \in X$, one of the activities $x$, $y(x)$, and $z(x)$ must be assigned outside of the variable players $v_1(x)$, $v_2(x)$, and $v_3(x)$. The only individually rational way to do this is to assign the variable activity~{$x$} to either the positive variable players $p_1(x)$, $p_2(x)$, and $p_3(x)$, or the negative variable players ${\bar p_1}(x)$, ${\bar p_2}(x)$, and ${\bar p_3}(x)$, meaning that either a pair of positive literal activities $x_1$ and $x_2$ or a pair of negative literal activities ${\bar x_1}$ and ${\bar x_2}$ should be assigned to the corresponding pair of variable players. Further, for each clause $c$, at least one of the literal activities $\ell^c_1$, $\ell^c_2$, and $\ell^c_3$ should be played outside of the clause players $c_1$, $c_2$, and $c_3$. Then, take the truth assignment that sets the variables $x$ to {\em true} if their positive variable players $p_1(x)$ and $p_2(x)$ are assigned to positive literal activities $x_1$ and $x_2$; otherwise, $x$ is set to {\em false}; this can be easily seen to satisfy $\phi$. 

\noindent
$(${\rm iii}$)\Longrightarrow(${\rm i}$)$ and $(${\rm ii}$)$: 
Suppose that there exists a truth assignment that satisfies $\phi$. We construct a feasible assignment $\pi$ as follows. First, for every variable $x \in X$, we assign $z(x)$ to the variable players $v_1(x)$, $v_2(x)$, and $v_3(x)$. Then, for each variable $x$ that is set to {\em true}, we assign positive literal activities $x_1$, $x_2$, and $a(x)$ to positive variable players $p_1(x)$, $p_2(x)$, and $p_3(x)$, respectively, and assign a variable activity $x$ to players ${\bar p_1}(x)$, ${\bar p_2}(x)$, and~${\bar p_3}(x)$. For each variable $x$ that is set to {\em  false}, we assign negative literal activities ${\bar x_1}$, ${\bar x_2}$, and ${\bar a}(x)$ to negative literal players ${\bar p_1}(x)$, ${\bar p_2}(x)$, and ${\bar p_3}(x)$, respectively, and assign a variable activity $x$ to players $p_1(x)$, $p_2(x)$, and $p_3(x)$. Note that this procedure uses at least one of the literal activities $\ell^c_1$, $\ell^c_2$ and $\ell^c_3$ of each clause $c$, since the given truth assignment satisfies $\phi$. Then, for each clause $c \in C$, we assign activities to the clause players $c_1$, $c_2$, and $c_3$ as follows.
\begin{itemize}
\item If all the activities are already assigned in a previous step, then we assign the void activity to all the clause players $c_1$, $c_2$, and $c_3$.
\item If $\ell^c_1$ is already assigned and $\ell^c_3$ is not assigned, then we assign  $\ell^c_3$ to all the clause players $c_1$, $c_2$, and $c_3$.
\item If $\ell^c_3$ is already assigned and $\ell^c_1$ is not assigned, then we assign $\ell^c_1$ to the players $c_2$, and $c_3$.
\item If both $\ell^c_1$ and $\ell^c_3$ are already assigned and $\ell^c_2$ is not assigned, then we assign $\ell^c_2$ to the player $c_1$ and $c_2$.
\end{itemize}
The resulting assignment $\pi$ of players to activities is individually stable, because no variable player wishes to change their alternative and no player of each $N_{c}$ has an IS feasible deviation. Similarly, it can be easily verified that $\pi$ is core stable.
\end{proof}

\section*{Appendix. Proofs omitted from Section 5}
\subsection*{Proof of Theorem \ref{thm:FPT:smallcomponents:CR}}
\begin{proof}
	Again, we give a dynamic programming algorithm. Suppose our graph $(N,L)$ has $k$ connected components $(N_{1},L_{1}), 
	(N_{2},L_{2}), \ldots, (N_{k},L_{k})$. For each $i\in[k]$, each set $B \subseteq A^{*}$ of activities 
	assigned to $N$, and each set $B^{\prime} \subseteq B$ of activities assigned to $\bigcup^{i}_{j=1}N_{j}$, we let $\CR[i,B,B^{\prime}]$ to be {\em  true} if there exists an individually rational feasible assignment $\pi:\bigcup^{i}_{j=1}N_{j} \rightarrow A$ such that 
\begin{itemize}
\item the set of activities assigned to $\bigcup^{i}_{j=1}N_{j}$ is exactly $B^{\prime}$; and
\item no connected subset $S \subseteq \bigcup^{i}_{j=1}N_{j}$ together with an activity in $B^{\prime}\cup(A^{*} \setminus B)$ strongly blocks $\pi$
\end{itemize}
Otherwise, $\CR[i,B,B^{\prime}]$ is {\em  false}.
	
	For $i=1$, each $B \subseteq A^{*}$, and each $B^{\prime}\subseteq B$, we compute the value of $\CR[1,B,B^{\prime}]$ by trying all possible mappings $\pi:N \rightarrow B^{\prime} \cup \{a_{\emptyset}\}$, and checking whether it is an individually rational feasible assignment using all activities in $B^{\prime}$ and such that none of the connected subsets $S \subseteq N_{1}$ together with an activity in $B^{\prime}\cup(A^{*} \setminus B)$ strongly blocks $\pi$.	

	For $i=2,3,\ldots,k$, each $B \subseteq A^{*}$, and $B^{\prime} \subseteq B$, we set $\CR[i,B,B^{\prime}]$ to {\em true} if there exists a bipartition of $B^{\prime}$ into $P$ and $Q$ such that $\CR[i-1,B,P]$ is {\em true} and there exists an individually rational feasible assignment $\pi:N_{i} \rightarrow Q\cup \{a_{\emptyset}\}$ such that each activity in $Q$ is assigned to some player in $N_i$, and none of the connected subsets $S \subseteq N_{i}$ together with an activity in $Q\cup(A^{*} \setminus B)$ strongly blocks $\pi$. 

	It is not difficult to see that a core stable solution 
	exists if and only if $\CR[k,B,B]$ is {\em true} for some $B \subseteq A^{*}$. If this is the case, such a stable feasible 
	assignment can be found using standard dynamic programming techniques. 
	The size of the dynamic programming table is at most $4^p n$. There are at most $2^p$ bipartitions of $B^{\prime}$, and for each bipartition, we consider $O(p^c)$ possible assignments; further, for each assignment, we have shown that there is an $O(pc^3)$ time algorithm to check the existence of a strongly blocking coalition. Thus, each entry can be filled in time $O(2^p p^{c+1}c^3)$. This completes the proof.
\end{proof}

\subsection*{Proof of Theorem \ref{thm:FPT:tree:IS}}
\begin{proof}
Again, given a tree $(N,L)$, we choose an arbitrary node as the root and construct a rooted tree by orienting the edges in $L$ towards the leaves; we denote by $\ch(i)$ the set of children of $i$ and by $\desc(i)$ the set of descendants of $i$ (including $i$). 

We process the nodes from the leaves to the root. For each $i \in N$, each $B \subseteq A^{*}$, each $B^{\prime} \subseteq B$, each $(a,k) \in B^{\prime} \times [n] \cup \{(a_{\emptyset},1)\}$, and each $t \in [k]$, we will again check whether there is a partial assignment of $B^{\prime}$ to $i$'s descendants that is extensible to a stable assignment, assuming that the whole players are assigned to the activities in $B$. Here, there are three cases: first, the $i$'s coalition may admit a deviation from their descendants but some player who does not want to increase the size of the coalition may join the coalition later on ({\em possibly stable}); second, there is already a player in the $i$'s coalition who does not want to increase the size of the coalition ({\em definitely stable}); third, no descendant of $i$'s coalition wants to deviate to it ({\em weakly stable}). For all three concepts, all coalitions of descendants of $i$ that do not involve $i$ are immune to any IS-deviation. 

%%%%define PS
Formally, we first let $\PS[i,B,B^{\prime},(a,k),t]$ be {\em true} if there exists a feasible assignment $\pi:N \rightarrow A$ where
\begin{enumerate}
\item[$(${\rm i}$)$]  the set of activities assigned to players in $\desc(i)$ is exactly $B^{\prime}$;
\item[$(${\rm ii}$)$]  player $i$ is assigned to $a$ and is in a coalition with $k$ other players;
\item[$(${\rm iii}$)$]  exactly $t$ players in $\desc(i)$ belong to the same group as $i$;
\item[$(${\rm iv}$)$]  the $t$ players in $\desc(i) \cap \pi_i$ weakly prefer $(a,k)$ to $(b, 1)$ for each $b\in A\setminus B$, and has no IS-deviation to the activities assigned to players in $D_i \setminus \pi_i$; and
\item[$(${\rm v}$)$]  the players in $\desc(i)\setminus \pi_i$ weakly prefer their alternative under $\pi$ to engaging alone in any of the activities in $A\setminus B$, and have no IS-deviations to activities in $B^{\prime}\setminus \{a\}$. 
\end{enumerate}
Otherwise, we let $\PS[i,B,B^{\prime},(a,k),t]$ be {\em false}. 
%%%%define DS
Second, we let $\DS[i,B,B^{\prime},(a,k),t]$ be \emph{true} if there exists a feasible assignment $\pi:N \rightarrow \bbR$ such that $(${\rm i}$)$ - $(${\rm v}$)$ hold, and $(${\rm vi}$)$ $a=a_{\emptyset}$ or some player in $\desc(i) \cap \pi_i$ strictly prefers $(a,k)$ to $(a,k+1)$; we let $\DS[i,B,B^{\prime},(a,k),t]$ be \emph{false} otherwise. 
%%%define WS
Third, we let $\WS[i,B,B^{\prime},(a,k),t]$ be \emph{true} if there exists a feasible assignment $\pi:N \rightarrow \bbR$ such that $(${\rm i}$)$ - $(${\rm v}$)$ hold, and $(${\rm vii}$)$ $a=a_{\emptyset}$ or no player in $\desc(i) \setminus \pi_i$ strictly prefers $(a,k+1)$ to their alternative under $\pi$; we let $\WS[i,B,B^{\prime},(a,k),t]$ be \emph{false} otherwise. 

By construction, our instance admits an individually stable assignment if and only if there exists a combination of the arguments $B$, $B^{\prime}$, and $(a,k)$ such that $\DS[t,B,B^{\prime},(a,k),k]$ is {\em true} or $\WS[r,B,B^{\prime},(a,k),k]$ is {\em true}, where $r$ is the root. In what follows, we will describe how to compute each entry of all the three tables in a bottom-up manner. 

\emph{Leaves.} 
%leaves
If $i$ is a leaf, we set both $\PS[i,B,B^{\prime},(a,k),t]$ and $\WS[i,B,B^{\prime},(a,k),t]$ to \emph{true} if $B^{\prime}=\{a\}$, $t=1$, and $i$ weakly prefers $(a,k)$ to any other alternative $(b, 1)$ such that $b\in A\setminus B$; otherwise, we set both $\PS[i,B,B^{\prime},(a,k),t]$ and $\WS[i,B,B^{\prime},(a,k),t]$ to \emph{false}. We set $\DS[i,B,B^{\prime},(a,k),t]$ to \emph{true} if $\PS[i,B,B^{\prime},(a,k),t]$ is true, and $a=a_{\emptyset}$ or $i$ strictly prefers $(a,k)$ to $(a,k+1)$; otherwise, we set $\DS[i,B,B^{\prime},(a,k),t]$ to \emph{false}.

\emph{Internal vertices.} 
%internal vertices
Now consider the case where $i$ is an internal vertex; we assume that all the values for $i$'s descendants have been computed.
%IR condition
We first check whether $i$ strictly prefers some alternative $(b,1)$ such that $b\in A\setminus B$ to $(a,k)$; if so, we set $\PS[i,B,B^{\prime},(a,k),t]$,  $\DS[i,B,B^{\prime},(a,k),t]$, and $\WS[i,B,B^{\prime},(a,k),t]$ to \emph{false}. 
%%%
Otherwise, we proceed and check for each partition $\calP$ of $B^{\prime} \setminus \{a\}$ and each of its permutations whether there is an allocation of each activity set $P\in \calP$ to some subtree rooted at $i$'s child that gives rise to an assignment with the conditions described before. One can show that determining the existence of a desired assignment can be computed in polynomial time. Then, similarly to the proof of Theorem \ref{thm:FPT:tree:IS}, our problem for trees is in FPT and, moreover, we can easily extend this to arbitrary forests.

Now let us fix a partition $\calP=\{P_1,P_2,\ldots,P_{|\calP|}\}$ of $B^{\prime}\setminus \{a\}$. Again, without loss of generality, we consider an ordering $P_1,P_2,\ldots,P_{|\calP|}$ and seek to assign each activity set to the subtrees in that order. It remains to show that there exists a polynomial time algorithm that determines whether there exists a feasible assignment $\pi:N \rightarrow A$ such that the conditions $(${\rm i}$)$ - $(${\rm v}$)$ (respectively, $(${\rm i}$)$ - $(${\rm vi}$)$, and $(${\rm i}$)$ - $(${\rm v}$)$, $(${\rm vii}$)$) hold, and each activity set in $\calP$ is assigned to the players in $\desc(j)$ for some $j \in \ch(i)$ in such a way that for each $q \in [|\calP|]$ and each $c \in [|\ch(i)|]$, $P_q$ is assigned to $\desc(j_c)$ only if the prefixes $P_1,P_2,\ldots,P_{q-1}$ are assigned to the subtrees $\desc(j_1),\desc(j_2),\ldots, \desc(j_{c-1})$. 
To this end, we give a dynamic programming for a respective problem; we denote by $P_0$ the empty activity set.

%subproblem %%%%%%%%%%%%%%%%
\vspace{0.2cm}
%\noindent
{\em Computation of $\PS[i,B,B^{\prime},(a,k),t]$}: 
First, for each $c \in [|\ch(i)|]$ and $q \in [0,|\calP|]$, we will check whether the activity sets $P_0,P_1,\ldots,P_{q}$ can be assigned to the subtrees rooted at $j_1,j_2,\ldots,j_c$, and exactly $\ell$ players can be assigned to the activity $a$; we refer to this subproblem by $TP[j_c,P_q,\ell]$. We initialize $TP[j_1,P_q,\ell]$ to {\em true} if one of the following statements holds:
\begin{itemize}
\item the empty activity set $P_0$ can be allocated to the first subtree $\desc(j_1)$, i.e., $q=0$, $\ell=0$, and $\PS[j_1,B,\emptyset,(a_{\emptyset},1),1]$ is {\em true}; 
\item only the activity $a$ can be assigned to $\ell$ players in $\desc(j_1)$, i.e., $q=0$, $\ell\geq 1$, and $\PS[j_1,B,\{a\},(a,k),\ell]$ is {\em true};
\item only the activity set $P_1$ can be assigned to players in $\desc(j_1)$, i.e., $q=1$, $\ell=0$, and there exists an alternative $(b,x)\in P_1 
\times [n]\cup \{(a_{\emptyset},1)\}$ such that $\DS[j_1,B,P_1,(b,x),x]$ is {\em true}, or $\WS[j_1,B,P_1,(b,x),x]$ is {\em true} and ($b= a_{\emptyset}$ or $i$ weakly prefers $(a,k)$ to $(b,x+1)$);
\item $P_1$ can be assigned to players in $\desc(j_1)$ while activity $a$ can be assigned to $\ell$ players from $\desc(j_1)$, i.e., $q=1$, $\ell\geq 1$, and $\PS[j_1,B,P_1 \cup \{a\},(a,k),\ell]$ is {\em true}.
\end{itemize}
We set $TP[j_1,P_q,\ell]$ to {\em false} otherwise. Then, we iterate through $j_1,j_2,\ldots,j_{|\ch(i)|}$ and $P_0,P_1,\ldots,P_{|\calP|}$, and update $T[j_c,P_q,\ell]$: for each $c \in [|\ch(i)|]$, for each $q \in [|\calP|]$, and for each $\ell \in [0,t]$, we set $T[j_c,P_q,\ell]$ to {\em true} if one of the following statements holds:
\begin{itemize}
\item $P_1, P_2, \ldots, P_{q}$ can be assigned to the subtrees $\desc(j_1), \desc(j_2),\ldots, \desc(j_{c-1})$ with $\ell$ players from the subtrees being assigned to the activity $a$, and the void activity can be assigned to the subtree $\desc(j_c)$, i.e., both $TP[j_{c-1},P_q,\ell]$ and $\PS[j_c,B,\emptyset,(a_{\emptyset},1),1]$ are {\em true};
\item $P_1, P_2, \ldots, P_{q}$ can be assigned to the subtrees $\desc(j_1), \desc(j_2),\ldots, \desc(j_{c-1})$ while the activity $a$ is assigned to $x$ players from the subtrees $\desc(j_1), \desc(j_2),\ldots, \desc(j_{c-1})$ and to $\ell - x$ players from the subtree $\desc(j_c)$, i.e., $\ell \geq 2$, and there exists an $x \in [\ell-1]$ such that both $TP[j_{c-1},P_{q},\ell-x]$ and $\PS[j_c,B,\{a\},(a,k),x]$ are {\em true};
\item $P_1, P_2, \ldots, P_{q-1}$ can be assigned to the subtrees $\desc(j_1), \desc(j_2),\ldots, \desc(j_{c-1})$ with $\ell$ players from the subtrees being assigned to the activity $a$, and only the activity set $P_q$ can be assigned to the subtree $\desc(j_c)$, i.e., $TP[j_{c-1},P_{q-1},\ell]$ is {\em true}, and there exists an alternative $(b,x)\in P_q \times [n]\cup \{(a_{\emptyset},1)\}$ such that $\DS[j_c,B,P_q,(b,x),x]$ is {\em true}, or $\WS[j_c,B,P_q,(b,x),x]$ is {\em true} and ($b = a_{\emptyset}$ or $i$ weakly prefers $(a,k)$ to $(b,x+1)$); 
\item $P_1, P_2, \ldots, P_{q-1}$ can be assigned to the subtrees $\desc(j_1), \desc(j_2),\ldots, \desc(j_{c-1})$, and the activity set $P_q$ can be assigned to the subtree $\desc(j_c)$ while the activity $a$ is assigned to $x$ players from the subtrees $\desc(j_1), \desc(j_2),\ldots, \desc(j_{c-1})$ and to $\ell - x$ players from the subtree $\desc(j_c)$, i.e., $\ell \geq 2$, and there exists an $x \in [\ell-1]$ such that both $TP[j_{c-1},P_{q-1},\ell-x]$ and $\PS[j_c,B,P_q \cup \{a\},(a,k),x]$ are {\em true}.
\end{itemize}
We set $TP[j_c,P_q,\ell]$ to {\em false} otherwise. 

\vspace{0.2cm}

%\noindent
{\em Computation of $\DS[i,B,B^{\prime},(a,k),t]$}: 
Second, for each $c \in [|\ch(i)|]$ and $q \in [0,|\calP|]$, we will check whether the activity sets $P_0,P_1,\ldots,P_{q}$ can be assigned to the subtrees rooted at $j_1,j_2,\ldots,j_c$, and exactly $\ell$ players can be assigned to the activity $a$, and ($a=a_{\emptyset}$ or some of the $\ell$ players strictly prefers $(a,k)$ to $(a,k+1)$); we refer to this by $TD[j_c,P_q,\ell]$.

If $i$ strictly prefers $(a,k)$ to $(a,k+1)$, then we set each value $TD[j_c,P_q,\ell]$ to $TP[j_1,P_q,\ell]$. Otherwise, we compute $TD[j_c,P_q,\ell]$ as follows.
We initialize $TD[j_1,P_q,\ell]$ to {\em true} if one of the following statements holds:
\begin{itemize}
\item only the activity $a$ can be assigned to $\ell$ players in $\desc(j_1)$, i.e., $q=0$, $\ell\geq 1$, and $\DS[j_1,B,\{a\},(a,k),\ell]$ is {\em true}; 
\item $P_1$ can be assigned to players in $\desc(j_1)$ while activity $a$ can be assigned to $\ell$ players from $\desc(j_1)$, i.e., $q=1$, $\ell\geq 1$, and $\DS[j_1,B,P_1 \cup \{a\},(a,k),\ell]$ or ) is {\em true}.
\end{itemize}
We set $TD[j_1,P_q,\ell]$ to {\em false} otherwise.  Then, for each $c \in [|\ch(i)|]$, for each $q \in [|\calP|]$, and for each $\ell \in [0,t]$, we set $TD[j_c,P_q,\ell]$ to {\em true} if one of the following statements holds:
\begin{itemize}
\item $P_1, P_2, \ldots, P_{q}$ can be assigned to the subtrees $\desc(j_1), \desc(j_2),\ldots, \desc(j_{c-1})$ with $\ell$ players from the subtrees being assigned to the activity $a$, and the void activity can be assigned to the subtree $\desc(j_c)$, i.e., both $TD[j_{c-1},P_q,\ell]$ and $\PS[j_c,B,\emptyset,(a_{\emptyset},1),1]$ are {\em true};

\item $P_1, P_2, \ldots, P_{q}$ can be assigned to the subtrees $\desc(j_1), \desc(j_2),\ldots, \desc(j_{c-1})$ while the activity $a$ is assigned to $x$ players from the subtrees $\desc(j_1), \desc(j_2),\ldots, \desc(j_{c-1})$ and to $\ell - x$ players from the subtree $\desc(j_c)$, i.e., $\ell \geq 2$, and there exists an $x \in [\ell-1]$ such that 
both $TP[j_{c-1},P_q,\ell-x]$ and $\DS[j_c,B,\{a\},(a,k),x]$ are {\em true}, or 
both $TD[j_{c-1},P_q,\ell-x]$ and $\PS[j_c,B,\{a\},(a,k),x]$ are {\em true}; 

\item $P_1, P_2, \ldots, P_{q-1}$ can be assigned to the subtrees $\desc(j_1), \desc(j_2),\ldots, \desc(j_{c-1})$ with $\ell$ players from the subtrees being assigned to the activity $a$, and only the activity set $P_q$ can be assigned to the subtree $\desc(j_c)$, i.e., 
$TD[j_{c-1},P_{q-1},\ell]$ is {\em true}, and there exists an alternative $(b,x)\in P_q 
\times [n]\cup \{(a_{\emptyset},1)\}$ such that $\DS[j_c,B,P_q,(b,x),x]$ is {\em true}, or $\WS[j_c,B,P_q,(b,x),x]$ is {\em true} and ($b = a_{\emptyset}$ or $i$ weakly prefers $(a,k)$ to $(b,x+1)$); 

\item $P_1, P_2, \ldots, P_{q-1}$ can be assigned to the subtrees $\desc(j_1), \desc(j_2),\ldots, \desc(j_{c-1})$ and $P_q$ can be assigned to the subtree $\desc(j_c)$ while the activity $a$ is assigned to $x$ players from the subtrees $\desc(j_1), \desc(j_2),\ldots, \desc(j_{c-1})$ and to $\ell - x$ players from the subtree $\desc(j_c)$, i.e., $\ell \geq 2$, and there exists an $x \in [\ell-1]$ such that 
both $TD[j_{c-1},P_{q-1},\ell-x]$ and $\PS[j_c,B,P_q \cup \{a\},(a,k),x]$ are {\em true}, or 
both $TP[j_{c-1},P_{q-1},\ell-x]$ and $\DS[j_c,B,P_q \cup \{a\},(a,k),x]$ are {\em true}.
\end{itemize}
We set $TD[j_c,P_q,\ell]$ to {\em false} otherwise. 

\vspace{0.2cm}
%\noindent
{\em Computation of $\WS[i,B,B^{\prime},(a,k),t]$}: 
Third, for each $c \in [|\ch(i)|]$ and $q \in [0,|\calP|]$, we will check whether the activity sets $P_0,P_1,\ldots,P_{q}$ can be assigned to the subtrees rooted at $j_1,j_2,\ldots,j_c$, and exactly $\ell$ players can be assigned, and no player in $\desc(j)$ has an incentive to deviate to $i$'s coalition;
we refer to this subproblem by $TW[j_c,P_q,\ell]$.
We initialize $TW[j_1,P_q,\ell]$ to {\em true} if one of the following statements holds:
\begin{itemize}
\item the empty activity set $P_0$ can be allocated to the first subtree $\desc(j_1)$, i.e., $q=0$, $\ell=0$, $\PS[j_1,B,\emptyset,(a_{\emptyset},1),1]$ is {\em true}, and ($a=a_{\emptyset}$ or $j_1$ weakly prefers $(a_{\emptyset},1)$ to $(a,k+1)$); 
\item only the activity $a$ can be assigned to $\ell$ players in $\desc(j_1)$, i.e., $q=0$, $\ell\geq 1$, and $\WS[j_1,B,\{a\},(a,k),\ell]$ is {\em true};
\item only the activity set $P_1$ can be assigned to players in $\desc(j_1)$, i.e., $q=1$, $\ell=0$, and there exists an alternative $(b,x)\in P_1 \times [n]\cup \{(a_{\emptyset},1)\}$ such that (1). $\DS[j_1,B,P_1,(b,x),x]$ is {\em true}, or ($\WS[j_1,B,P_1,(b,x),x]$ is {\em true} and ($b=a_{\emptyset}$ or $i$ weakly prefers $(a,k)$ to $(b,x+1)$)), and (2). ($a=a_{\emptyset}$ or $j_1$ weakly prefers $(b,x)$ to $(a,k+1)$);
\item $P_1$ can be assigned to players in $\desc(j_1)$ while activity $a$ can be assigned to $\ell$ players from $\desc(j_1)$, i.e., $q=1$, $\ell\geq 1$, and $\WS[j_1,B,P_1 \cup \{a\},(a,k),\ell]$ is {\em true}.
\end{itemize}
We set $TW[j_1,P_q,\ell]$ to {\em false} otherwise.  
Then, for each $c \in [|\ch(i)|]$, for each $q \in [|\calP|]$, and for each $\ell \in [0,t]$, we set $TW[j_c,P_q,\ell]$ to {\em true} if one of the following statements holds:
\begin{itemize}
\item  $P_1, P_2, \ldots, P_{q}$ can be assigned to the subtrees $\desc(j_1), \desc(j_2),\ldots, \desc(j_{c-1})$ with $\ell$ players from the subtrees being assigned to the activity $a$, and the void activity can be assigned to the subtree $\desc(j_c)$, i.e., both $TW[j_{c-1},P_q,\ell]$ and $\PS[j_c,B,\emptyset,(a_{\emptyset},1),1]$ are {\em true}, 
and $a=a_{\emptyset}$ or player $j_c$ weakly prefers $(a_{\emptyset},1)$ to $(a,k+1)$;

\item $P_1, P_2, \ldots, P_{q}$ can be assigned to the subtrees $\desc(j_1), \desc(j_2),\ldots, \desc(j_{c-1})$ while the activity $a$ is assigned to $x$ players from the subtrees $\desc(j_1), \desc(j_2),\ldots, \desc(j_{c-1})$ and to $\ell - x$ players from the subtree $\desc(j_c)$, i.e., $\ell \geq 2$, and there exists an $x \in [\ell-1]$ such that both $TW[j_{c-1},P_q,\ell-x]$ and $\WS[j_c,B,\{a\},(a,k),x]$ are {\em true}; 

\item $P_1, P_2, \ldots, P_{q-1}$ can be assigned to the subtrees $\desc(j_1), \desc(j_2),\ldots, \desc(j_{c-1})$ with $\ell$ players from the subtrees being assigned to the activity $a$, and $P_q$ can be assigned to the subtree $\desc(j_c)$ with no player from the subtree being assigned to $a$, i.e., $TW[j_{c-1},P_{q-1},\ell]$ is {\em true}, and there exists an alternative $(b,x)\in P_q \times [n]\cup \{(a_{\emptyset},1)\}$ such that (1). $\DS[j_c,B,P_q,(b,x),x]$ is {\em true}, or ($\WS[j_c,B,P_q,(b,x),x]$ is {\em true} and ($b=a_{\emptyset}$ or $i$ weakly prefers $(a,k)$ to $(b,x+1)$)), and (2). $a=a_{\emptyset}$ or $j_c$ weakly prefers $(b,x)$ to $(a,k+1)$;

\item $P_1, P_2, \ldots, P_{q-1}$ can be assigned to the subtrees $\desc(j_1), \desc(j_2),\ldots, \desc(j_{c-1})$ and $P_q$ can be assigned to the subtree $\desc(j_c)$ while the activity $a$ is assigned to $x$ players from the subtrees $\desc(j_1), \desc(j_2),\ldots, \desc(j_{c-1})$ and to $\ell - x$ players from the subtree $\desc(j_c)$, i.e., $\ell \geq 2$, and there exists an $x \in [\ell-1]$ such that both $TW[j_{c-1},P_{q-1},\ell-x]$ and $\WS[j_c,B,P_q \cup \{a\},(a,k),x]$ are {\em true}. 
\end{itemize}
We set $TW[j_c,P_q,\ell]$ to {\em false} otherwise. 

A desired assignment exists if and only $TP[j_c,P_q,t-1]$ (respectively, $TD[j_c,P_q,t-1]$ and $TW[j_c,P_q,t-1]$) is {\em true}. Clearly, the size of each dynamic programming table is at most $pn^2$ and each entry can be filled in polynomial time. 
\end{proof}

\subsection*{Proof of Theorem \ref{thm:W1:activities:clique:IS}}
%Here, we prove the part of Theorem~\ref{thm:W1:activities:clique} about individual stability.
\begin{proof}
We provide a parameterized reduction from {\sc Clique} on regular graphs. 
Given a regular graph $G=(V,E)$ and a positive integer $k$, where $|V|=n$, $|E|=m$, and each vertex of $G$ has degree $\delta \geq k-1$, 
we create an instance of \gGASP\ whose underlying graph is a clique, as follows. 

%Vertex and Edge Activities
We create three {\em vertex activities} $a^{(1)}_i$, $a^{(2)}_i$, and $a^{(3)}_i$ for each $i \in [k]$, one {\em edge activity} $b_j$ for each $j \in [k(k-1)/2]$, and four other activities $d$, $x$, $y$, and $z$. 
%Vertex and Edge Players
For each $v \in V$, we create one {\em vertex player} $v$, and for each edge $e=\{u,v\} \in E$, we create two {\em edge players} 
$e_{uv}$ and $e_{vu}$.

\emph{Idea.}
We will create $k$ empty IS-instances $N_i=\{p^{(i)}_1,p^{(i)}_2,p^{(i)}_3\}$ for each $i \in [k]$, and another empty IS-instance $N_g=\{g_1,g_2,g_3\}$ together with the stabilizer $g$; in order to stabilize these gadgets, each activity $a^{(1)}_i$ should be assigned to some vertex player, and the stabilizer $g$ needs to form a coalition with the players in $N_g$. Further, we will create dummies of vertex and edge players in such a way that a stable assignment has the following properties:
\begin{enumerate}
\item if a vertex player $v$ is assigned to a vertex activity, it forms a coalition of size $\alpha(v)$ that consists of its dummies and $\delta-k+1$ edge players incident to the vertex; and
\item if an edge player $e_{vu}$ is assigned to an edge activity, it forms a coalition of size $\beta(e)$ that consists of the edge player $e_{uv}$ and the dummies of $\{u,v\}$.
\end{enumerate}

\emph{Construction details.}
%Vertex and Edge dummy players
Now, let $P=\{j(k+3)+n\mid j\in[n]\}$, and let $\alpha:V\to P$ be a bijection that assigns a 
distinct number in $P$ to each vertex $v\in V$. Note that $u\neq v$
implies that the intervals $[\alpha(u),\alpha(u)+k+1]$ and 
$[\alpha(v)-1,\alpha(v)+k]$ are disjoint.
Similarly, let $Q=\{2j\mid j\in[m]\}$ and 
let $\beta:E \rightarrow Q$ be a bijection that assigns a distinct number
in $Q$ to each edge $e\in E$.
For each $v\in V$ we construct a set $\Dummy(v)$ of $\alpha(v)-\delta+k-2$ {\em dummy vertex players}. 
Similarly, for each $e\in E$ we construct a set $\Dummy(e)$ of $\beta(e)-2$ {\em dummy edge players}.

%Preferences for vertex and edge players
We will now define the players' preferences.
\begin{itemize}
\item Each vertex player $v \in V$ and the players in $\Dummy(v)$ approve each 
alternative $(a^{(1)}_i,s)$ where $i\in [k]$ and $s \in [\alpha(v),\alpha(v)+k]$.
\item Each edge player $e_{vu}$ approves each alternative $(a^{(1)}_i,s)$ where $i\in [k]$ and $s \in [\alpha(v),\alpha(v)+k]$ as well as each alternative
in $(b_j,\beta(e))$ where $j \in [k(k-1)/2]$.
\item  The dummies in $\Dummy(e)$ only approve the alternatives in $(b_j,\beta(e))$ where $j \in [k(k-1)/2]$.
\end{itemize}
All of these players are indifferent among all alternatives they approve. The stabilizer $g$ approves each alternative of the form $(a^{(1)}_i, s)$ with $i \in [k]$, $s\in \bigcup_{v \in V}[\alpha(v)+2,\alpha(v)+k]$ and is indifferent among them; 
she also approves $(d,4)$, but likes it less than all other approved alternatives.

%Preferences for empty IS instances
For players in $N_g$, we have
\begin{align*}
g_1:&~ (d,4) \succ (y,2) \succ (z,1) \succ  (x,3)  \succ (x,2) \succ (x,1) \succ (a_{\emptyset},1),\\
g_2:&~ (d,4) \succ (x,3) \succ (x,2) \succ (z,2) \succ (y,2) \succ (y,1) \succ (a_{\emptyset},1),\\
g_3:&~ (d,4) \succ (x,3) \succ (z,2) \succ (z,1) \succ (a_{\emptyset},1).
\end{align*}
Notice that their preference over the alternatives of $x$, $y$, and $z$ is cyclic; hence, in an individually stable assignment, the activity $d$ should be assigned to all the players in $N_g$ together with its stablizer $g$; otherwise $\pi$ cannot be stable as we have seen in Example \ref{ex:IS:empty}. Similarly, the preference for players in $N_{i}$ is cyclic and given by
\begin{align*}
p^{(i)}_1:&~ (a^{(2)}_i,2) \succ (a^{(1)}_i,1) \succ (a^{(3)}_i,3) \succ (a^{(3)}_i,2) \succ (a^{(3)}_i,1) \succ (a_{\emptyset},1),\\
p^{(i)}_2:&~(a^{(3)}_i,3) \succ (a^{(3)}_i,2) \succ (a^{(1)}_i,2) \succ (a^{(2)}_i,2) \succ (a^{(2)}_i,1) \succ (a_{\emptyset},1),\\
p^{(i)}_3:&~(a^{(3)}_i,3) \succ (a^{(1)}_i,2) \succ (a^{(1)}_i,1) \succ (a_{\emptyset},1).
\end{align*}
Again, in an individually stable assignment, at least one of the three activities $a^{(1)}_i, a^{(2)}_i$, $a^{(3)}_i$ must be assigned outside of $N_i$; the only individual rational way to do this is to assign the activity $a^{(1)}_i$ to other players.

Finally, we take the underlying social network to be a complete graph.
Note that the number of activities depends on $k$, but not on $n$, and the size of our instance of \gGASP\ is bounded by $O(n^2+m^2)$.

\emph{Correctness.}
We will now argue that the graph $G$ contains a clique of size $k$ if and only if there exists an individually stable assignment
for our instance of \gGASP. 
Suppose that $G$ contains a clique $S$ of size $k$. We construct an assignment $\pi$ as follows. 
We establish a bijection $\eta$ between $S$ and $[k]$, and for each $v\in V$
we form a coalition of size $\alpha(v)$ that engages in $a_{\eta(v)}$: this coalition consists
of $v$, all players in $\Dummy(v)$, and all edge players $e_{vu}$ such that $u\not\in S$. 
Also, we establish a bijection $\xi$ between the edge set $\{\,e=\{u,v\}\in E\mid u,v \in S\,\}$ 
and $[k(k-1)/2]$, and assign the activity $b_{\xi(e)}$ 
to the edge players $e_{uv}$, $e_{vu}$ as well as to all players in $\Dummy(e)$. 
Finally, we set $\pi(p^{(i)}_1)=\pi(p^{(i)}_2)=\pi(p^{(i)}_3)=a^{(3)}_i$ for each $i \in [k]$ and $\pi(g)=\pi(g_1)=\pi(g_2)=\pi(g_3)=d$, 
and assign the void activity to the remaining players.
We will now argue that the resulting assignment $\pi$ is individually stable. 

%assigned vertex player
Clearly, no player assigned to an activity $a^{(1)}_i$ or $b_j$
wishes to deviate.
Now, consider a player $v\in N$ with $\pi(v)=a_\emptyset$;
by construction, $v$ only wants to join a coalition if it engages in an activity $a^{(1)}_i$
and its size is in the interval $[\alpha(v)-1, \alpha(v)+k]$, and no such coalition
exists. The same argument applies to players in $\Dummy(v)$.
Similarly, consider an edge player $e_{vu}$ with $\pi(e_{vu})=a_\emptyset$.
We have $v\not\in S$, and therefore $e_{vu}$ does not want to join
any of the existing coalitions; the same argument applies to all dummies of $e=\{u,v\}$.
Further, players in $N_i$ do not want to deviate
since players $p^{(i)}_2$ and $p^{(i)}_3$ are allocated to their top alternative and the activity $a^{(1)}_i$ is assigned to at least one player.
Also, players in $N_g$ do not want to deviate since they are allocated one of their top choices. 
Finally, the stabilizer $g$ does not want to deviate, since there is no coalition of size 
$s\in \bigcup_{v\in V}[\alpha(v)+1,\alpha(v)+k-1]$ that engages in an activity $a^{(1)}_i$. 
Hence, $\pi$ is individually stable.

Conversely, suppose that there exists an individually stable feasible assignment $\pi$.
Notice that $\pi$ cannot allocate an activity $a^{(1)}_i$ to players in $N_i$, 
or leave it unallocated, since no such assignment can be individually stable.
Thus, each vertex activity $a^{(1)}_i$ is allocated to a coalition 
whose size lies in the interval $[\alpha(v),\alpha(v)+k]$ 
for some $v \in V$.
%%%%%
Further, individual stability implies that $\pi$ allocates the activity $d$ to the players in $N_g$ and its stabilizer $g$.
Now, if some vertex activity $a^{(1)}_i$ is assigned to $s$ players,
where $s\in[\alpha(v)+1, \alpha(v)+k-1]$ for some $v \in V$, 
the stabilizer $g$ would then deviate to that coalition. Further, for each $v \in V$, the number of players other than $g$ who approve the alternatives of the form $(a^{(1)}_i,\alpha(v)+k)$ is at most $\alpha(v)+k-1$ but $g$ is assigned to $d$; thus, for each $a^{(1)}_i$ we have $|\pi^{a^{(1)}_i}|=\alpha(v)$ for some $v\in V$. 
Consider a player $v\in S$, and let $a(v)$ be the activity assigned to $\alpha(v)$
players under $\pi$. We have $\pi(v)= a(v)$, since otherwise $v$ has an IS-deviation to $a(v)$.

Now, let $S=\{v\in V\mid \pi(v) =a^{(1)}_i\text{ for some }i \in [k]\}$.
By construction, $|S|=k$. We can show that $S$ is a clique in a similar manner to the proof of Claim \ref{claim:clique}. 
\end{proof}

\subsection*{Proof of Theorem \ref{thm:NP:clique:core}}
\begin{proof}
Our problem is in NP by Proposition~\ref{prop:in-core}. We again reduce from X3C. Let $(V,\calS)$ be an instance of X3C where $V=\{v_1,v_2,\ldots,v_{3k}\}$ and $\calS=\{S_1,S_2,\ldots,S_{m}\}$. 

We define the set of activities by $A^{*}=\{a,b,c_1,c_2\}$. We introduce three players $x_1,x_2,x_3$ and one player $S_j$ for each $S_j \in \calS$. 

\emph{Idea.}
Again, we will create the dummies of each $v_i \in V$ and the preferences of players as follows:
\begin{itemize}
\item The players $x_1$, $x_2$, and $x_3$ form an empty core instance. In a stable assignment, they need to be assigned to the activity $a$. 
\item If the dummies of $v_i$ are assigned to the void activity and none of the players corresponding to the sets including $v_i$ (denoted by $\calS(v_i)$) is assigned to the activity $c_2$, then these players will form a blocking coalition of size $\beta(v_i)$ and deviate to activity $b$.
\end{itemize}

\emph{Construction details.}
For each $v_i\in V$, we let $\calS(v_i)=\{\, S_j \mid v_i \in S_j \in \calS \,\}$ and $\beta(v_i)=i+3k+1$ 
and create a set $\Dummy(v_i)=\{d^{(1)}_i,d^{(2)}_i,\ldots,d^{(\beta(v_i)-|\calS(v_i)|)}_i\}$ of dummy players. For each $i\in[3k]$
the number $\beta(v_i)$ is the target coalition size when all players 
in $\calS(v_i)$ are engaged in activity $b$, together with the players in $\Dummy(v_i)$. We then create an edge between any pair of players.

%preferences
The agents' preferences over alternatives are defined as follows. 
For each $S_j \in \calS$, we let $B_j=\{b\}\times \{\, \beta(v_i)\mid v_i \in S_j \,\}$. The preferences of each player $S_j \in \calS$ are given by
\begin{align*}
&S_j:~(c_2,k)\succ B_j \succ (c_1,n-k) \succ (a_{\emptyset},1).
\end{align*}
For each $v_i \in V$ the dummy players in $\Dummy(v_i)$ only approve the alternative $(b,\beta(v_i))$.

Finally, the 
preferences of players $x_1$, $x_2$, and $x_3$ are given by
\begin{align*}
&x_{1}:~(c_1,2) \sim (c_2,2) \succ (a,3) \succ (a_{\emptyset},1)\\
&x_{2}:~(a,2) \succ (c_1,2) \sim (c_2,2) \succ (a,3) \succ (a_{\emptyset},1)\\
&x_{3}:~(a,3) \succ (c_1,1) \sim (c_2,1) \succ (a,2) \succ (a_{\emptyset},1).
\end{align*}
Notice that the preferences of $x_1$, $x_2$, and $x_3$, when restricted to $\{a,c_1,a_{\emptyset}\}\times [1,2,3]$ or $\{a,c_2,a_{\emptyset}\}\times [1,2,3]$, form an empty core instance (Example \ref{ex:core:empty}).

We will show that $(V,\calS)$ contains an exact cover if and only if there exists a core stable feasible assignment.

\emph{Correctness.}
Suppose that $(V,\calS)$ admits an exact cover $\calS^{\prime}$. Then, we construct a feasible assignment $\pi$ by setting $\pi(x_i)=a$ for $i=1,2,3$, $\pi(S_j)=c_2$ for $S_j \in \calS^{\prime}$, $\pi(S_j)=c_1$ for $S_j \in \calS\setminus \calS^{\prime}$, and assigning the remaining players to the void activity.  
Clearly, no subset together with $c_i$ $(i=1,2)$ or $a$ strongly blocks $\pi$. 
We will show that no subset $T$ together with activity $b$ strongly blocks $\pi$. Suppose towards a contradiction that such a subset $T$ exists; 
as no players in $\{x_1,x_2,x_3\}$ approves alterntives of $b$, it must be the case 
that $|T|=\beta(v_i)$ for some $v_i\in V$ and hence $T$ consists of agents who approve $(b, \beta(v_i))$, i.e., $T=\calS(v_i) \cup \Dummy(v_i)$ for some $v_i\in V$. However, since $\calS^{\prime}$ is an exact cover, 
there is an agent $S_j\in \calS^{\prime}\cap \calS(v_i)$ with $\pi(S_j)=c_2$, and this agent prefers $(c_2, k)$ to $(b, \beta(v_i))$, a contradiction.
Hence, $\pi$ is core stable.

Conversely, suppose that there exists a core stable feasible assignment $\pi$ and let $\calS^{\prime}=\{\, S_j \in \calS \mid \pi(S_j)=c_2 \,\}$. We will show that $\calS^{\prime}$ is an exact cover.  Observe that by individual rationality, the only agents who can be allocated to $a$ are the players $x_1$, $x_2$, and $x_3$. Hence, by core stability, both activities $c_1$ and $c_2$ must be allocated outside of the players $x_1$, $x_2$, and $x_3$; otherwise, no core stable outcome would exist as we have seen in Example \ref{ex:core:empty}. The only indivudally rational way to do this is to assign activity $c_2$ to $k$ players from $\calS$, and assign activity $c_1$ to the remaining $n-k$ players in $\calS$. Thus, $|\calS^{\prime}|=k$. Then, no player in $\calS$ can be engaged in activity $b$, and hence, by individual rationality, all dummy players must be assigned to the void activity.
Now suppose towards a contradiction that $\calS^{\prime}$ is not a cover, i.e., there exists an element $v_i \in V$ such that $\calS(v_i) \cap \calS^{\prime}=\emptyset$. This would mean that $\pi$ assigns all players in $\calS(v_i)$ to the activity $c_1$, and hence the coalition $\calS(v_i) \cup \Dummy(v_i)$ together with the activity $b$ strongly blocks $\pi$, contradicting the stability of $\pi$.
\end{proof}

\section*{Appendix. Proofs omitted from Section 6}
\subsection*{Proof of Theorem \ref{thm:W1players:clique:NSIS}}
\begin{proof}
We describe a parameterized reduction from the W[1]-hard \MCC{} problem.
Given an undirected graph~$G=(V,E)$, a positive integer~$h\in \mathbb{N}$,
and a vertex coloring~$\phi \colon V\to [h]$, \MCC{} asks whether~$G$ admits a colorful $h$-clique, that is,
a size-$h$ vertex subset~$Q\subseteq V$ such that the vertices in $Q$ are pairwise adjacent
and have pairwise distinct colors.
Without loss of generality, we assume that there are exactly $q$ vertices of each color for some $q\in\mathbb N$, 
and that there are no edges between vertices of the same color.

Let $(G,h,\phi)$ be an instance of \MCC{} with $G=(V, E)$.
For convenience,
we write $V^{(i)} = \{v_1^{(i)}, \ldots, v_q^{(i)}\}$ to denote the set of vertices of color~$i$
for every $i \in [h]$. 
We construct our \gGASP\ instance as follows.
We have one \emph{vertex activity}~$v$ for each vertex~$v \in V$, 
one \emph{edge activity}~$e$ for each edge~$e \in E$,
three \emph{color activities}~$a_i$, $b_i$, and $c_i$ for each color~$i \in [h]$, 
three \emph{colorpair activities}~$a_{\{ij\}}$, $b_{\{ij\}}$, and $c_{\{ij\}}$ for each color pair~$\{i,j\} \subseteq [h], i \neq j$, 
and four other additional activities $d$, $x$, $y$, and $z$.

\emph{Idea.}
We will have one \emph{color gadget} $\colour(i)$ for each color~$i \in [h]$, one \emph{colorpair gadget} $\pair(\{i,j\})$ for each color pair~$\{i,j\},i \neq j$, and one empty IS-instance $N_g=\{g_1,g_2,g_3\}$ together with the stabilizer $g$.
For most of the possible assignments, these gadgets will be unstable, unless the following holds:
\begin{enumerate}
 \item For each color~$i \in [h]$ the first two players from the color gadget select
       a vertex of color~$i$ (by being assigned together to the corresponding vertex activity).
 \item For each color pair~$\{i,j\} \subseteq [h], i \neq j$ the first two players of the colorpair gadget select 
       an edge connecting one vertex of color~$i$ and one vertex of color~$j$ (by being assigned together to
       the corresponding edge activity).
 \item Every selected edge for the color pair~$\{i,j\}$ must be incident to both vertices selected for color~$i$ and color~$j$. 
 \item The stabilizer $g$ as well as the players in $N_g$ are engaged in the activity $d$.
\end{enumerate}
If the four conditions above hold, then the assignment must encode a colorful $h$-clique.

\emph{Construction details.}
The color gadget~$\colour(i)$, $i\in [h]$ consists of the following three players.
 \begin{align*}
p_1^{(i)}:&~ (v_1^{(i)},2) \succ E(v^{(i)}_1) \times [3,5]\succ (v_2^{(i)},2) \succ E(v^{(i)}_2) \times [3,5] \succ \dots \succ (v_q^{(i)},2)  \succ E(v^{(i)}_q) \times [3,5] \succ 
\\ &~ (b_i,2) \succ (a_i,1) \succ (c_i,3) \succ (c_i,2) \succ (c_i,1) \succ (a_{\emptyset},1),\\
p_2^{(i)}:&~ (v_q^{(i)},2)\succ E(v^{(i)}_q) \times [3,5]\succ (v_{q-1}^{(i)},2) \succ E(v^{(i)}_{q-1})\times [3,5] \succ \dots \succ (v_1^{(i)},2) \succ E(v^{(i)}_{1}) \times [3,5] \succ 
\\ &~(c_i,3) \succ (c_i,2) \succ (a_i,2) \succ (b_i,2) \succ (b_i,1) \succ (a_{\emptyset},1),\\
p_3^{(i)}:&~ (c_i,3) \succ (a_i,2) \succ (a_i,1) \succ (a_{\emptyset},1),
\end{align*}
where $E(v)$ to denote the set of activities corresponding to edges incident to vertex $v$ for every $v \in V$.

The colorpair gadget~$\pair(\{i,j\})$, $\{i,j\} \subseteq [h], i \neq j$, consists of three players.
\begin{align*}
p_1^{\{i,j\}}:&~E(\{i,j\}) \times [2,5] \succ (b_{\{ij\}},2) \succ (a_{\{ij\}},1) \succ (c_{\{ij\}},3) \succ (c_{\{ij\}},2) \succ (c_{\{ij\}},1) \succ (a_{\emptyset},1)\\
p_2^{\{i,j\}}:&~E(\{i,j\}) \times [2,5] \succ (c_{\{ij\}},3) \succ (c_{\{ij\}},2) \succ (a_{\{ij\}},2) \succ (b_{\{ij\}},2) \succ (c_{\{ij\}},1) \succ (a_{\emptyset},1),\text{ and}\\
p_3^{\{i,j\}}:&~(c_{\{ij\}},3) \succ (a_{\{ij\}},2) \succ (a_{\{ij\}},1) \succ (a_{\emptyset},1)
\end{align*}
where $E(\{i,j\})$ denotes the set of edge activities incident to vertices of color $i$ and $j$.

There is a stabilizer player $g$ with the following preferences.
\begin{align*}
g:&~ E \times [4,5] \succ (d,4) \succ  (a_{\emptyset},1).
\intertext{The set $N_g$ consists of three players.}
g_1:&~ (d,4) \succ (y,2) \succ (z,1) \succ  (x,3)  \succ (x,2) \succ (x,1) \succ (a_{\emptyset},1),\\
g_2:&~ (d,4) \succ (x,3) \succ (x,2) \succ (z,2) \succ (y,2) \succ (y,1) \succ (a_{\emptyset},1),\\
g_3:&~ (d,4) \succ (x,3) \succ (z,2) \succ (z,1) \succ (a_{\emptyset},1).
\end{align*}

We take the underlying social network to be a complete graph. 
Together, there are $3h + 3\binom h2 + 4$ players, namely
\[  N = \bigcup_{i\in [h]} \colour(i) \cup \bigcup_{i\neq j\in [h]} \pair(\{i,j\}) \cup \{g\} \cup N_g.  \]
Note that the number of players depends on $h$, but not on $n$, and the size of our instance of gGASP is bounded by $O(n^2 + m^2)$.

\emph{Correctness.}
We will now argue that the graph $G$ admits a colorful clique of size $h$ 
if and only if our instance of \gGASP\ admits a Nash stable feasible assignment or an individually stable feasible assignment.

Suppose that there exists a colorful clique~$H$ of size $h$. We construct a Nash stable assignment $\pi$ where 
\begin{itemize}
\item the first two players $p_1^{(i)}$ and $p_2^{(i)}$ of the color gadget~$\colour(i)$ are assigned to the activity corresponding to the vertex of color~$i$ from~$H$,
\item the last player $p_3^{(i)}$ of the color gadget~$\colour(i)$ is assigned to the activity $a_i$,
\item the first two players $p_1^{\{i,j\}}$ and $p_2^{\{i,j\}}$ of the colorpair gadget~$\pair(\{i,j\})$ are assigned to the activity corresponding
to the edge between the vertices of color~$i$ and~$j$ in~$H$,
\item the last player $p_3^{\{i,j\}}$ of the colorpair gadget~$\pair(\{i,j\})$ is assigned to the activity $a_{\{ij\}}$, and
\item the stabilizer $g$ and the players in $N_g$ are assigned to the activity $d$.
\end{itemize} 
Observe that for a successful Nash deviation a player must join an existing non-empty coalition, because no player
prefers a size-one activity to the currently assigned one.
By construction, the last player of each color gadget and the last player of the colorpair gadget cannot deviate
(no other players engage in an approved activity).
Consider the first two players of a color gadget~$\colour(i)$.
They cannot deviate to a vertex activity, because their current activity is their only approved vertex activity 
that has some players assigned to it.
They cannot deviate to an edge activity either, because they would only prefer edge activities corresponding to
edges that are not incident to the vertex of color~$i$; these activities, however, 
have no players assigned to them.
The first two players of each colorpair gadget~$\pair(\{i,j\})$ and the players in $N_g$
do not wish to deviate since they are each assigned to their top alternative.
Finally, the stabilizer $g$ does not wish to deviate since there is no coalition of size $3$ or $4$ assigned to an edge activity.
Thus, we have a Nash stable feasible assignment, and hence, an individually stable assignment.

Conversely, suppose that there exists an individually stable feasible assignment~$\pi$. 
Then, $\pi$ must assign the activity $d$ to the stabilizer as well as the players in $N_g$; otherwise, $\pi$ cannot be stable as we have seen in Example~\ref{ex:IS:empty}. 
Likewise, the first two players of each colorpair gadget~$\pair(\{i,j\})$ must be assigned to the same edge activity
corresponding to some edge~$\hat{e}_{i,j} \in E$. 
We say that these players ``select edge~$\hat{e}_{i,j}$''. 

Now suppose towards a contradiction that the first player $p_1^{(i)}$ of color gadget~$\colour(i)$ is assigned to an edge activity $\hat{e}_{i,j}$ incident to vertices of color $i$ and $j$. Then, such a coalition has size in $[3,5]$ by individual rationality, and must not contain the stabilizer $g$, because $g$ has to be engaged in activity $d$ as discussed before. Hence the player $p_1^{(i)}$ can only be assigned to $\hat{e}_{i,j}$ together with $p_2^{(i)}$ and the first two players from the colorpair gadget $\pair(\{i,j\})$, which results in a coalition of size $3$ or $4$. This would, however, cause an IS-deviation by the stabilizer $g$ to the edge activity $\hat{e}_{i,j}$, a contradiction. The same argument applies to when the second player $p_2^{(i)}$ is assigned to an edge activity. Therefore, the first two players of each color gadget~$\colour(i)$ are assigned to the same vertex activity corresponding to some vertex~$v_{\ell_i} \in V^{(i)}$. We say that these players ``select vertex~$v_{\ell_i}$''.
(Note that at this point, the coalition structure of any stable outcome is already fixed:
The first two players of each color gadget and of each colorpair gadget
must form a coalition of size two, respectively, and
all other players must be form singleton coalitions.)

Now, assume towards a contradiction that the selected vertices and edges do not form a colorful clique of size~$h$.
The size and colorfulness are clear from the construction.
Hence, there must be some pair~$\{v_{\ell_i},v_{\ell_j}\}$ of selected vertices that are not adjacent.
However, this would imply that colorpair gadget~$\pair(\{i,j\})$ selected an edge 
that is not incident to vertex~$v_{\ell_i}$
or not incident to vertex~$v_{\ell_j}$.
Without loss of generality let it be non-incident to vertex~$v_{\ell_i}$.
Now, there are two cases.
First, if $\hat{e}_{i,j} \in E(v_{x}^{(i)})$ with $x<\ell_i$, then
player~$p_1^{(i)}$ would have an IS-deviation to the activity corresponding to~$\hat{e}_{i,j}$.
Second, if $\hat{e}_{i,j} \in E(v_{x}^{(i)})$ with $x>\ell_i$,
then  player~$p_2^{(i)}$ would have an IS-deviation to the activity corresponding to~$\hat{e}_{i,j}$.
In both cases we have a contradiction to the assumption that~$\pi$ is individually stable. A similar argument applies to the case when there is a Nash stable assignment.
\end{proof}

\end{document}